\documentclass{siamonline1116}

\usepackage{umoline}
\usepackage{datetime}
\usepackage[margin=1in]{geometry}
\usepackage{setspace}
\usepackage{bbm}
\usepackage{bm}
\usepackage{xspace}

\usepackage{enumerate}
\usepackage{epstopdf}
\usepackage[caption=false]{subfig}
\usepackage{subfig}
\usepackage{amssymb}
\usepackage{array}
\usepackage{multirow}

\usepackage{cleveref}
\crefname{equation}{}{}

\usepackage{cite}

\newtheorem{rem}{Remark}[section]

\newcommand{\fo}{\begin{eqnarray*}}
\newcommand{\mel}{\end{eqnarray*}}

\newcommand{\cz}{{\mathbb C}}

\newcommand{\qd}{\end{proof}\vspace{0.5ex}}

\newcommand\tnorm[1]{\left\vert\xspace\left\vert\xspace\left\vert\mskip2mu
#1\mskip2mu \right\vert\xspace\right\vert\xspace\right\vert}

\makeatletter
\newcommand{\opnorm}{\@ifstar\@opnorms\@opnorm}
\newcommand{\@opnorms}[1]{%
  \left|\mkern-1.5mu\left|\mkern-1.5mu\left|
   #1
  \right|\mkern-1.5mu\right|\mkern-1.5mu\right|
}
\newcommand{\@opnorm}[2][]{%
  \mathopen{#1|\mkern-1.5mu#1|\mkern-1.5mu#1|}
  #2
  \mathclose{#1|\mkern-1.5mu#1|\mkern-1.5mu#1|}
}
\makeatother

\newcommand{\norm}[1]{\Vert#1\Vert}

\newcommand{\transpose}{\top}

\newcommand{\conv}{\circledast}
\newcommand{\cconv}{\circledast}

\newcommand{\R}{\mathbb{R}}

\newcommand{\vct}[1]{\boldsymbol{#1}}
\newcommand{\mtx}[1]{\boldsymbol{#1}}

\DeclareMathOperator*{\minimize}{\text{minimize}}

\newcommand{\va}{\vct{a}}

\newcommand{\ve}{\vct{e}}
\newcommand{\vf}{\vct{f}}
\newcommand{\vg}{\vct{g}}
\newcommand{\vh}{\vct{h}}

\newcommand{\vq}{\vct{q}}

\newcommand{\vs}{\vct{s}}

\newcommand{\vu}{\vct{u}}
\newcommand{\vv}{\vct{v}}
\newcommand{\vw}{\vct{w}}
\newcommand{\vx}{\vct{x}}
\newcommand{\vy}{\vct{y}}
\newcommand{\vz}{\vct{z}}

\newcommand{\vzeta}{\vct{\zeta}}

\newcommand{\vxi}{\vct{\xi}}

\newcommand{\vphi}{\vct{\phi}}

\newcommand{\vpsi}{\vct{\psi}}

\newcommand{\vzero}{\vct{0}}

\newcommand{\mA}{\mtx{A}}
\newcommand{\mB}{\mtx{B}}
\newcommand{\mC}{\mtx{C}}

\newcommand{\mE}{\mtx{E}}
\newcommand{\mF}{\mtx{F}}

\newcommand{\mL}{\mtx{L}}
\newcommand{\mM}{\mtx{M}}

\newcommand{\mP}{\mtx{P}}
\newcommand{\mQ}{\mtx{Q}}
\newcommand{\mR}{\mtx{R}}
\newcommand{\mS}{\mtx{S}}
\newcommand{\mT}{\mtx{T}}
\newcommand{\mU}{\mtx{U}}

\newcommand{\mW}{\mtx{W}}

\newcommand{\mY}{\mtx{Y}}
\newcommand{\mZ}{\mtx{Z}}

\newcommand{\mLambda}{\mtx{\Lambda}}

\newcommand{\mUpsilon}{\mtx{\Upsilon}}
\newcommand{\mPhi}{\mtx{\Phi}}
\newcommand{\mPsi}{\mtx{\Psi}}

\newcommand{\mId}{{\bf I}}

\newcommand{\mzero}{{\bf 0}}

\usepackage{lipsum}
\usepackage{amsfonts}
\usepackage{graphicx}
\usepackage{epstopdf}
\usepackage{algorithmic}
\ifpdf
  \DeclareGraphicsExtensions{.eps,.pdf,.png,.jpg}
\else
  \DeclareGraphicsExtensions{.eps}
\fi

\numberwithin{theorem}{section}

\newcommand{\TheTitle}{Spectral Methods for Passive Imaging: Non-asymptotic Performance and Robustness}
\newcommand{\TheAuthors}{Kiryung Lee, Felix Krahmer, and Justin Romberg}

\headers{Spectral Methods for Passive Imaging}{\TheAuthors}

\title{{\TheTitle}\thanks{
The material in this paper
was presented in part at the International Conference on Sampling Theory and Applications (SampTA), Tallinn, Estonia, July 2017 \cite{lee2017sampta}. \funding{This work was supported in part by NSF grants IIS 14-47879, CCF 14-22540 and by DFG grants KR 4512/1-1, KR 4512/2-1.}}}

\author{
  Kiryung Lee\thanks{School of ECE, Georgia Institute of Technology, Atlanta, GA (\email{kiryung@ece.gatech.edu}, \email{jrom@ece.gatech.edu}).}
  \and
  Felix Krahmer\thanks{Department of Mathematics, Technical University of Munich, Garching, Germany (\email{felix.krahmer@tum.de}).}
  \and
  Justin Romberg\footnotemark[2]
}

\usepackage{amsopn}

\ifpdf
\hypersetup{
  pdftitle={\TheTitle},
  pdfauthor={\TheAuthors}
}
\fi

\begin{document}

\maketitle

\begin{abstract}

We study the problem of passive imaging through convolutive channels.  A scene is illuminated with an unknown, unstructured source, and the measured response is the convolution of this source with multiple channel responses, each of which is time-limited.  Spectral methods based on the commutativity of convolution, first proposed and analyzed in the 1990s, provide an elegant mathematical framework for attacking this problem.  However, these now classical methods are very sensitive to noise, especially when working from relatively small sample sizes.

In this paper, we show that a linear subspace model on the coefficients of the impulse responses of the channels can make this problem well-posed.  We derive non-asymptotic error bounds for the generic subspace model by analyzing the spectral gap of the cross-correlation matrix of the channels relative to the perturbation introduced by noise.  Numerical results show that this modified spectral method offers significant improvements over the classical method and outperforms other competing methods for multichannel blind deconvolution.

\end{abstract}

\begin{keywords}
  passive imaging, blind deconvolution, perturbation analysis, random matrices
\end{keywords}

\begin{AMS}
  15B52, 93B30, 94A12
\end{AMS}

\section{Introduction}

We give a rigorous analysis of the passive imaging problem.  A scene is illuminated by an ambient source that we cannot control or observe.  This source signal is always active, having no discernible ``on'' or ``off'' time, and is unstructured.  We observe the convolution of this source with $M$ unknown {\em channel impulse response} sequences over a window of time.  The goal is to estimate this ensemble of impulse responses, which in many applications reveals the structure of the environment being sensed.
Problems of this type arise in a wide variety of applications including as opportunistic channel estimation in underwater acoustics \cite{sabra2004blind,sabra2010ray,byun2017blind,tian2017multichannel}, seismic interferometry \cite{curtis2006seismic}, and passive synthetic aperture imaging \cite{garnier2015passive}.

As described fully in Section~\ref{sec:spectral}, this is a multichannel blind deconvolution problem, where we observe the output of a number of linear time-invariant systems all driven by a common source.  We will focus entirely on estimating the responses of these system, and treat the (unknown) source signal as a supporting actor whose only role is to help us collect information about these channels.

When the channel impulse responses have a finite length $K$, then a fundamental technique for performing this estimation, developed in the signal processing literature in the 1990s (see, for example, \cite{xu1995least,moulines1995subspace}), is to form a cross-correlation matrix from the channel outputs and then estimate the channel responses by estimating the null space of this matrix.  This method is reviewed in Section~\ref{sec:ccm} below.  This classical theory shows that the cross-convolution method is consistent: as the number of noisy observations we make increases, the channel estimates asymptotically become aligned with the true underlying impulse responses.

From a finite number of samples, the stability of this process, both in theory and in practice, depends critically on the {\em spectral gap} of this cross-correlation matrix.  For even the simplest concrete instances of this problem, this gap tends to be vanishingly small; a typical example is shown in Figure~\ref{fig:cgsp} (and described in Sections~\ref{sec:ccm} and \ref{sec:scccm} below).

The main contribution of this paper is to show that if an additional structural constraint on the channel responses is imposed, then this spectral gap provably widens, stabilizing this channel estimation procedure.  In particular, we constrain the length-$K$ channel responses to be members of a known $D$ dimensional subspace.  Enforcing this constraint requires only a straightforward modification to the cross-convolution method.  Our results show that for a generic $D$ dimensional subspace (i.e. a subspace chosen at random), the principal angle between the true channel responses and their estimates decreases as (a) the number of observations in each channel increases, (b) the signal-to-noise ratio of the observations increases, (c) the number of channels increases, (d) the model becomes more restrictive, meaning $D$ decreases relative to $K$.

Our analysis of the subspace constrained cross-convolution method (SCCC) uses as its starting point the classical Davis-Kahan bound on the deviation of the eigenvectors computed from a perturbed observation of a positive semidefinite matrix.  Bounding the size of this perturbation in terms of the observation noise involves bounding the spectral norms of random matrices with entries given as coupled high order polynomials of subgaussian random variables.  These norms are written as the suprema of second order chaos processes, for which there are recently developed concentration results \cite{krahmer2014suprema,lee2015rip}.  Application of these concentration results involves computing entropy estimates for various norms.  In particular, the entropy estimate for a block norm in Appendix~\ref{sec:entropy_blocknorm} is a novel result derived using the polytope approximation and polar duality, which might be of independent interest.

The Monte Carlo simulation results in Section~\ref{sec:numres} demonstrate the practical gains that the SCCC method offers over the classical cross-convolution method. In practice, the estimator produces accurate results when the number of samples per channel $L$ is a relatively small multiple of $K$.  In this regime, the estimation error scales (as a function of $L$) in the same manner as the oracle solution, where the source is known and the channels are recovered using standard least-squares.  We also demonstrate that SCCC outperforms both the classical method and the recently proposed approach in \cite{ling2016self} for a underwater acoustics simulation with a realistic (non-random) subspace model.

\subsection*{Related work}
As mentioned above, the multichannel blind deconvolution problem was studied with intense interest in the signal and image processing literature in the 1990s; the methods most closely related to the work below are described in \cite{xu1995least,gurelli1995evam,moulines1995subspace}, and good overviews of general work on this problem can be found in \cite{liu96re,tong98mu}.  Many of these algorithms use models on the source signal and channels, and develop consistency results under different modeling assumptions; see \cite{harikumar99pe,giannakis00bl,zhou06mu,gunther00on} for representative examples from image processing.  To our knowledge, no theoretical results exist for these algorithms when there are a finite number of noisy samples.  More recently, necessary and sufficient conditions for the generic identifiability of this problem under various geometric priors have been presented in \cite{li2016optimal}.

A different linearization for the multichannel problem was introduced in \cite{balzano2007blind,morrison2009mca,nguyen2013subspace} and recently studied thoroughly in \cite{ling2016self}.  The model presented there is different in that the channels are not limited in time, a key piece of beneficial structure that our method exploits.  These methods also impose a structural constraint on the source signal, while we view the source signal as unstructured.  

Single channel blind deconvolution of signals belonging to low-dimensional subspaces has also been rigorously studied recently.  Identifiability results under various models were studied in \cite{choudhary14fu,choudhary2014sparse,li2016identifiability,li2017identifiability,kech2017optimal}. Convex optimization algorithms based on ``lifting'' were analyzed in \cite{ahmed2014blind}, followed by a similar result for a gradient descent algorithm \cite{li2016rapid}.  An alternating minimization algorithm for blind deconvolution under sparsity models also with subsampling has been analyzed in \cite{lee2017blind}.  While it is possible to extend these methods to the multichannel scenario, unlike the context of passive imaging, one needs strong geometric priors on both the source and impulse responses.  This scenario is different from what we consider in this paper.  Simultaneous wavelet estimation and deconvolution of seismic reflection signals \cite{kaaresen1998multichannel}, auto-calibrated parallel imaging \cite{griswold2002generalized}, motion deblurring using multiple images \cite{zhu2012deconvolving} are examples of relevant applications of multichannel blind deconvolution with geometric priors.

\section{Spectral Methods for Blind Deconvolution}
\label{sec:spectral}

In this section, we formulate the FIR multichannel blind deconvolution and describe spectral methods based on the cross convolution.

\subsection{Problem formulation}

We observe an unknown signal $\vx\in\mathbb{C}^L$ convolved with multiple unknown channel impulse responses $\vh_1,\ldots,\vh_m\in\mathbb{C}^L$ with the observations corrupted by additive noise $\vw_m\in\mathbb{C}^L$:
\begin{equation}
	\label{eq:mcmdl}
	\vy_m = \vh_m \circledast \vx + \vw_m, \quad m = 1,\dots,M,
\end{equation}
where the convolution $\circledast$ is circular\footnote{In our model, the source is opportunistic and ``always on'', and so the observations in \eqref{eq:mcmdl} might better be modeled by a windowed (time-limited) linear convolution.  To make a strict correspondence, the source would need to be periodic, which is an additional structural assumption.  Having access to the full circular convolution greatly simplifies the analysis, and the discrepancy between these two models is marginal when the number of observations $L$ dominates the length of the impulses responses $K$.  Spectral methods similar to the one presented above that are explicitly based on the time-limited linear convolution are presented in \cite{xu1995least,moulines1995subspace}; in practice, small gains might be realized by using these closely related methods.}, i.e.
\[
	y_m[\ell] = \sum_{k=1}^L h_m[k] x[(k-\ell)\bmod L] \,+ w_m[\ell].
\]
Our goal is to recover the channel responses $\{\vh_m\}$ from the observations $\{\vy_m\}$.

We will assume that the filters have impulse responses of length $K$; this simply means that the last $L-K$ entries of each $\vh_m$ are zero.  We denotes these non-zero entries using $\underline{\vh}\in\mathbb{C}^K$, with the relation
\[
	\vh_m = \mS^*\underline{\vh}_m,\quad \mS = \begin{bmatrix} \mId_K & \mzero_{K,L-K}\end{bmatrix}.
\]
The operator $\mS: \mathbb{C}^L \to \mathbb{C}^K$ restricts a given vector of length $L$ to its subvector with the first $K$ elements.  The adjoint $\mS^*: \mathbb{C}^K \to \mathbb{C}^L$ pads $L-K$ zeros to a given vector of length $K$.

When $L\geq 2K$, the nonzero terms in the linear convolution $\underline{\vh}_m\ast\vx$ and the circular convolution $\vh_m\circledast\vx$ will match, so this model applies to scenarios when we have fixed channels that are being continuously excited by an unknown input, and we observe a ``snapshot'' of length $L$ of their outputs.

\subsection{Cross-convolution method}
\label{sec:ccm}

Our method is a modification of the cross-convolution method introduced 20 years ago in \cite{xu1995least}.  The core idea is simple: we use the fact that multiple convolutions commute with one another to impose a set of linear constraints that the channel responses must obey, and then find the (unique up to scale) set of channels responses that obey these linear constraints.  To see how this is done, suppose that the measurements we make are free of noise, $\vy_m = \vx\circledast\vh_m$.  Then for any pair of channels $m,n$
\begin{equation}
	\label{eq:convcommute}
	\vy_m\circledast\vh_n = \vx\circledast\vh_m\circledast\vh_n = \vy_n\circledast\vh_m.
\end{equation}
Thus the pair of observations $\vy_n,\vy_m$ can be used to construct a set of $L$ constraints on the $2K$ variables in the channel coefficients $\underline{\vh}_n,\underline{\vh}_m$.

To make this more precise, let $\mT_{\vv}$ be the $L\times K$ matrix whose action $\mT_{\vv}\underline{\vh}$ circularly convolves $\vv = [v_1,\dots,v_L]^\transpose \in\mathbb{C}^L$ with $\underline{\vh}$ after zero-padding:
\[
	\mT_{\vv} = \mC_{\vv}\mS^*,
\]
where $\mC_{\vv} \in \mathbb{C}^{L \times L}$ is a circulant matrix defined by
\[
	\mC_{\vv} :=
	\begin{bmatrix}
		v_1 & v_L & v_{L-1} & \cdots & v_2 \\
		v_2 & v_1 & v_L & \cdots & v_3 \\
		\vdots & & & \ddots & \\
		v_L & v_{L-1} & v_{L-2} & \cdots & v_1
	\end{bmatrix}.
\]
Then we can write \eqref{eq:convcommute} as
\[
	\mT_{\vy_m}\underline{\vh}_n - \mT_{\vy_n}\underline{\vh}_m = \vzero.
\]
We can represent all $M(M-1)/2$ such constraints in one linear system.  With
\begin{equation}
	\label{eq:xcorrmat}
	\mY =
	\begin{bmatrix}
		\mY^{(1)} \\
		\mY^{(2)} \\
		\vdots \\
		\mY^{(M-1)}
	\end{bmatrix},
	\quad\text{where}\quad
	\mY^{(i)}
	=
	\begin{bmatrix}
		\underbrace{
		\begin{matrix}
			\mzero_{L,K} & \dots & \mzero_{L,K} \\
			\vdots & & \vdots \\
			\mzero_{L,K} & \dots & \mzero_{L,K}
		\end{matrix}
		}_{\text{$(M-i) \times (i-1)$ blocks}}
		&
		\begin{matrix}
			\mT_{\vy_{i+1}} \\
			\vdots \\
			\mT_{\vy_M}
		\end{matrix}
		\underbrace{
		\begin{matrix}
			-\mT_{\vy_i} & & \\
			& \ddots & \\
			& & -\mT_{\vy_i}
		\end{matrix}
		}_{\text{$(M-i) \times (M-i)$ block diagonal}}
	\end{bmatrix},
\end{equation}
we know that $\underline{\vh} = [\underline{\vh}_1^\transpose, \ldots, \underline{\vh}_M^\transpose]^\transpose$ will be in the null space of $\mY$.  Indeed, under the mild condition that the $z$-transforms of the $\underline{\vh}_m$ do not share common zeros, the null space of $\mY$ is one dimensional, containing only the scalar multiples of $\underline{\vh}$ \cite{xu1995least}.

In any practical scenario, noise (and possibly other perturbations), will keep \eqref{eq:convcommute} from holding exactly, and $\mY$ will in general not have a null space.  The channel estimates, then, are formed by finding the vector that is as close to a null vector as possible; after forming $\mY$ from the observations, we solve
\begin{equation}
	\label{eq:minprog}
	\minimize_{\underline{\vv}\in\mathbb{C}^{MK}}~\|\mY\underline{\vv}\|_2^2 \quad\text{subject to}\quad \|\underline{\vv}\|_2 = 1.
\end{equation}
The solution to the above is of course given by the eigenvector of $\mY^*\mY$ corresponding to the smallest eigenvalue.

The matrix $\mY$ can be unwieldy for large $M$, its dimensions are $M(M-1)/2\times KM$.  However, we can form the smaller $KM\times KM$ matrix $\mY^*\mY$ in a computationally efficient way using fast convolutions.  $\mY^*\mY$ can be thought of as an $M\times M$ array of $K\times K$ matrices; a quick calculation shows that $K\times K$ block $\mB_{n,m}$, corresponding to rows $(n-1)K+1$ to $nK$ and columns $(m-1)K+1$ to $mK$ in $\mY^*\mY$, is given by
\[
	\mB_{n,m} =
	\begin{cases}
		\sum_{m'\not=m} \mT_{\vy_{m'}}^*\mT_{\vy_m} & m=n, \\
		-\mT_{\vy_m}^*\mT_{\vy_n} & m\not= n.
	\end{cases}
\]
Thus $\mY^*\mY$ can be computed with $M(M+1)/2$ convolutions of length $L$.  Computing the solution to \eqref{eq:minprog} can be done with an eigenvalue decomposition in $O(M^3K^3)$ time.  For large values of $MK$, the solution can be computed with the power method, with each application of $\mY^*\mY$ computed using fast convolutions.

Under certain statistical assumptions on the noise, this estimate is consistent: as $L\rightarrow\infty$, the smallest eigenvector of $\mY^*\mY$ goes to (a scalar multiple of) $\underline{\vh}$.  However, to date there is no rigorous analysis of the {\em stability} of this procedure.  There are no non-asymptotic accuracy bounds that tell us what kind of performance we should expect for a certain number of channels $M$, filter lengths $K$, and observation times $L$.

The effect of noise on the accuracy of the estimate given by \eqref{eq:minprog} can be understood using the spectral properties of the ``noise-free'' cross correlation matrix.  We write the noisy measurements as
\[
	\vy_m = \vs_m + \vw_m,\quad\text{where}\quad \vs_m = \vh_m\circledast\vx.
\]
The cross correlation matrix $\mY$ is simply the sum of the cross correlation matrix $\mY_{\mathrm{s}}$ for the signals $\vs_m$ (i.e.\ create $\mY_{\mathrm{s}}$ as in \eqref{eq:xcorrmat} using $\mT_{\vs_i}$ in place of the $\mT_{\vy_i}$) and the cross correlation matrix $\mY_w$ for the noise signals $\vw_m$.  The estimate of the channels is formed by solving
\begin{equation}
	\label{eq:minprog2}
	\minimize_{\underline{\vv}\in\mathbb{C}^{MK}}~\underline{\vv}^*\left(\mY_{\mathrm{s}}^*\mY_{\mathrm{s}} + \mE\right)\underline{\vv}
	\quad\text{subject to}\quad \|\underline{\vv}\|_2 = 1,
\end{equation}
where
\[
	\mE = \mY_{\mathrm{s}}^*\mY_{\mathrm{n}} + \mY_{\mathrm{n}}^*\mY_{\mathrm{s}} + \mY_{\mathrm{n}}^*\mY_{\mathrm{n}}.
\]
From the discussion above, we know that in the noise-free case ($\mE=\mzero$), we will recover the true channel responses.  In expectation, the $\mE$ matrix becomes a scalar multiple of the identity, and the eigenvectors (and relative order of the eigenvalues) does not change.  From a finite number of samples, how closely the solution to \eqref{eq:minprog2} matches the noise-free solution depends on the size of $\mE$ relative to the {\em spectral gap} of $\mY_{\mathrm{s}}^*\mY_{\mathrm{s}}$, which is the size of its second smallest (or smallest non-zero) eigenvalue.  This is codified in the classical Davis-Kahan $\sin\theta$-theorem \cite{davis1970rotation}.

\begin{theorem}[{$\sin\theta$ theorem \cite[Corollary~7.2.6]{golub2012matrix}}]
	\label{thm:sintheta}
	Let $\mA, \mE \in \mathbb{C}^{n \times n}$ satisfy that $\mA$ and $\mA+\mE$ are positive semidefinite. Let $\vq$ (resp. $\widehat{\vq}$) denote the eigenvector of $\mA$ (resp. $\mA+\mE$) corresponding to the smallest eigenvalue.
Suppose that $\lambda_{n-1}(\mA) > \lambda_n(\mA)$.  If
	\begin{equation}
		\label{eq:smallperturb}
		\norm{\mE} \leq \frac{\lambda_{n-1}(\mA)-\lambda_n(\mA)}{5},
	\end{equation}
	then
	\begin{equation}
		\label{eq:errbnd}
		\sin\angle(\vq,\widehat{\vq}) \leq \frac{4 \norm{\mE\vq}_2}{\lambda_{n-1}(\mA)-\lambda_n(\mA)}.	
	\end{equation}
\end{theorem}

\begin{rem}
\rm
The error bound in the Davis-Kahan theorem is known to be sharp for general perturbations.  Recent results in \cite{vu2011singular,o2013random} have provided refined bounds for unstructured random perturbations, but unfortunately do not apply to our perturbation matrix $\mE$.
\end{rem}

Since the eigenvectors are unit norm, having a bound on angle between them is almost the same as having an error bound (up to a global phase), i.e.
\begin{equation}
    \label{eq:comparable}
    \sin\angle(\vq,\widetilde{\vq})
    \leq \min_{\theta \in [0,2\pi)} \norm{\vq-e^{\mathfrak{i}\theta}\widetilde{\vq}}_2
    \leq \sqrt{2} \sin\angle(\vq,\widetilde{\vq}).
\end{equation}

As discussed above, when the channels are identifiable, $\lambda_{KM}(\mY_{\mathrm{s}}^*\mY_{\mathrm{s}}) = 0$, and so we will have guarantees for the robustness of \eqref{eq:minprog2} when $\lambda_{KM-1}(\mY_{\mathrm{s}}^*\mY_{\mathrm{s}})$ is large compared to $\norm{\mE}$.  Unfortunately, this smallest non-zero eigenvalue $\lambda_{KM-1}(\mY_{\mathrm{s}}^*\mY_{\mathrm{s}})$ is typically very small in magnitude.  Figure~\ref{fig:eig1} shows a typical example; here we create $\mY_{\mathrm{s}}$ from $M=4$ channels of length $K=256$; the channel impulse responses themselves were generated at random.  In this example, $\lambda_1(\mY_{\mathrm{s}}^*\mY_{\mathrm{s}}) = 1$ but $\lambda_{KM-1}(\mY_{\mathrm{s}}^*\mY_{\mathrm{s}}) = 4.7 \times 10^{-5}$, and so we only have robustness guarantees for the mildest perturbations.  The practical performance of the estimator is poor in even mild amount of additive noise, as the experiments in Figure~\ref{fig:cgsp} suggest.

\subsection{Subspace-constrained cross-convolution (SCCC) method}
\label{sec:scccm}

In this paper, we show that the introduction of a linear model for the channel responses can tangibly increase the size of this spectral gap.  Using a linear subspace to model the channel responses has had some empirical success in the literature.  For example, in \cite{tian2017multichannel} a data-driven linear model is constructed for underwater acoustic channels for the purpose of ocean tomography.

Along with having an impulse response of limited length, we will also assume that the $\underline{\vh}_m$ lie in known subspaces of dimension $D<K$.  This means that each $\underline{\vh}_m$ can be expressed as $\underline{\vh}_m = \mPhi_m\vu_m$, where the columns of $\mPhi_m$ form a basis for the model subspace, and the $\vu_m$ are the expansion coefficients in this basis --- recovering the $\vu_m$ is now the same as recovering the channel responses $\underline{\vh}_m$.  The concatenated channels are written
\begin{equation}
    \label{eq:spmdl}
	\underline{\vh} = \mPhi\vu,
	\quad
	\mPhi =
	\begin{bmatrix}
		\mPhi_1 &  \\
		& \ddots \\
		& & \mPhi_M
	\end{bmatrix},
	\quad
	\vu =
	\begin{bmatrix}
		\vu_1 \\ \vdots \\ \vu_M
	\end{bmatrix}.
\end{equation}
With this model in place, the channel coefficients $\vu$ will be in the null space of $\mY_s\mPhi$.

The estimation procedure has to be modified to account for a slight bias introduced by the linear model.  With random uncorrelated noise, $\mathbb{E}[\vw_m] = \vzero$, $\mathbb{E}[\vw_m\vw_m^*] = \sigma_w^2\mId$, we have
\[
	\mathbb{E}[\mY_n] = \mzero,\quad\text{and}\quad \mathbb{E}[\mY_n^*\mY_n] = \sigma_w^2(M-1)L\mId,
\]
and so
\begin{align*}
	\mathbb{E}[\mPhi^*\mY^*\mY\mPhi] &= \mPhi^*\mY_s^*\mY_s\mPhi + \sigma_w^2(M-1)L\cdot\mPhi^*\mPhi.
\end{align*}
To make the perturbation from the noise-free cross-correlation matrix zero mean, we will solve
\begin{equation}
	\label{eq:Phiminprog}
	\minimize_{\vz\in\R^{MD}}~\vz^*\mPhi^*(\mY^*\mY - \sigma_w^2(M-1)L\mId_{MK})\mPhi\vz,
	\quad\text{subject to}\quad \|\vz\|_2 = 1.
\end{equation}
Again, the solution is the eigenvector corresponding to the smallest eigenvalue of $\mPhi^*(\mY^*\mY - \sigma_w^2(M-1)L\mId_{MK})\mPhi$.

\begin{rem}
{\rm
If $\mPhi^* \mPhi = \mId_{MD}$, then since adding a scalar multiple of the identity does not perturb eigenvectors, we may ignore $\sigma_w^2(M-1)L\mId_{MK}$ in \eqref{eq:Phiminprog}. Otherwise, subtracting the noise covariance $\sigma_w^2(M-1)L\mId_{MK}$ from $\mY^*\mY$ further suppresses the error in the estimated impulse responses.  In practice, the noise variance $\sigma_w^2$ needs to be estimated and the error in this estimate will propagate to the estimate of the impulse responses.  For simplicity of analysis, we assume that $\sigma_w^2$ is known a priori.
}
\end{rem}

Figure~\ref{fig:eig2} shows the effect of the subspace constraint on the spectral gap.  Here, a generic subspace was chosen by generating $\mPhi$ at random.  The entries of $\mPhi \in \mathbb{C}^{256 \times 8}$ were generated as independent copies of a standard complex Gaussian random variable.  The size of smallest non-zero eigenvalue is now significantly more distinct ($\lambda_{MD-1}(\mPhi^* \mY_{\mathrm{s}}^* \mY_{\mathrm{s}} \mPhi) / \lambda_1(\mPhi^* \mY_{\mathrm{s}}^* \mY_{\mathrm{s}} \mPhi) = 0.4$).  As the numerical results in Section~\ref{sec:numres} show, adding subspace constraints of this nature does indeed lead to significant robustness of the method in the presence of noise.

Our main results, detailed in Section~\ref{sec:mainres}, quantify this spectral gap for generic subspaces $\mPhi$.

\begin{figure}[htbp]
  \centering
  \subfloat[$\mY_{\mathrm{s}}^* \mY_{\mathrm{s}}$]{\includegraphics[width=2.1in]{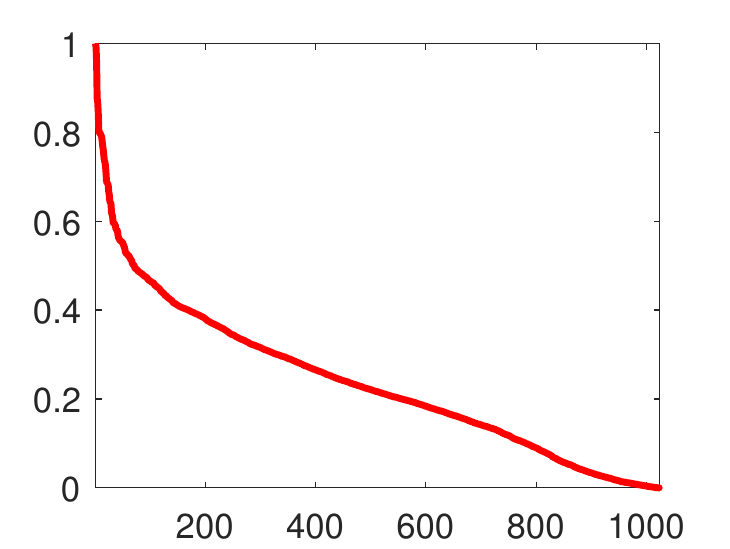}\label{fig:eig1}}
  \subfloat[$\mPhi^* \mY_{\mathrm{s}}^* \mY_{\mathrm{s}} \mPhi$]{\includegraphics[width=2.1in]{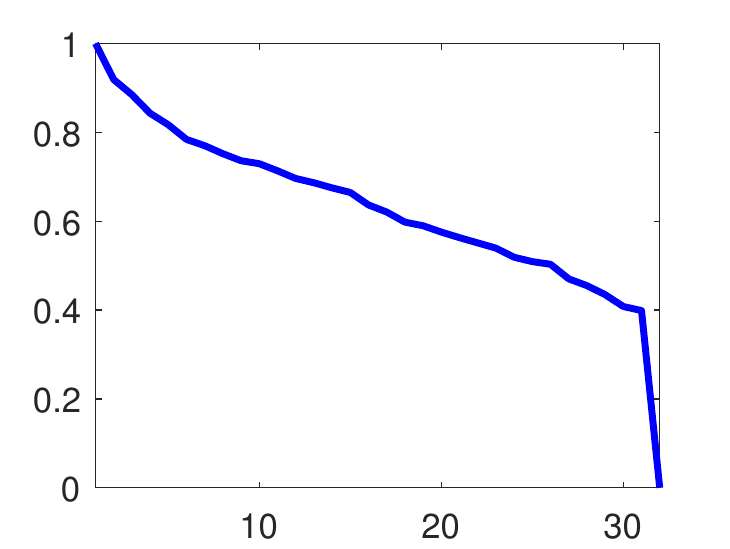}\label{fig:eig2}}
  \caption{(a) The eigenvalue spectrum of an example cross correlation matrix created from $M=4$ channels of length $K=256$ with random coefficients.  Note that there is almost no gap between the second smallest eigenvalue and zero.  (b) The eigenvalue spectrum after introducing a linear model of dimension $D=30$.  Note that the spectral gap is now pronounced.}
  \label{fig:eig}%
\end{figure}

\section{Main Results}
\label{sec:mainres}

\subsection{Non-asymptotic Analysis}

Our main results give non-asymptotic performance guarantees for the subspace-constrained cross-convolution method.  We make the following two assumptions throughout:
\begin{itemize}

  \item[(A1)] {\bf Generic subspaces.}  It is clear that some choices of linear channel models will be better than others.  We will investigate the ``generic'' case, where the bases themselves are generated at random.  In particular, we will assume that $\mPhi_1,\dots,\mPhi_M$ are independent copies of a $K$-by-$D$ complex Gaussian matrix whose entries are independent and identically distributed (iid) as $\mathcal{CN}(0,1)$.  Our theorems below hold with high probability with respect to this draw of the $\mPhi_m$; we might interpret this as saying that the results hold for ``most'' subspace models.  In Section~\ref{sec:numres} below, we empirically confirm this performance also for choices of $\mPhi_m$ with additional structure as it would appear in applications, even though they do not arise from the model analyzed in this paper.

  \item[(A2)] {\bf Random noise.}  The perturbations to the measurements $\vw_1,\dots,\vw_M \in \mathbb{C}^L$ are iid subgaussian vectors with $\mathbb{E}[\vw_m]=\vzero$ and $\mathbb{E}[\vw_m\vw_m^*] = \sigma_w^2\mId_L$, and are independent of the bases $\{\mPhi_m\}$.

\end{itemize}

We present two main theorems below.  In the first one, we assume that the input itself is a white random process. In the second one, we study deterministic inputs with a mild structural assumption on the common source signal that essentially amounts to $\vx$ being spread out in the frequency domain; the resulting error bounds are slightly weaker than for the random model.

The theorems provide sufficient conditions on the number of samples $L$ we need to observe at the output of each channel in order to guarantee a certain level of accuracy in the estimate $\hat{\underline{\vh}}$ found by solving \eqref{eq:Phiminprog} for $\widehat{\vu}$ and then taking $\hat{\underline{\vh}} = \mPhi\widehat{\vu}$.  The number of samples we need will depend on the length of the filter responses $K$, their intrinsic dimensions $D$, the number of channels $M$, and the signal-to-noise-ratio (SNR) defined as
\begin{equation}
    \label{eq:SNRdef}
    \eta
    := \frac{\mathbb{E}_{\vphi}[\sum_{m=1}^M \norm{\vh_m \cconv \vx}_2^2]}{\mathbb{E}_{\vw}[\sum_{m=1}^M \norm{\vw_m}_2^2]}.
\end{equation}
Under (A1) and (A2), it follows from the commutativity of convolution and Lemma~\ref{lemma:expectation1} that $\eta$ simplifies as
\begin{equation}\label{eq:etasimp}
\eta = \frac{K \norm{\vx}_2^2 \norm{\vu}_2^2}{M L \sigma_w^2}.
\end{equation}
In addition, the bounds will depend on the channel impulse responses all being roughly the same size.  We measure the disparity in impulse response energies using the flatness parameter
\begin{align}
	\label{eq:flat}
	\mu &:= \max_{1\leq m\leq M} \frac{\sqrt{M} \norm{\vu_m}_2}{\norm{\vu}_2}.
\end{align}
When a small number of the $M$ impulse responses are significantly greater than the others, we will have $\mu\approx\sqrt{M}$.  In this case, we expect to have longer observation times, as we are only getting a small number of diverse looks at the signal.  Our results are most interesting when $\mu$ is a constant on the order of $1$.  Qualitatively, this means that each channel is roughly as important as the others.

We now present the first of our main results.  Theorem~\ref{thm:main_randx} below assumes a {\em random common source} signal $\vx$.  The bound on the number of observations $L$ sufficient to guarantee a certain accuracy in the channel estimates is a complicated expression involving the number of channels $M$, their maximum impulse response lengths $K$, their intrinsic dimension $D$, the SNR $\eta$, the channel flatness $\mu$, and the level of accuracy $\epsilon$.  But in reasonable scenarios where the noise is not extreme ($\eta$ is a constant), and $K$ and $M$ are not too different, we have
\[
	\sin\angle(\widehat{\underline{\vh}},\underline{\vh})\leq\epsilon
	\quad\text{when}\quad
	L \gtrsim \sqrt{KD}/\epsilon,
\]
with the inequality on the right holding to within log factors.  As a point of reference, we are estimating $MD$ channel coefficients from $ML$ samples at the outputs; we have the same number observations as unknowns when $L\gtrsim D$.  As $D\leq K$, our estimate scales in a mildly unsatisfying way, though as the recovery procedure is highly nonlinear, it is unclear what form an optimal scaling would take.

\begin{theorem}[Random Source]
	\label{thm:main_randx}
	We observe noisy channel outputs $\{\vy_m\}$ as in \eqref{eq:mcmdl}, with SNR $\eta$ as in \eqref{eq:SNRdef}, and form an estimate $\underline{\vh}$ of the channel responses by solving \eqref{eq:Phiminprog}.  Suppose assumptions (A1) and (A2) above hold, let $\vx$ be a sequence of zero-mean iid subgaussian random variables with variance $\sigma_x^2$, $\eta \geq 1$, $\mu = O(1)$, and $L \geq 3K$.\footnote{Without the subspace prior, $L > K$ is necessary to claim that $\mY^* \mY$ has nullity 1 in the noiseless case.  We used $L \geq 3K$ in the proof in order to use the identity that the circular convolutions of three vectors of length $K$ modulo $L$ indeed coincide with their linear convolution.}
	Then for any $\beta \in \mathbb{N}$, there exist absolute constants $C > 0, \alpha \in \mathbb{N}$ and constants $C_1(\beta), C_2(\beta)$  such that if there are a sufficient number of channels,
	\begin{equation}
		\label{eq:condM}
        M \geq C_1(\beta) \log^\alpha (MKL),
	\end{equation}
	that are sufficiently long,
	\begin{equation}
		\label{eq:condK}
        K \geq C_1(\beta) D \log^\alpha (MKL),
	\end{equation}
	and we have observed the a sufficient number of samples at the output of each channel,
	\begin{equation}
		\label{eq:condL_randx}
        L \geq
        \frac{C_1(\beta) \log^\alpha (MKL)}{\eta}
        \Big(
        \frac{K}{M^2} + D \,
        \Big),
	\end{equation}
	then with probability exceeding $1-CK^{-\beta}$, we can bound the approximation error as
	\begin{equation}
		\label{eq:error_randx}
        \sin \angle(\widehat{\underline{\vh}},\underline{\vh})
        \leq
        C_2(\beta) \log^{\alpha}(MKL)
        \Big(
        \frac{1}{\sqrt{\eta L}}
        \Big(
        \frac{\sqrt{K}}{M} + \sqrt{D} \,
        \Big)
        + \frac{\sqrt{D}}{\eta \sqrt{ML}}
        \Big).
	\end{equation}

\end{theorem}

\begin{rem}
\rm
The SNR requirement $\eta \geq 1$ was introduced to simplify the expressions in Theorem~\ref{thm:main_randx}.  The conditions in the low SNR regime $\eta <1$ can be easily extracted from the proof of the theorem and Proposition~\ref{prop:main} below.
\end{rem}

Theorem~\ref{thm:main_randx} is interpreted as follows: When the dimension $D$ of the subspaces in \eqref{eq:spmdl} is small (up to a fraction of the ambient dimension $K$), the number of channels $M$ is large (depending weakly on the other dimension parameters only through log), and the length of observation is large enough ($L \geq 3K + C\sqrt{KD}$ for an absolute constant $C$), we can apply the Davis-Kahan theorem which provides an error bound for $\widehat{\vh}$.  The error bound in \eqref{eq:error_randx} converges to 0 if either $L$ or $\eta$ grows toward infinity.  Moreover, the error bound is nonasymptotic since it explicitly shows how the error depends on $L$ and $\eta$ when they are finite.  In a heuristic argument that counts the number of unknown parameters and the number of given equations, a necessary condition for the unique recovery of $\vx$ and $\{\underline{\vh}_m\}_{m=1}^M$ from noise-free measurements is given as $L \geq MD/(M-1)$.  It is unclear whether this is also a valid necessary condition under the finite impulse response structure.  On the other hand, with some diversity in $\vx$ (not necessarily under an explicit stochastic model), it has been shown $L \geq 3K$ suffices in this scenario.  Still, it is not clear whether this remains a valid necessary condition with the extra subspace model in \eqref{eq:spmdl}.  However, empirically, spectral methods including the classical and what we propose in this paper break down when $L < K$.  Along the above discussion, it is still open to answer whether the requirement on $L$ in \eqref{eq:condL_randx} is near optimal or not.

To prove Theorem~\ref{thm:main_randx}, we establish an intermediate result for the case where the input signal $\vx$ is deterministic.  In this case, our bounds depend on the spectral norm $\rho_{\vx}$ of the (appropriately restricted) autocorrelation matrix of $\vx$,
\[
	\rho_{\vx} := \|\widetilde{\mS}\mC_{\vx}^*\mC_{\vx}\widetilde{\mS}^*\|,
\]
where
\begin{equation}
\label{eq:deftildeS}
\widetilde{\mS} =
\begin{bmatrix}
\bm{0}_{K-1,L-K+1} & \mId_{K-1} \\
\mId_{2K-1} & \bm{0}_{2K-1,L-2K+1}
\end{bmatrix}.
\end{equation}

Then the deterministic version of our recovery result is:
\begin{theorem}[Deterministic Source]
	\label{thm:main_detx}
	Suppose that the same assumptions hold as in Theorem~\ref{thm:main_randx}, only with $\vx$ as a fixed sequence of numbers obeying
	\begin{equation}
		\label{eq:cond_rho}
		\rho_{\vx} \leq C_3\|\vx\|_2^2.
	\end{equation}
	If  \eqref{eq:condK} and \eqref{eq:condM} hold, and
	\begin{equation}
		\label{eq:condL_detx}
        L \geq
        \frac{C_1(\beta) \log^\alpha (MKL)}{\eta}
        \Big(
        \frac{K^2}{M^2} + KD
        \Big),
	\end{equation}
	then with probability exceeding $1-CK^{-\beta}$, we can bound the approximation error as
	\begin{equation}
		\label{eq:error_detx}
        \sin \angle(\widehat{\underline{\vh}},\underline{\vh})
        \leq
        \frac{C_2(\beta) \log^{\alpha}(MKL)}{\sqrt{\eta L}}
        \Big(
        \frac{K}{M} + \sqrt{KD} \,
        \Big).
	\end{equation}
\end{theorem}

The condition \eqref{eq:cond_rho} can be interpreted as a kind of incoherence condition on the input signal $\vx$.  Since
\[
	\rho_{\vx} \leq \|\mC_{\vx}\|^2 = L \norm{\widehat{\vx}}_\infty^2,
\]
where $\widehat{\vx} \in \mathbb{C}^L$ is the normalized discrete Fourier transform of $\vx$, it is sufficient that $\hat\vx$ is approximately flat for \eqref{eq:cond_rho} to hold.  This is a milder assumption than imposing an explicit stochastic model on $\vx$ as in Theorem~\ref{thm:main_randx}.  For the price of this relaxed condition, the requirement on $L$ in \eqref{eq:condL_detx} that activates Theorem~\ref{thm:main_detx} is more stringent compared to the analogous condition \eqref{eq:condL_randx} in Theorem~\ref{thm:main_randx}.

Theorems~\ref{thm:main_randx} and \ref{thm:main_detx} distinguish from a recent result \cite{ling2016self} in the following sense.  Ling and Strohmer \cite{ling2016self} analyzed the error bound for the least squares solution to a different linearized formulation in \cite{balzano2007blind}.  In their analysis, the unknown filters were assumed to follow stochastic subspace models, which span vectors fully supported on the entire observation period.  Obviously, these models do not explain the FIR structures arising in applications.  Unlike their analysis, we explicitly considered the case where the unknown filters are supported on a short interval.  On the other hand, the number of observations $L$ enabling the error bound in \cite{ling2016self} scales near optimally whereas $L$ grows faster in Theorems~\ref{thm:main_randx} and \ref{thm:main_detx}.  Again, models considered in these analyses are different and it is still open to verify whether a near optimal scaling can be achieved with limited randomness satisfying the FIR structure.

\subsection{Proof of Main Results}

The main results in Theorems~\ref{thm:main_randx} and \ref{thm:main_detx} are obtained by the following proposition.  Proposition~\ref{prop:main} identifies a sufficient condition for \eqref{eq:smallperturb}, which enables Theorem~\ref{thm:sintheta} and provides an error estimate of computing the most dominant eigenvector from a noisy matrix.  The sufficient condition is stated in terms of scaling of key parameters for the sake of interpretation.

\begin{proposition}
\label{prop:main}
Suppose the assumptions in (A1) and (A2) hold, $\rho_x$ satisfies \eqref{eq:cond_rho}, $L \geq 3K$, and $\mu \leq \sqrt{M}/2$.\footnote{By definition, the parameter $\mu$ always satisfies $1 \leq \mu \leq \sqrt{M}$.  In this perspective, $\mu \leq \sqrt{M}/2$ is a mild condition.}  Let $\rho_{x,w}$ denote the cross-correlation among the input $\vx$ and the noise terms $\vw_1,\dots,\vw_M$ defined by
\[
\rho_{x,w} := \max_{1\leq m\leq M} \norm{\widetilde{\mS} \mC_{\vx}^* \mC_{\vw_m} \widetilde{\mS}^*},
\]
where $\widetilde{\mS} \in \mathbb{R}^{(3K-2) \times L}$ is as in \eqref{eq:deftildeS}.
For any $\beta \in \mathbb{N}$, there exist absolute constants $C > 0, \alpha \in \mathbb{N}$ and constants $C_1(\beta), C_2(\beta)$ that only depend on $\beta$, for which the following holds.
If
\begin{equation}
\label{eq:condK_prop}
K \geq C_1(\beta) \mu^4 D \log^\alpha (MKL),
\end{equation}
\begin{equation}
\label{eq:condM_prop}
M \geq C_1(\beta) \mu^4 \log^\alpha (MKL),
\end{equation}
and
\begin{equation}
\label{eq:condL}
\frac{L}{C_1(\beta) \log^\alpha (MKL)} \geq
\frac{\rho_{x,w}^2}{\eta K \sigma_w^2 \norm{\vx}_2^2}
\Big(
\mu^2 \Big( \frac{D^2}{K} + \frac{K}{M^2} \Big) + D
\Big)
+ \frac{D}{\eta^2}
\end{equation}
then
\begin{equation}
\label{eq:prop_esterr}
\begin{aligned}
& \sin \angle(\widehat{\underline{\vh}},\underline{\vh}) \\
& \leq C_2(\beta) \log^{\alpha}(MKL)
\Big(
\frac{\rho_{x,w}}{\sqrt{\eta KL}\sigma_w \norm{\vx}_2}
\Big(
\mu \Big( \frac{D}{\sqrt{K}} + \frac{\sqrt{K}}{M} \Big) + \sqrt{D}
\Big)
+ \frac{\sqrt{D}}{\eta \sqrt{ML}}
\Big)
\end{aligned}
\end{equation}
holds with probability $1-CK^{-\beta}$.
\end{proposition}

\begin{proof}[Proof of Proposition~\ref{prop:main}]
Recall that we first compute an estimate $\widehat{\vu}$ of $\vu$. Then $\widehat{\underline{\vh}} = \mPhi \widehat{\vu}$ serves as an estimate of $\underline{\vh} = \mPhi \vu$.
Since the estimation error is measured in the principal angle, which is invariant under scalar multiplication, without loss of generality, we may assume that $\norm{\vu}_2 = 1$.
Indeed, the errors in the estimates $\widehat{\underline{\vh}}$ and $\widehat{\vu}$ are related by
\begin{equation}
\label{eq:rel_angles}
\sin \angle(\underline{\vh},\widehat{\underline{\vh}})
= \frac{\big\|\mP_{\widehat{\underline{\vh}}^\perp} \underline{\vh}\big\|_2}{\norm{\underline{\vh}}_2}
\leq \frac{\norm{\underline{\vh} - \widehat{\underline{\vh}}}_2}{\norm{\underline{\vh}}_2}
\leq \frac{\sigma_{\max}(\mPhi) \norm{\vu - \widehat{\vu}}_2}{\sigma_{\min}(\mPhi) \norm{\vu}_2}
\leq \frac{\sigma_{\max}(\mPhi)}{\sigma_{\min}(\mPhi)} \cdot \sqrt{2} \sin \angle(\vu,\widehat{\vu}),
\end{equation}
where the last step follows from \eqref{eq:comparable}.

By the assumption in (A1) and \eqref{eq:condK_prop}, the standard results on singular values of subgaussian matrices (e.g., see \cite[Theorem~II.13]{davidson2001local}) imply that the condition number of $\mPhi_m$ is upper bounded by 3 for $m=1,\dots,M$ with high probability as we choose $C_1(\beta)$ in \eqref{eq:condK_prop} large enough. We proceed the proof under this event. Then the condition number of $\mPhi$ is also upper bounded by 3.

Therefore, it suffices to focus on bounding the estimation error in $\widehat{\vu}$ in the principal angle. Note that $\widehat{\vu}$ is computed as the least dominant eigenvector of $\widetilde{\mA}=\mPhi^* (\mY^* \mY - \sigma_w^2 (M-1)L \mId_{MK}) \mPhi$. Furthermore, the target vector $\vu$ is the unique null vector of ${\mA} = \mathbb{E}[\mPhi^* (\mY^* \mY - \sigma_w^2 (M-1)L \mId_{MK}) \mPhi]$.
	
To see this, we decompose $\mY$ as $\mY = \mY_{\mathrm{s}} + \mY_{\mathrm{n}}$, where the noise-free portion $\mY_{\mathrm{s}}$ (resp. the noise portion $\mY_{\mathrm{n}}$) is obtained as we replace $\vy_m=\vh_m \conv \vx +\vw_m$ in $\mY$ by its first summand $\vh_m \conv \vx$ (resp. by its second summand $\vw_m$) for all $m=1,\dots,M$.
	Consequently, we have
	\[
	\mathbb{E}_{\vw}[\mY_{\mathrm{n}}^* \mY_{\mathrm{n}}] = \sigma_w^2 (M-1)L \mId_{MK}
	\quad \text{and} \quad
	\mathbb{E}_{\vw}[\mY_{\mathrm{s}}^* \mY_{\mathrm{n}}] = \vzero
	\]
	as well as
	\begin{align*}
	\mathbb{E}_{\vw}[\widetilde{\mA}]
	&= \mathbb{E}_{\vw}[\mPhi^* (\mY^* \mY - \sigma_w^2 (M-1)L \mId_{MK}) \mPhi] \\
	&= \mPhi^* \mY_{\mathrm{s}}^* \mY_{\mathrm{s}} \mPhi
	+ \mPhi^* \mathbb{E}_{\vw}[\mY_{\mathrm{s}}^* \mY_{\mathrm{n}}] \mPhi
	+ \mPhi^* \mathbb{E}_{\vw}[\mY_{\mathrm{n}}^* \mY_{\mathrm{s}}] \mPhi
	\\
	&+ \mathbb{E}_{\vw}[\mPhi^* (\mY_{\mathrm{n}}^* \mY_{\mathrm{n}} - \sigma_w^2 (M-1)L \mId_{MK}) \mPhi] \\
	&= \mPhi^* \mY_{\mathrm{s}}^* \mY_{\mathrm{s}} \mPhi.
	\end{align*}

	As shown in Section~\ref{sec:ccm}, by the construction of $\mY_{\mathrm{s}}$, the vector $\underline{\vh} = \mPhi \vu$ with the true filter coefficients is in the null space of $\mY_{\mathrm{s}}$. Therefore, $\vu$ is almost surely a null vector of the noise-free matrix $\mPhi^* \mY_{\mathrm{s}}^* \mY_{\mathrm{s}} \mPhi$ and hence also of its expectation $\mA = \mathbb{E}[\widetilde{\mA}] = \mathbb{E}[\mPhi^* \mY_{\mathrm{s}}^* \mY_{\mathrm{s}} \mPhi]$. The uniqueness follows from the first part of the following lemma, which is proved in Section~\ref{sec:proof:lemma:gap}.
	\begin{lemma}
		\label{lemma:gap}
		Under the hypothesis of Proposition~\ref{prop:main}, the following are true: i) The nullity of $\mathbb{E}[\mPhi^* \mY_{\mathrm{s}}^* \mY_{\mathrm{s}} \mPhi]$ is 1; ii) Nonzero eigenvalues of $\mathbb{E}[\mPhi^* \mY_{\mathrm{s}}^* \mY_{\mathrm{s}} \mPhi]$ are no less than $K^2 \norm{\vx}_2^2 \norm{\vu}_2^2/2=:\delta$.
	\end{lemma}
This lemma also establishes a lower bound for the gap between the two smallest eigenvalues of $\mA$. This spectral gap allows to distinguish the corresponding eigenspaces of $\mA$. Provided condition \eqref{eq:smallperturb}, that is, $\widetilde{\mA}$ does not deviate too much from its expectation $\mA$ in the spectral norm (this will be the main task of the remainder of this proof), this property also carries over to the eigenspaces of $\widetilde \mA$ and it follows from Theorem~\ref{thm:sintheta} that the least dominant eigenspace of $\widetilde \mA$ and $\mA$ are close to each other. Thus, up to a global phase, $\widehat{\vu}$ is a good estimate of $\vu$.

It remains to show that condition \eqref{eq:smallperturb} is satisfied with high probability.  To this end, we derive a tail estimate of the spectral norm of the random perturbation $\mE = \widetilde{\mA} - \mA$ and show that the perturbation relative to the spectral gap $\delta$ satisfies
\begin{equation}
\label{eq:bnd_snE}
\begin{aligned}
\frac{\norm{\mE}}{\delta }
&\leq C(\beta) \log^\alpha(MKL) \Big[ \Big(\sqrt{\frac{1}{M}} + \sqrt{\frac{D}{K}} \, \Big) \mu^2 \\
& + \frac{\rho_{x,w}}{\sqrt{\eta KL}\sigma_w \norm{\vx}_2}
\Big(
\mu \Big( \frac{D}{\sqrt{K}} + \frac{\sqrt{K}}{M} \Big) + \sqrt{D} \,
\Big)
+ \frac{\sqrt{D}}{\eta \sqrt{L}} \Big]
\end{aligned}
\end{equation}
with probability $1-CK^{-\beta}$, where $C(\beta)$ is a constant depending only on $\beta$.  By choosing $C_1(\beta)$ in \eqref{eq:condK_prop}, \eqref{eq:condM_prop}, and \eqref{eq:condL} large enough, we can make the right hand side of \eqref{eq:bnd_snE} less than 1/5. Thus \eqref{eq:smallperturb} is satisfied.

The derivation of \eqref{eq:bnd_snE} is rather involved for the following reasons:  The entries of the perturbation matrix $\mE$ are given as fourth order polynomials of subgaussian random variables.  In addition, the convolution structure in the construction of $\mY$ creates dependence relations between the matrix entries.  To analyze the perturbation, we decompose $\mE$ into three components of different polynomial order as follows.
\begin{equation}
\label{eq:Edecomp}
\begin{aligned}
\mE &=
\underbrace{\mPhi^* \mY_{\mathrm{s}}^* \mY_{\mathrm{s}} \mPhi - \mathbb{E}[\mPhi^* \mY_{\mathrm{s}}^* \mY_{\mathrm{s}} \mPhi]}_{\mE_s} \nonumber\\
&+ \underbrace{\mPhi^* \mY_{\mathrm{s}}^* \mY_{\mathrm{n}} \mPhi}_{\mE_c}
+ \mPhi^* \mY_{\mathrm{n}}^* \mY_{\mathrm{s}} \mPhi \nonumber\\
&+ \underbrace{\mPhi^* (\mY_{\mathrm{n}}^* \mY_{\mathrm{n}}- \sigma_w^2 (M-1)L \mId_{MK}) \mPhi}_{\mE_n}.
\end{aligned}
\end{equation}

The following lemmas, the proofs of which will be presented in Section~\ref{sec:proof_tech_lemma}, provide tail estimates of the components; the tail estimate in \eqref{eq:bnd_snE} is then obtained by combining these results via the triangle inequality.

\begin{lemma}
\label{lemma:Es}
Suppose that (A1) holds.
For any $\beta \in \mathbb{N}$, there exist a numerical constant $\alpha \in \mathbb{N}$ and a constant $C(\beta)$ that depends only on $\beta$ such that
\begin{equation}
\label{eq:bnd_snEs}
\frac{\norm{\mPhi^* \mY_{\mathrm{s}}^* \mY_{\mathrm{s}} \mPhi - \mathbb{E}[\mPhi^* \mY_{\mathrm{s}}^* \mY_{\mathrm{s}} \mPhi]}}{K^2 \norm{\vx}_2^2 \norm{\vu}_2^2}
\leq
C(\beta) \log^\alpha(MKL) \Big(\sqrt{\frac{1}{M}} + \sqrt{\frac{D}{K}} \, \Big) \mu^2
\end{equation}
holds with probability $1-CK^{-\beta}$.
\end{lemma}

\begin{lemma}
\label{lemma:Ec}
Suppose that (A1) holds.
For any $\beta \in \mathbb{N}$, there exists a constant $C(\beta)$ that depends only on $\beta$ such that, conditional on the noise vector $\vw$,
\begin{equation}
\label{eq:bnd_snEc}
\frac{\norm{\mPhi^* \mY_{\mathrm{s}}^* \mY_{\mathrm{n}} \mPhi}}{K^2 \norm{\vx}_2^2 \norm{\vu}_2^2}
\leq
\frac{C(\beta) \rho_{x,w}}{\sqrt{\eta KL}\sigma_w \norm{\vx}_2}
\Big(
\mu \Big( \frac{D}{\sqrt{K}} + \frac{\sqrt{K}}{M} \Big) + \sqrt{D}
\Big)
\end{equation}
holds with probability $1-CK^{-\beta}$.
\end{lemma}

\begin{lemma}
\label{lemma:En}
Suppose that (A1) holds.
For any $\beta \in \mathbb{N}$, there is a constant $C(\beta)$ that depends only on $\beta$ such that
\begin{equation}
\label{eq:bnd_snEn}
\frac{\norm{\mPhi^* (\mY_{\mathrm{n}}^* \mY_{\mathrm{n}}- \sigma_w^2 (M-1)L \mId_{MK}) \mPhi}}{K^2 \norm{\vx}_2^2 \norm{\vu}_2^2}
\leq
\frac{C(\beta) \log^\alpha (MKL)}{\eta} \cdot \sqrt{\frac{D}{L}}
\end{equation}
with probability $1-CK^{-\beta}$.
\end{lemma}

Finally, under the event where \eqref{eq:smallperturb} is satisfied, Theorem~\ref{thm:sintheta} implies that
\begin{equation}
\label{eq:bnd_snEu}
\sin(\vu, \widehat \vu)\leq \frac{4\norm{\mE \vu}_2}{\delta \norm{\vu}_2 }.
\end{equation}

To estimate the right hand side, we again decompose $\mE$ as in \eqref{eq:Edecomp}, so the triangle inequality yields
\begin{equation}
\label{eq:Eu_triangle}
\norm{\mE \vu}_2 \leq \norm{\mE_{\mathrm{s}} \vu}_2 + \norm{\mE_{\mathrm{c}} \vu}_2 + \norm{\mE_{\mathrm{c}}^* \vu}_2 + \norm{\mE_{\mathrm{n}} \vu}_2.
\end{equation}

To bound the first term, recall that $\vu$ is in the null space of $\mPhi^* \mY_{\mathrm{s}}^* \mY_{\mathrm{s}} \mPhi$, so we obtain that
\begin{equation}
\label{eq:bnd_snEsu}
\mE_{\mathrm{s}} \vu = \vzero.
\end{equation}

For the second and third summand, Lemma~\ref{lemma:Ec} yields that with probability $1-CK^{-\beta}$
\begin{equation}
\label{eq:bnd_snEcu}
\max(\norm{\mE_{\mathrm{c}} \vu}_2,\norm{\mE_{\mathrm{c}}^* \vu}_2) \leq \norm{\mE_{\mathrm{c}}} \norm{\vu}_2 \leq \frac{C(\beta) \rho_{x,w} K^{3/2} \norm{\vx}_2 \norm{\vu}_2^3 }{\sqrt{\eta L}\sigma_w}
\Big(
\mu \Big( \frac{D}{\sqrt{K}} + \frac{\sqrt{K}}{M} \Big) + \sqrt{D} \,
\Big).
\end{equation}
A bound for the last summand is provided by the following lemma, which is proved in Section~\ref{sec:proof:lemma:Enu}.
\begin{lemma}
\label{lemma:Enu}
Suppose that (A1) holds.
For any $\beta \in \mathbb{N}$, there is a constant $C(\beta)$ that depends only on $\beta$ such that
\begin{equation}
\label{eq:bnd_snEnu}
\frac{\norm{\mPhi^* (\mY_{\mathrm{n}}^* \mY_{\mathrm{n}}- \sigma_w^2 (M-1)L \mId_{MK}) \mPhi \vu}_2}{K^2 \norm{\vx}_2^2 \norm{\vu}_2^4}
\leq
\frac{C(\beta) \log^\alpha (MKL)}{\eta} \cdot \sqrt{\frac{D}{ML}}
\end{equation}
with probability $1-CK^{-\beta}$.
\end{lemma}

Inserting the bounds for the four summands into \eqref{eq:Eu_triangle} yields the error bound in \eqref{eq:prop_esterr}, which completes the proof.
\end{proof}

In the remainder of this section, we show how Theorems~\ref{thm:main_randx} and \ref{thm:main_detx} can be deduced    from Proposition~\ref{prop:main}.

\begin{proof}[Proof of Theorem~\ref{thm:main_detx}] Since most assumptions of the Theorem agree with the ones of
 Proposition~\ref{prop:main} it only remains to bound $\rho_{x,w}$. This is achieved by the following lemma, which is proved in Appendix~\ref{sec:proof:lemma:rho_wx}.

\begin{lemma}
\label{lemma:rho_wx}
Suppose (A2) holds and let $\vx$ be a fixed sequence of numbers obeying \eqref{eq:cond_rho}. For any $\beta \in \mathbb{N}$, there exists an absolute constant $C$ such that
\[
\rho_{x,w} \leq C K \sigma_w \sqrt{\rho_x} \sqrt{1 + \log M + \beta \log K}
\]
holds with probability $1 - K^{-\beta}$.
\end{lemma}

The theorem follows from a direct application of Proposition~\ref{prop:main}.
\end{proof}

\begin{proof}[Proof of Theorem~\ref{thm:main_randx}]
We again need to show a bound for $\rho_{x,w}$, but in addition we need to estimate $\rho_x$, as it is not part of the assumptions. The following lemma, which is proved in Appendix~\ref{sec:proof:lemma:rho_wx_randx}, provides these two bounds.

\begin{lemma}
\label{lemma:rho_wx_randx}
Suppose (A2) holds and let $\vx$ be a sequence of zero-mean iid subgaussian random variables with variance $\sigma_x^2$.  Then
\[
\frac{\rho_x}{\norm{\vx}_2^2}
\leq
\frac{L + C_\beta \sqrt{KL} \log^5(KL)}{L - \sqrt{2L \beta \log K}}
\]
and
\[
\frac{\rho_{x,w}}{\sigma_w \norm{\vx}_2}
\leq
\frac{C_\beta \sqrt{KL} \log^5 (MKL)}{\sqrt{L - \sqrt{2L \beta \log K}}}
\]
hold with probability $1-3K^{-\beta}$.
\end{lemma}

Again, the theorem follows from a direct application of Proposition~\ref{prop:main}.
\end{proof}

\section{Numerical Results}
\label{sec:numres}

We compare the classical cross-convolution (CC) method and our modification with additional subspace prior, which is the subspace-constrained cross-convolution (SCCC) method in a set of Monte Carlo simulations.

Our first tests concern  the random signal model of Theorem~\ref{thm:main_randx} with additional subspace constraints. As expected, our method SCCC, which exploits this information, significantly outperforms the original CC, which does not, see Figure~\ref{fig:cgsp}.  Specifically, the estimation error monotonically decreases (resp. increases) with $L$ and $M$ (resp. $D$).  This is consistent with the prediction in Theorem~\ref{thm:main_randx}.

Figure~\ref{fig:cgsp_pt} compares the empirical phase transition of SCCC and the least squares approach in the non-blind case (where $\vx$ is known).   As in the limit when $M$ goes to infinity, \eqref{eq:condL_randx} simplifies to $L/K \gtrsim \sqrt{D/K}$, we provide the plot in terms of the quantities $D/K$ and $L/K$.  In the case of noisy measurements, our performance measure is the 95th percentile of the estimation error, i.e., we consider the worst case, but ignore those $5\%$ of the instances, which performed worst. In  Figure~\ref{fig:cgsp_pt}, we compare this error to the non-blind case.  Our color coding uses a logarithmic scale with blue denoting the smallest and red the largest estimation error within the regime.  We observe that, unlike the non-blind case, SCCC totally fails when $D/K$ is larger than certain threshold determined by $M$.  This phenomenon clearly justifies the need to introduce a strong subspace prior to stabilize the eigenvector estimation.

\begin{figure}
  \centering
  \subfloat[]{\includegraphics[width=2.1in]{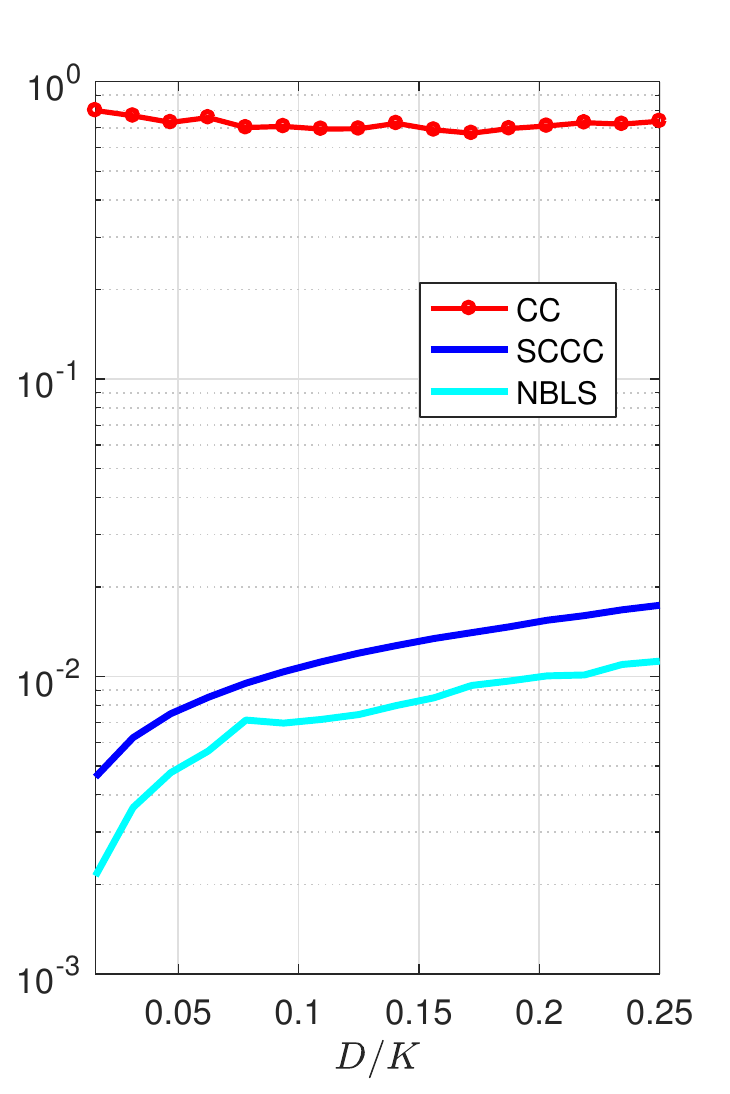}\label{fig:cgsp_v1}}
  \subfloat[]{\includegraphics[width=2.1in]{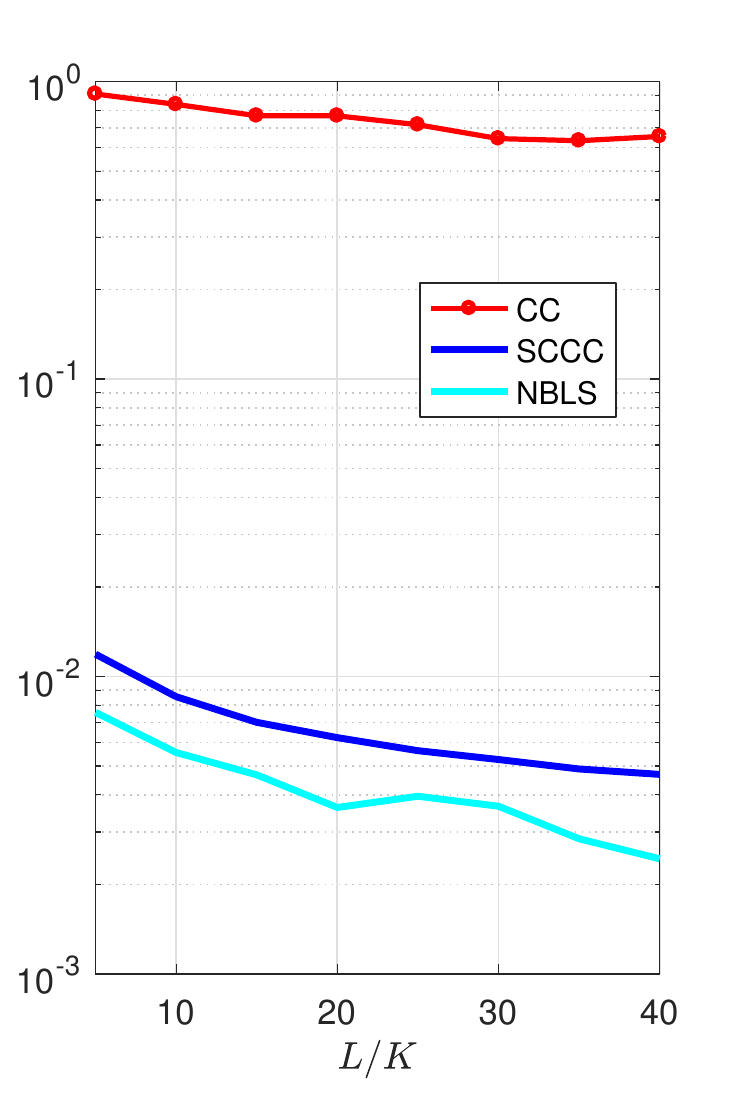}\label{fig:cgsp_v2}}
  \subfloat[]{\includegraphics[width=2.1in]{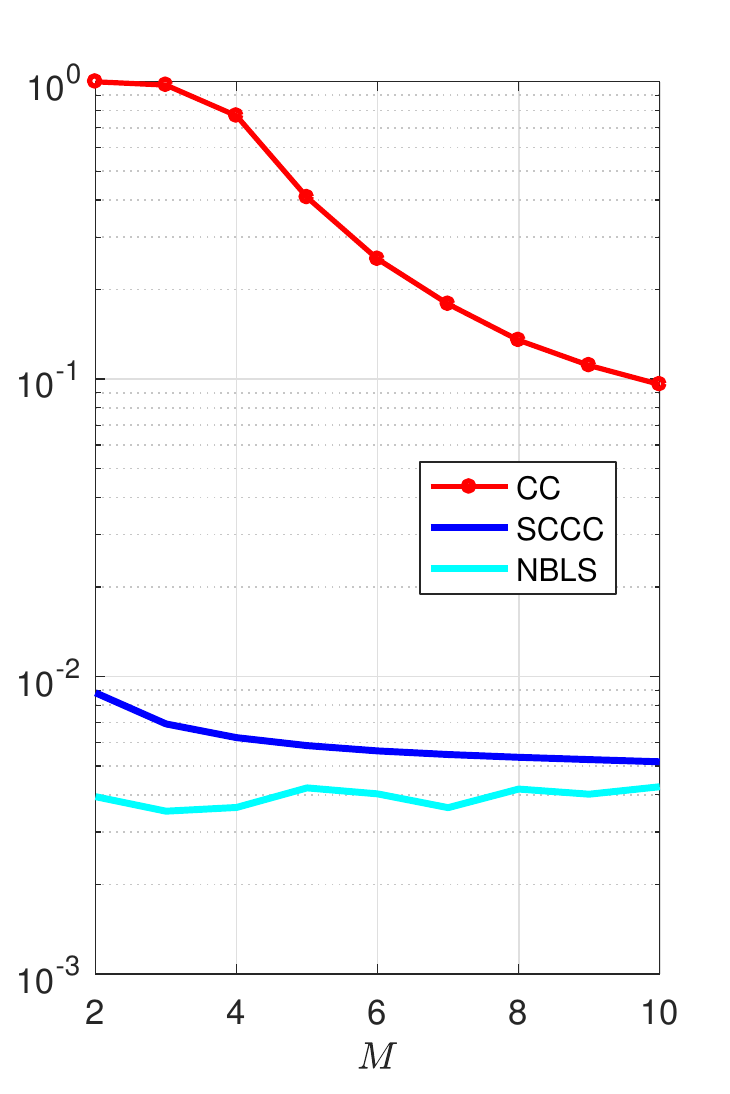}\label{fig:cgsp_v3}}
  \caption{Estimation error $\sin \angle(\widehat{\underline{\bm{h}}},\underline{\bm{h}})$ for cross-convolution (CC) and subspace-constrained cross-convolution (SCCC). (95th percentile for 1,000 trials). I.i.d. Gaussian basis. Default parameters: $K = 256$, $M = 4$, $D = 8$, $L = 20 K$, SNR = 20 dB. \protect\subref{fig:cgsp_v1} For different dimensions. \protect\subref{fig:cgsp_v2} For different observation lengths. \protect\subref{fig:cgsp_v3} For different numbers of channels.}
  \label{fig:cgsp}%
\end{figure}

\begin{figure}
  \centering
  \subfloat[]{\includegraphics[width=2.1in]{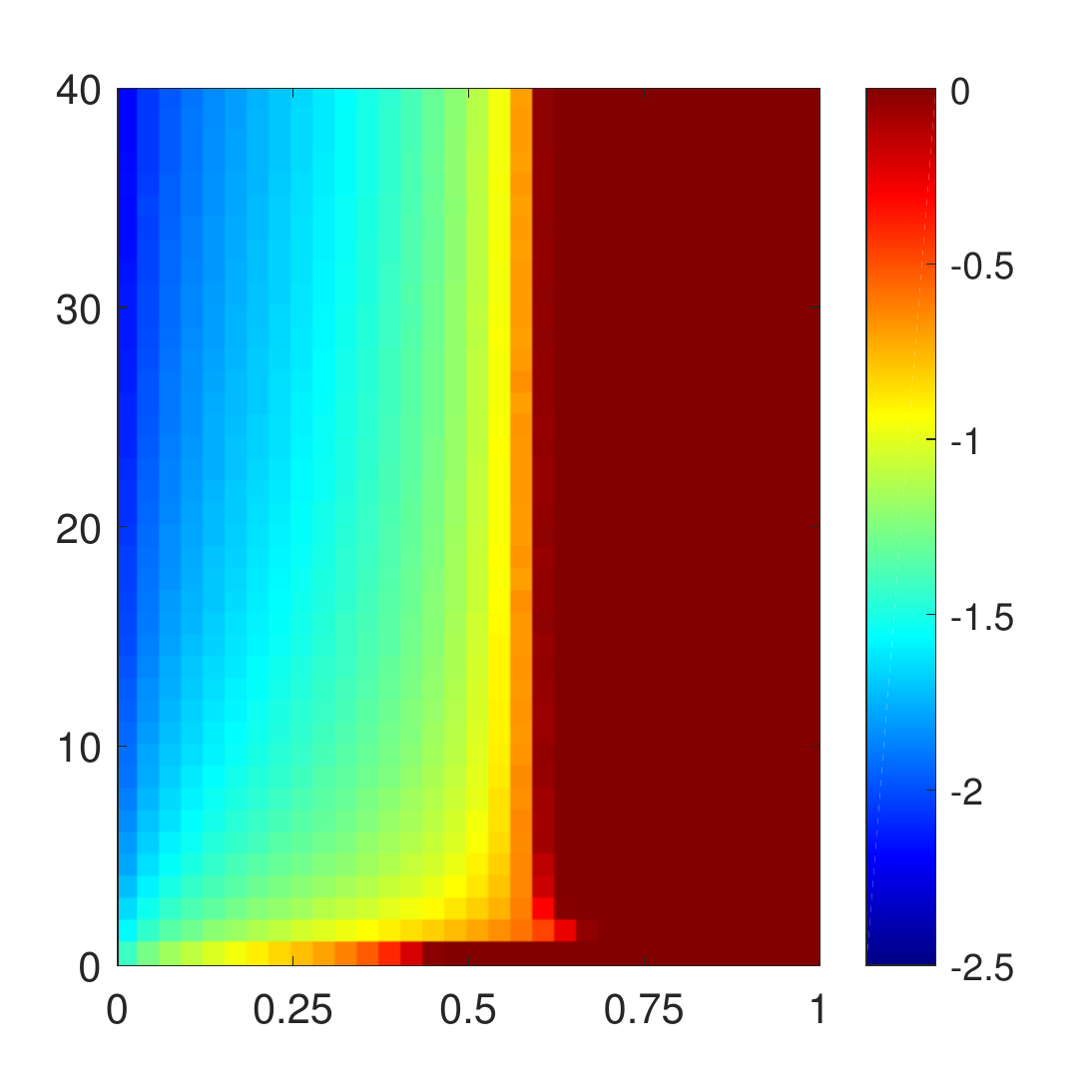}\label{fig:cgsp_v5pt_M2}}
  \subfloat[]{\includegraphics[width=2.1in]{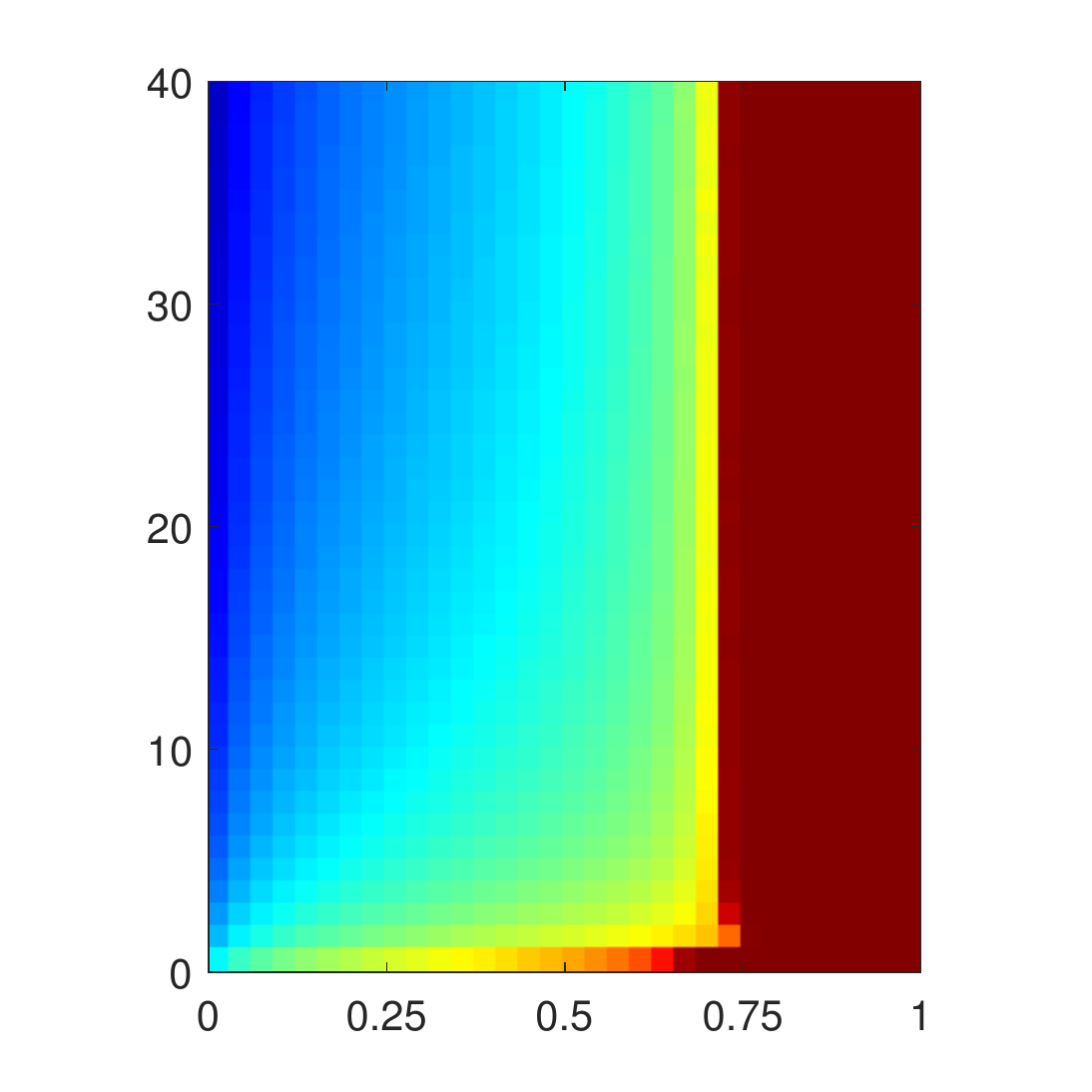}\label{fig:cgsp_v5pt_M4}}
  \subfloat[]{\includegraphics[width=2.1in]{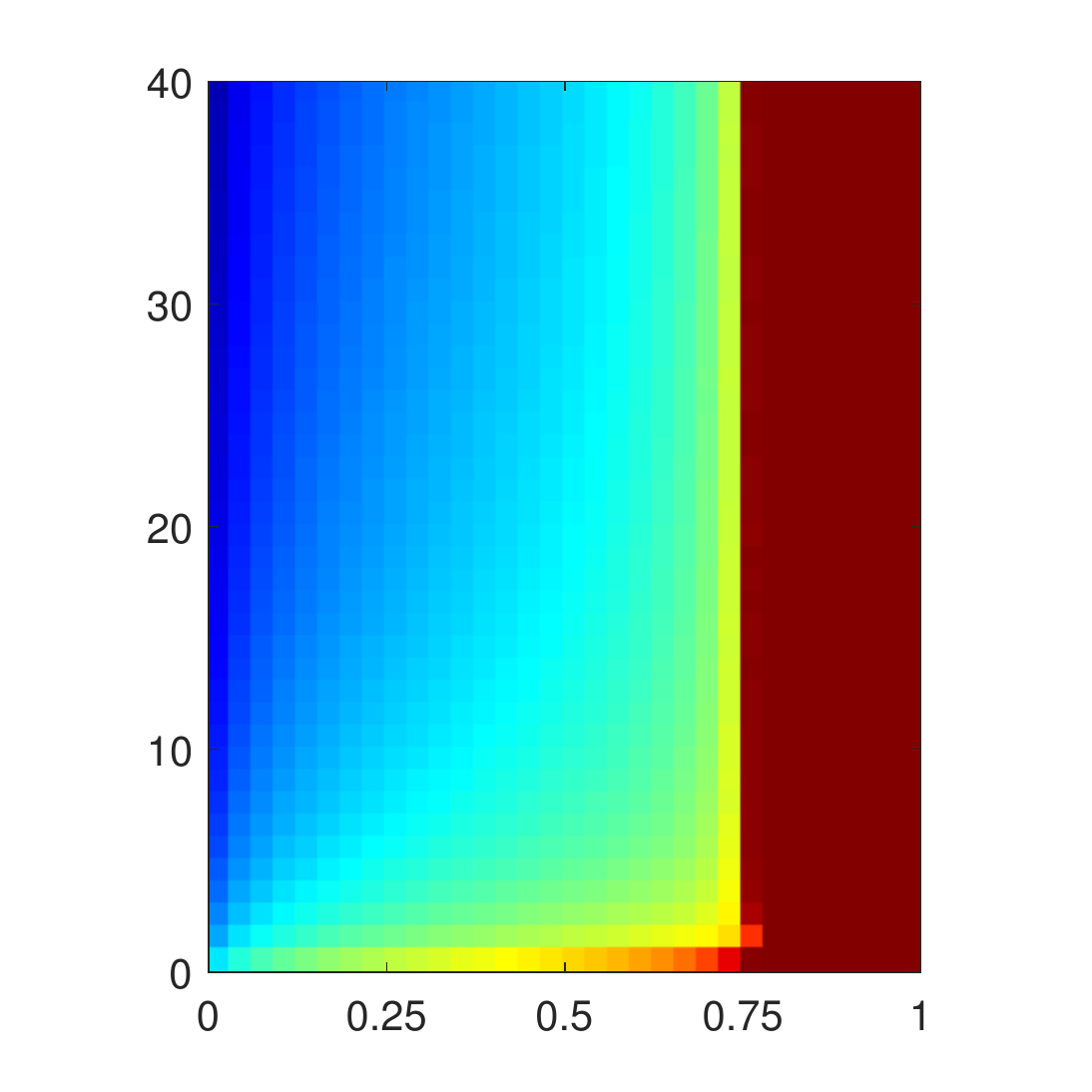}\label{fig:cgsp_v5pt_M6}}
  \\
  \subfloat[]{\includegraphics[width=2.1in]{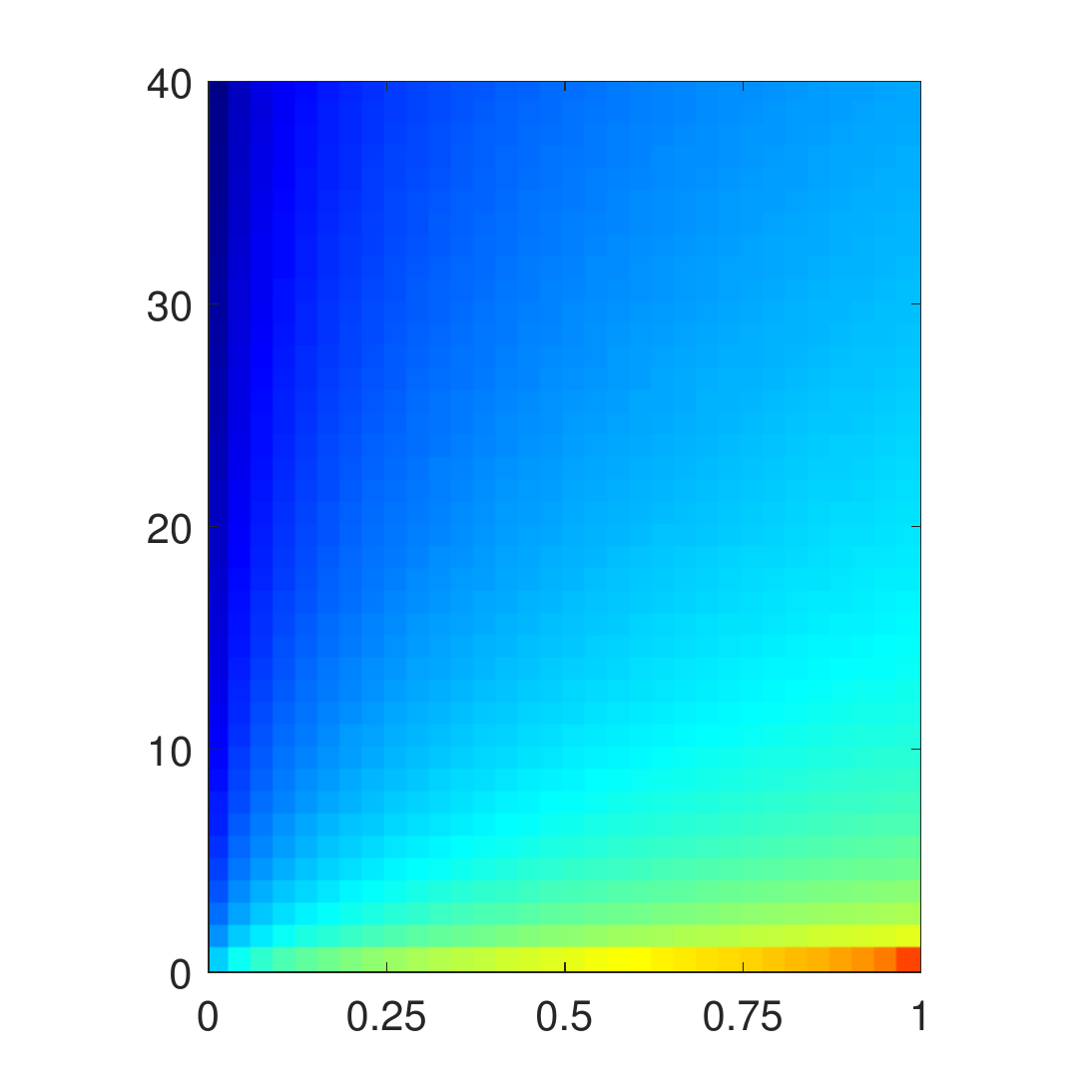}\label{fig:cgsp_v5pto_M2}}
  \subfloat[]{\includegraphics[width=2.1in]{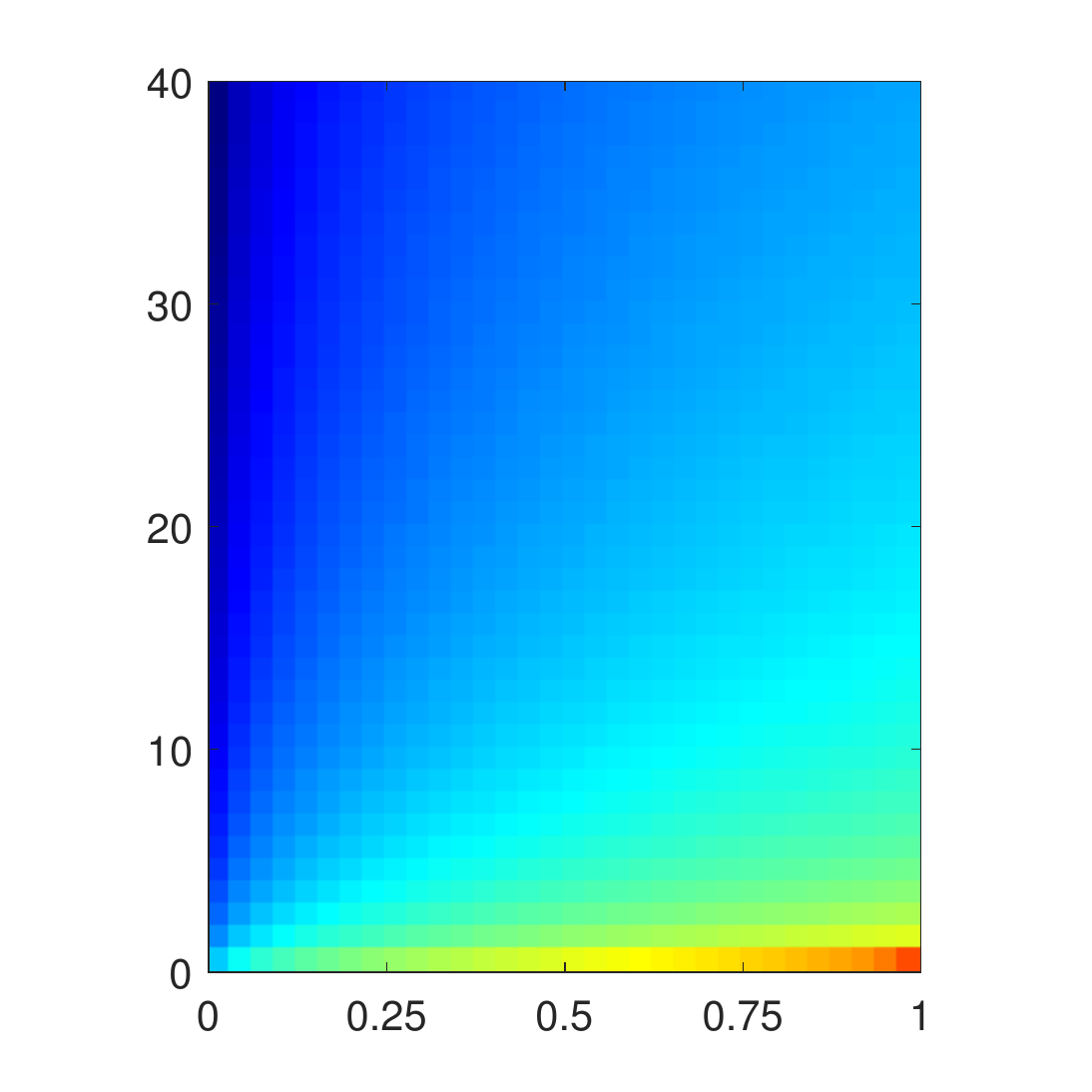}\label{fig:cgsp_v5pto_M4}}
  \subfloat[]{\includegraphics[width=2.1in]{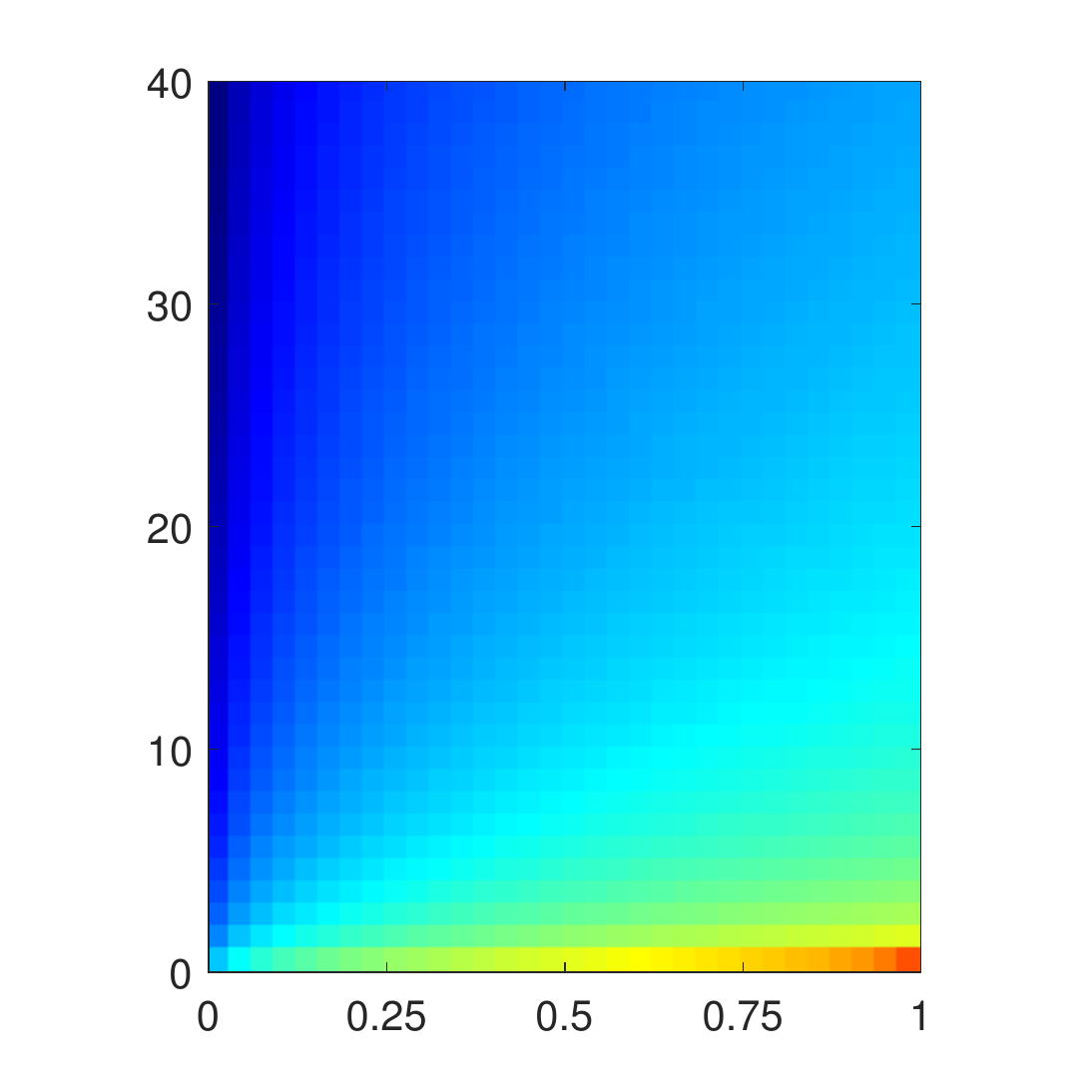}\label{fig:cgsp_v5pto_M6}}
  \caption{Estimation error for SCCC (upper row) and nonblind deconvolution with known $\vx$ (bottom row). (95th percentile of the log of the estimation error for 1,000 trials). I.i.d. Gaussian basis. $x$-axis: $D/K$. $y$-axis: $L/K$. $K = 256$, SNR = 20 dB.
  \protect\subref{fig:cgsp_v5pt_M2}, \protect\subref{fig:cgsp_v5pto_M2} $M=2$.
  \protect\subref{fig:cgsp_v5pt_M4}, \protect\subref{fig:cgsp_v5pto_M4} $M=4$.
  \protect\subref{fig:cgsp_v5pt_M6}, \protect\subref{fig:cgsp_v5pto_M6} $M=6$.
  }
  \label{fig:cgsp_pt}%
\end{figure}

Finally we study data obtained from a parametric channel impulse response model and apply SCCC under subspace model obtained empirically by principal component analysis \cite{tian2017multichannel}.  More precisely, the unknown filters are generated by sampling a known continuous function with random shifts (not necessarily on a given grid) followed by scaling with random amplitudes.  We compare our method to the classical cross-convolution (CC) method as well as to the least squares (LS) approach to a different linearized formulation \cite{balzano2007blind}, which also incorporate the same prior models on the impulse responses.
As shown in Figure~\ref{fig:uw}, SCCC outperforms CC and LS in this scenario.  Although the assumptions of Theorem~\ref{thm:main_randx} are not satisfied, similarly to the previous experiment, the estimation error for SCCC monotonically decreases with $L$ and $M$.  The other two methods did not perform satisfactorily even under a very high SNR of 80 dB.  We have already explained why the classical method fails in terms of the spectral gap.  For the least squares methods, which recovers both the input and filters simultaneously, was not successful because the system of convolution with multiple channels is highly ill-conditioned.  Even when the unknown filters are known, the corresponding system has condition number typically larger than 5,000.  This happened since the known continuous function is close to a strict band-pass filter and the unknown signal has a white spectrum.  Figure~\ref{fig:uw_v4} demonstrates that even under moderate SNRs, SCCC provides stable recovery whereas the other methods totally failed in this regime.

\begin{figure}
  \centering
  \subfloat[]{\includegraphics[width=2.1in]{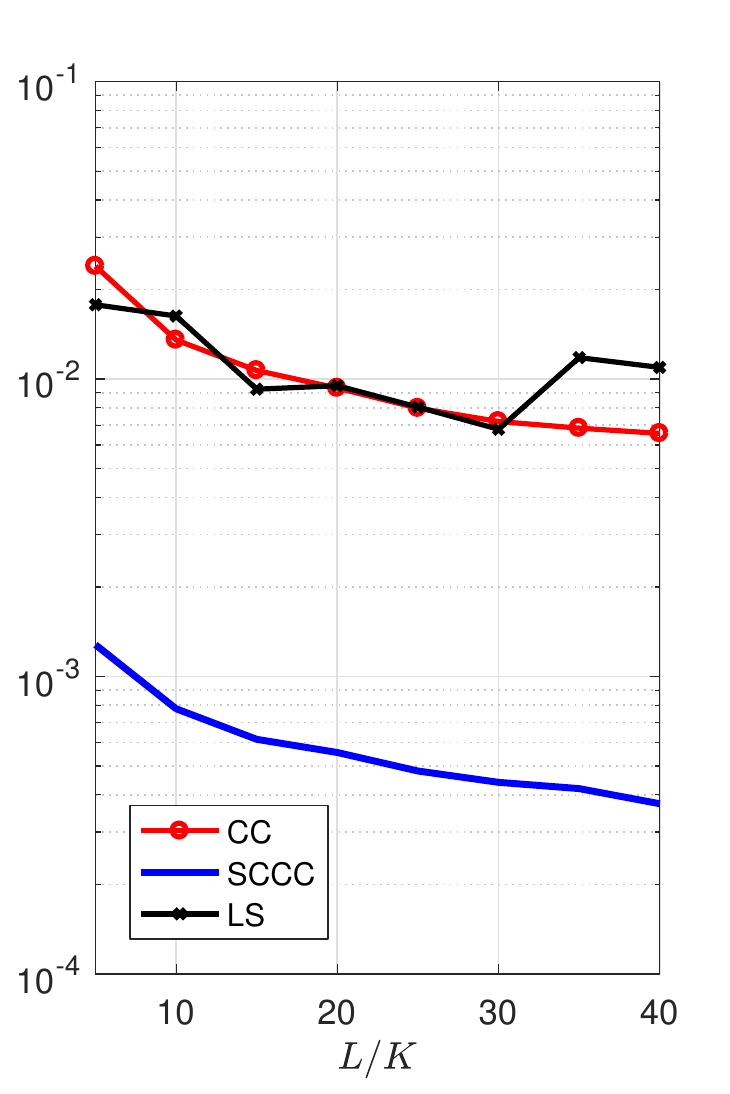}\label{fig:uw_v2}}
  \subfloat[]{\includegraphics[width=2.1in]{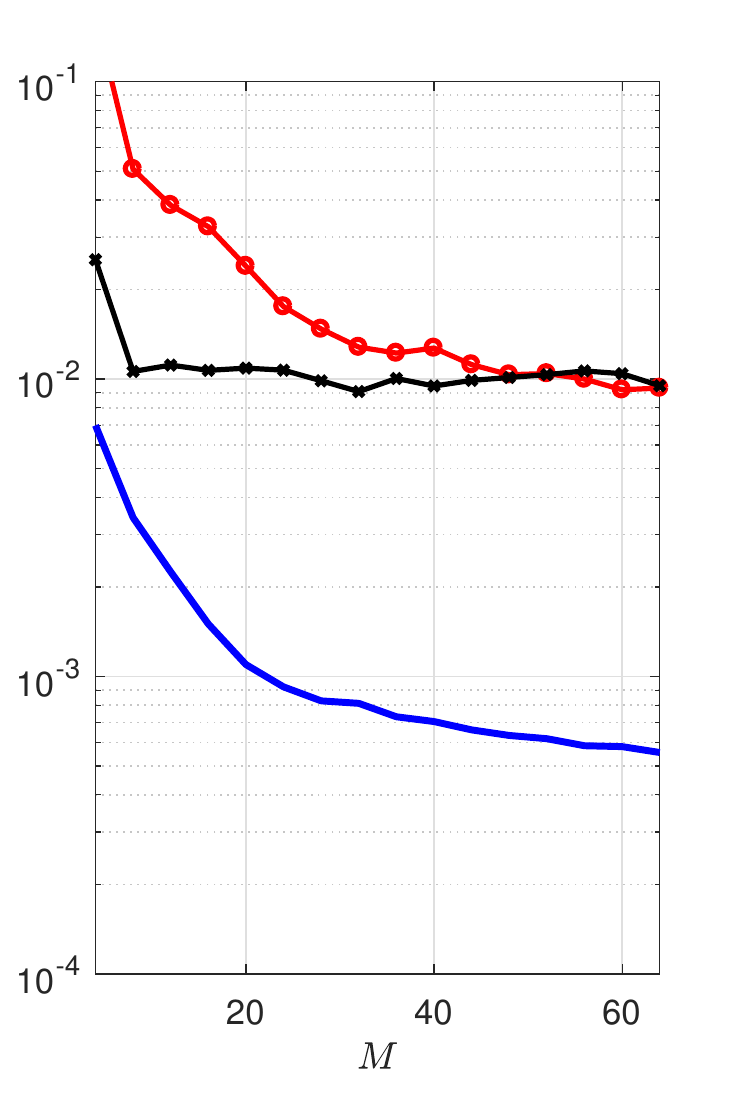}\label{fig:uw_v3}}
  \subfloat[]{\includegraphics[width=2.1in]{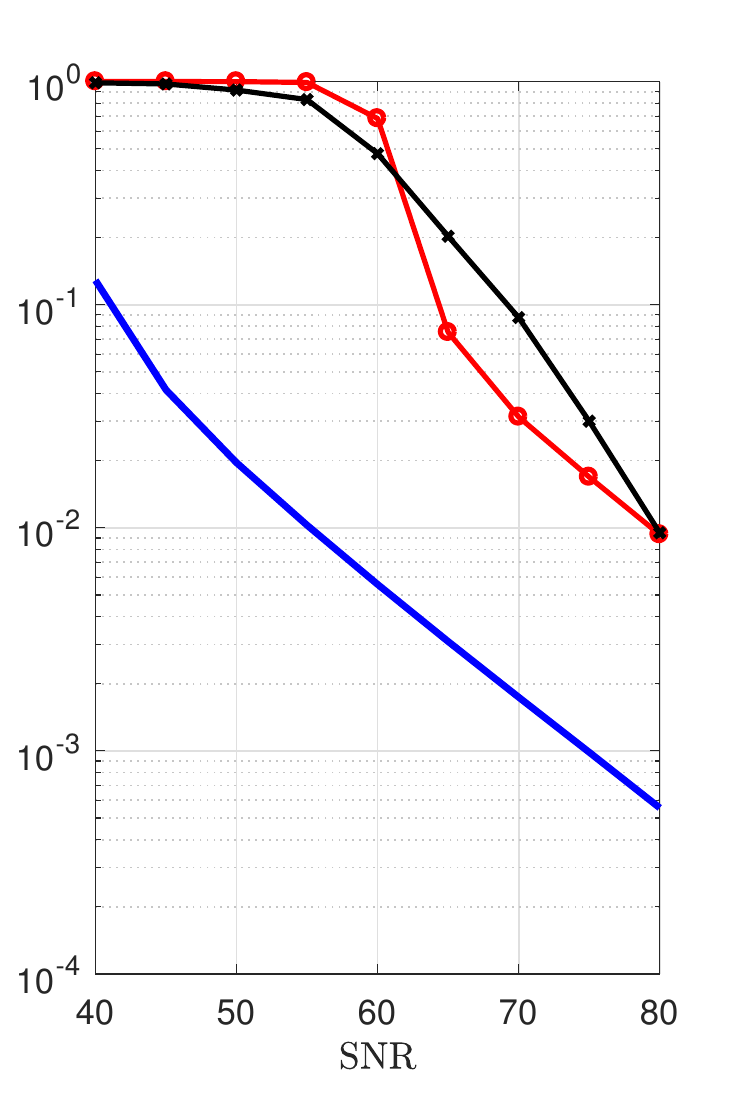}\label{fig:uw_v4}}
  \caption{Estimation error $\sin \angle(\widehat{\underline{\bm{h}}},\underline{\bm{h}})$ for cross-convolution (CC) and subspace-constrained cross-convolution (SCCC) (95th percentile for 1,000 trials). PCA basis. Default parameters: $K = 64$, $M = 64$, $D = 6$, $L = 20 K$, SNR = 80 dB. \protect\subref{fig:uw_v2} For different observation lengths. \protect\subref{fig:uw_v3} For different numbers of channels. \protect\subref{fig:uw_v4} For different SNR.}
  \label{fig:uw}%
\end{figure}

\section{Proof of Key Lemmas}
\label{sec:proof_tech_lemma}

In this section, we prove some important lemmas required in the proofs of our main theorems. For this it will be of particular importance to compute tail estimates of the spectral norms of several structured random matrices with entries given as functions of second order polynomials of Gaussian variables.  In certain cases, such bounds can be established via matrix concentration inequalities (e.g., \cite{tropp2012user,koltchinskii2011nuclear,junge2013noncommutative}).  These matrix concentration inequalities, however, require upper estimate of spectral norms of summands and high order moments, the computation of which turns out rather complicated for those matrices arising in the proofs.  Moreover, there are cases where these inequalities do not apply as the target matrix is not decomposed as a sum of independent variables.  For these reasons, we rather proceed by writing the spectral norms of these random matrices in a variational form as suprema of corresponding chaos processes.
Tail estimates for such suprema of second order chaos processes, as given in the following theorem, have been derived using chaining arguments \cite{krahmer2014suprema}.

\begin{theorem}[{Tail estimates for suprema of chaos processes \cite[Theorem~3.1]{krahmer2014suprema}}]
\label{thm:kmr}
Let $\vxi \in \cz^n$ be an $L$-subgaussian vector with $\mathbb{E}[\vxi \vxi^*] = \mId_n$, $\Delta \subset \cz^{m \times n}$, and $0 < \zeta < 1$.
There exists a constant $C(L)$ that only depends on $L$ such that
\[
\sup_{\mM \in \Delta} |\norm{\mM \vxi}_2^2 - \mathbb{E}[\norm{\mM \vxi}_2^2]|
\leq C(L) ( K_1 + K_2 \sqrt{\log(2\zeta^{-1})} + K_3 \log(2\zeta^{-1}) )
\]
holds with probability $1-\zeta$, where $K_1$, $K_2$, and $K_3$ are given by
\begin{align*}
K_1 {} & := \gamma_2(\Delta, \norm{\cdot})
\Big( \gamma_2(\Delta, \norm{\cdot}) + d_{\mathrm{F}}(\Delta) \Big)
+ d_{\mathrm{F}}(\Delta) d_{\mathrm{S}}(\Delta), \\
K_2 {} & := d_{\mathrm{S}}(\Delta) ( \gamma_2(\Delta, \norm{\cdot}) + d_{\mathrm{F}}(\Delta) ), \\
K_3 {} & := d_{\mathrm{S}}^2(\Delta).
\end{align*}
\end{theorem}

Using the polarization identity, this result on the suprema of second order chaos processes has been extended from a subgaussian quadratic form to a subgaussian bilinear form \cite{lee2015rip}.

\begin{theorem}[{A corollary of \cite[Theorem~2.3]{lee2015rip}}]
\label{thm:ip}
Let $\vxi \in \cz^n$ be an $L$-subgaussian vector with $\mathbb{E}[\vxi \vxi^*] = \mId_n$, $\Delta_2,\Delta_1 \subset \cz^{m \times n}$, $0 < \zeta < 1$, and $a > 0$.
There exists a constants $C(L)$ that only depends on $L$ such that
\begin{align*}
\sup_{\mM_1 \in \Delta_1, \mM_2 \in \Delta_2}
\big| \langle \mM_1 \vxi, \mM_2 \vxi \rangle
- \mathbb{E}[\langle \mM_1 \vxi, \mM_2 \vxi \rangle] \big|
\leq
C(L)( \widetilde{K}_1 + \widetilde{K}_2 \sqrt{\log(8\zeta^{-1})} + \widetilde{K}_3 \log(8\zeta^{-1}) ),
\end{align*}
holds with probability $1-\zeta$, where $\widetilde{K}_1$, $\widetilde{K}_2$, and $\widetilde{K}_3$ are given by
\begin{align*}
\widetilde{K}_1 {} & := \Big(a\gamma_2(\Delta_1, \norm{\cdot}) + a^{-1}\gamma_2(\Delta_2, \norm{\cdot})\Big) \Big(a\gamma_2(\Delta_1, \norm{\cdot}) + a^{-1}\gamma_2(\Delta_2, \norm{\cdot}) + a d_{\mathrm{F}}(\Delta_1) + a^{-1} d_{\mathrm{F}}(\Delta_2)\Big) \\
&+ \Big(a d_{\mathrm{F}}(\Delta_1) + a^{-1} d_{\mathrm{F}}(\Delta_2)\Big)
\Big(a d_{\mathrm{S}}(\Delta_1) + a^{-1} d_{\mathrm{S}}(\Delta_2)\Big), \\
\widetilde{K}_2 {} & := \Big(a d_{\mathrm{S}}(\Delta_1) + a^{-1} d_{\mathrm{S}}(\Delta_2)\Big) \Big( a \gamma_2(\Delta_1, \norm{\cdot}) + a^{-1} \gamma_2(\Delta_2, \norm{\cdot}) + a d_{\mathrm{F}}(\Delta_1) + a^{-1} d_{\mathrm{F}}(\Delta_2)\Big), \\
\widetilde{K}_3 {} & := \Big(a d_{\mathrm{S}}(\Delta_1) + a^{-1} d_{\mathrm{S}}(\Delta_2)\Big)^2.
\end{align*}
\end{theorem}

A special case of Theorem~\ref{thm:ip} where $a = 1$ was shown in \cite[Theorem~2.3]{lee2015rip}.
Note that the bilinear form satisfies
\[
\langle \mM_1 \vxi, \mM_2 \vxi \rangle
= \langle a \mM_1 \vxi, a^{-1} \mM_2 \vxi \rangle,
\quad \forall a > 0.
\]
Moreover, the $\gamma_2$ functional and the radii with respect to the Frobenius and spectral norms are all 1-homogeneous functions.
Therefore, Theorem~\ref{thm:ip} is a direct consequence of \cite[Theorem~2.3]{lee2015rip}.

Since $a > 0$ in Theorem~\ref{thm:ip} is arbitrary, one can minimize the tail estimate over $a > 0$.

\subsection{Proof of Lemma~\ref{lemma:gap}}
\label{sec:proof:lemma:gap}

Note that $\mPhi^* \mY_{\mathrm{s}}^* \mY_{\mathrm{s}} \mPhi \in \mathbb{C}^{MD \times MD}$ is an $M$-by-$M$ block matrix, where the size of each block is $D$-by-$D$. Then it follows from the special structure of $\mY$ (with each row consisting only of some $\mT_{\vy_k}$ in position $j$ and the corresponding $-\mT_{\vy_j}$ in position $k$) that the $(m,m)$th diagonal block of $\mPhi^* \mY_{\mathrm{s}}^* \mY_{\mathrm{s}} \mPhi$ is given by
\begin{equation}
\label{eq:diagblock}
\sum_{\begin{subarray}{c} m'=1 \\ m' \neq m \end{subarray}}^M
\mPhi_m^* \mT_{\vx \conv \vh_{m'}}^* \mT_{\vx \conv \vh_{m'}} \mPhi_m
=
\sum_{\begin{subarray}{c} m'=1 \\ m' \neq m \end{subarray}}^M
\widetilde{\mPhi}_m^*
\mC_{\widetilde{\mPhi}_{m'} \vu_{m'}}^* \mC_{\vx}^* \mC_{\vx} \mC_{\widetilde{\mPhi}_{m'} \vu_{m'}}
\widetilde{\mPhi}_m,
\end{equation}
where $\widetilde{\mPhi}_m = \mS^* \mPhi_m$ for $m=1,\dots,M$.
Similarly, the $(m,m')$th off-diagonal block of $\mPhi^* \mY_{\mathrm{s}}^* \mY_{\mathrm{s}} \mPhi$ for $m \neq m'$ is given by
\begin{equation}
\label{eq:offdiagblock}
- \mPhi_m^* \mT_{\vx \conv \vh_{m'}}^* \mT_{\vx \conv \vh_m} \mPhi_{m'}
= - \widetilde{\mPhi}_m^*
\mC_{\widetilde{\mPhi}_{m'} \vu_{m'}}^* \mC_{\vx}^* \mC_{\vx} \mC_{\widetilde{\mPhi}_m \vu_m}
\widetilde{\mPhi}_{m'}.
\end{equation}

By Lemma~\ref{lemma:expectation3}, the expectation of the $(m,m)$th diagonal block is given by
\begin{align*}
\mathbb{E}
\Big[\sum_{\begin{subarray}{c} m'=1 \\ m' \neq m \end{subarray}}^M
\widetilde{\mPhi}_m^* \mC_{\widetilde{\mPhi}_{m'} \vu_{m'}}^* \mC_{\vx}^* \mC_{\vx} \mC_{\widetilde{\mPhi}_{m'} \vu_{m'}} \widetilde{\mPhi}_m \Big]
= \sum_{\begin{subarray}{c} m'=1 \\ m' \neq m \end{subarray}}^M \norm{\vu_{m'}}_2^2 K^2 \norm{\vx}_2^2 \mId_D
\end{align*}
and a similar calculation yields that the expectation of the $(m,m')$th diagonal block for $m' \neq m$ is given by
\begin{align*}
- \mathbb{E}[\widetilde{\mPhi}_m^* \mC_{\widetilde{\mPhi}_{m'} \vu_{m'}}^* \mC_{\vx}^* \mC_{\vx} \mC_{\widetilde{\mPhi}_m \vu_m} \widetilde{\mPhi}_{m'}]
= - K^2 \norm{\vx}_2^2 \vu_m \vu_{m'}^*.
\end{align*}

Collecting the above expectations, we obtain that $\mathbb{E}[\mPhi^* \mY_{\mathrm{s}}^* \mY_{\mathrm{s}} \mPhi]$ satisfies
\begin{align*}
\frac{\mathbb{E}[\mPhi^* \mY_{\mathrm{s}}^* \mY_{\mathrm{s}} \mPhi]}{K^2 \norm{\vx}_2^2} = \norm{\vu}_2^2 \mP_{\vu^\perp} - \mUpsilon,
\end{align*}
where
\[
\mUpsilon
= \sum_{m=1}^M \ve_m \ve_m^* \otimes \norm{\vu_m}_2^2 \mP_{\vu_m^\perp},
\]
where $\otimes$ denotes the Kronecker product and $\ve_1,\dots,\ve_M$ denote the standard basis vectors in $\mathbb{R}^M$.

Since $\mUpsilon \vu = \vzero$ and hence $\vu$ is in the kernel of $\mathbb{E}[\mPhi^* \mY_{\mathrm{s}}^* \mY_{\mathrm{s}} \mPhi]$, the gap between its two smallest eigenvalues is given by its smallest non-zero eigenvalue, which, using that $\norm{\mUpsilon} \leq \max_{1 \leq m \leq M} \norm{\vu_m}_2^2$, is estimated from below by
\begin{equation}
\label{eq:gap}
K^2 \norm{\vx}_2^2 \norm{\vu}_2^2 \Big( 1 - \frac{\max_{1\leq m\leq M} \norm{\vu_m}_2^2}{\norm{\vu}_2^2} \Big)
\geq
\frac{K^2 \norm{\vx}_2^2 \norm{\vu}_2^2}{2}.
\end{equation}
Here the last inequality follows from our assumption that
\[
\mu = \max_{1\leq m\leq M} \frac{\sqrt{M} \norm{\vu_m}_2}{\norm{\vu}_2}
\leq \frac{\sqrt{M}}{2}.
\]

\subsection{Proof of Lemma~\ref{lemma:Es}}
\label{sec:proof:lemma:Es}

We use the following notation in the proof. For $\vu = [\vu_1^\transpose,\dots,\vu_M^\transpose]^\transpose$, where $\vu_m \in \mathbb{C}^D$ for $k=1,\dots,M$, we define a block $(p,q)$ norm by
\begin{equation}
\label{eq:defblknorm}
\norm{\vu}_{p,q} =
\begin{cases}
\displaystyle \Big( \sum_{m=1}^M \norm{\vu_m}_p^q \Big)^{1/q} & 1 \leq q < \infty, \\
\displaystyle \max_{1 \leq m \leq M} \norm{\vu_m}_p & q = \infty.
\end{cases}
\end{equation}

By \eqref{eq:diagblock} and \eqref{eq:offdiagblock}, $\mPhi^* \mY_{\mathrm{s}}^* \mY_{\mathrm{s}} \mPhi - \mathbb{E}[\mPhi^* \mY_{\mathrm{s}}^* \mY_{\mathrm{s}} \mPhi]$ is rewritten as the sum of its diagonal block portion $\text{(a)}$ and the off-diagonal block portion $\text{(b)}$, where
\begin{equation}
\label{eqn_dbl_x}
\begin{aligned}
\text{(a)} &=
\sum_{m=1}^M \ve_m \ve_m^* \otimes
\Big(
\sum_{\begin{subarray}{c} m'=1 \\ m' \neq m \end{subarray}}^M
\widetilde{\mPhi}_m^* \mC_{\widetilde{\mPhi}_{m'} \vu_{m'}}^* \mC_{\vx}^* \mC_{\vx} \mC_{\widetilde{\mPhi}_{m'} \vu_{m'}} \widetilde{\mPhi}_m
- \mathbb{E}[\widetilde{\mPhi}_m^* \mC_{\widetilde{\mPhi}_{m'} \vu_{m'}}^* \mC_{\vx}^* \mC_{\vx} \mC_{\widetilde{\mPhi}_{m'} \vu_{m'}} \widetilde{\mPhi}_m]
\Big), \\
\text{(b)} &= -
\sum_{m=1}^M \sum_{\begin{subarray}{c} m'=1 \\ m' \neq m \end{subarray}}^M
\ve_m \ve_{m'}^* \otimes
(
\widetilde{\mPhi}_m^* \mC_{\widetilde{\mPhi}_{m'} \vu_{m'}}^* \mC_{\vx}^* \mC_{\vx} \mC_{\widetilde{\mPhi}_m \vu_m} \widetilde{\mPhi}_{m'}
- \mathbb{E}[\widetilde{\mPhi}_m^* \mC_{\widetilde{\mPhi}_{m'} \vu_{m'}}^* \mC_{\vx}^* \mC_{\vx} \mC_{\widetilde{\mPhi}_m \vu_m} \widetilde{\mPhi}_{m'}]
).
\end{aligned}
\end{equation}
Therefore, by the triangle inequality, we have
\[
\norm{\mPhi^* \mY_{\mathrm{s}}^* \mY_{\mathrm{s}} \mPhi - \mathbb{E}[\mPhi^* \mY_{\mathrm{s}}^* \mY_{\mathrm{s}} \mPhi]}
\leq \norm{\text{(a)}} + \norm{\text{(b)}}.
\]
It remains to compute tail estimates for the spectral norms of $\text{(a)}$ and $\text{(b)}$.

\noindent\textbf{Diagonal block portion:} The $(m,m)$th block of $\text{(a)}$ is written as $\mPhi_m^* \mUpsilon_{\mathrm{s},m} \mPhi_m$, where
\[
\mUpsilon_{\mathrm{s},m}
= \sum_{\begin{subarray}{c} m'=1 \\ m' \neq m \end{subarray}}^M
\mS \mC_{\widetilde{\mPhi}_{m'} \vu_{m'}}^* \mC_{\vx}^* \mC_{\vx} \mC_{\widetilde{\mPhi}_{m'} \vu_{m'}} \mS^*.
\]
Then due to the block diagonal structure, we have
\[
\norm{\text{(a)}}
= \max_{1\leq m\leq M} \norm{\mPhi_m^* \mUpsilon_{\mathrm{s},m} \mPhi_m - \mathbb{E}[\mPhi_m^* \mUpsilon_{\mathrm{s},m} \mPhi_m]
}.
\]

Since $\mPhi_m$ and $\mUpsilon_{\mathrm{s},m}$ are independent, $\mathbb{E}[\mPhi_m^* \mUpsilon_{\mathrm{s},m} \mPhi_m]$ is rewritten as
\[
\mathbb{E}[\mPhi_m^* \mUpsilon_{\mathrm{s},m} \mPhi_m]
= \mathbb{E}_{\mPhi_m}[\mPhi_m^* \mathbb{E}_{\{\mPhi_{m'}\}_{m'\neq m}}[\mUpsilon_{\mathrm{s},m}] \mPhi_m]
= \mathbb{E}_{\{\mPhi_{m'}\}_{m'\neq m}}[\mathrm{tr}(\mUpsilon_{\mathrm{s},m}) \mId_D].
\]
Therefore, we have
\begin{equation}
\label{eq:blockXis1}
\begin{aligned}
\mPhi_m^* \mUpsilon_{\mathrm{s},m} \mPhi_m
- \mathbb{E}[\mPhi_m^* \mUpsilon_{\mathrm{s},m} \mPhi_m]
=
\underbrace{
\mPhi_m^* \mUpsilon_{\mathrm{s},m} \mPhi_m - \mathrm{tr}(\mUpsilon_{\mathrm{s},m}) \mId_D
}_{\text{(c)}}
+
\underbrace{
(\mathrm{tr}(\mUpsilon_{\mathrm{s},m})
- \mathbb{E}[\mathrm{tr}(\mUpsilon_{\mathrm{s},m})]) \mId_D
}_{\text{(d)}}.
\end{aligned}
\end{equation}
We will compute tail estimates of the spectral norms of $\text{(c)}$ and $\text{(d)}$ separately and then combine them using the triangle inequality.
First, we compute a tail bound of $\norm{\text{(c)}}$ using the following lemmas, which are proved in Appendices~\ref{sec:proof:lemma:gaussquadasym} and \ref{sec:proof:lemma:Upsilon1m}.

\begin{lemma}
\label{lemma:gaussquadasym}
Let $\mPsi = [\vpsi_1,\dots,\vpsi_D] \in \mathbb{C}^{K \times D}$ satisfy that $\mathrm{vec}(\mPsi)$ follows $\mathcal{CN}(\vzero_{KD,1},\mId_{KD})$, where $\mathrm{vec}(\mPsi) =[\vpsi_1^\transpose,\dots,\vpsi_D^\transpose]^\transpose$.
Then
\begin{align*}
& \norm{
\mPsi^* \mA \mPsi
- \mathbb{E}[\mPsi^* \mA \mPsi]
}
\leq C \norm{\mA} \sqrt{KD} \log(8 \zeta^{-1})
\end{align*}
holds with probability $1-\zeta$.
\end{lemma}

\begin{lemma}
\label{lemma:Upsilon1m}
Suppose that (A1) holds.
For any $\beta \in \mathbb{N}$, there exist a numerical constant $\alpha \in \mathbb{N}$ and a constant $C(\beta)$ that depends only on $\beta$ such that
\begin{align*}
\max_{1 \leq m \leq M} \norm{\mUpsilon_{\mathrm{s},m} - \mathbb{E}[\mUpsilon_{\mathrm{s},m}]}
& \leq C(\beta) \rho_x K \norm{\vu}_{2,\infty} \norm{\vu}_2 \log^\alpha (MKL)
\end{align*}
holds with probability $1-K^{-\beta}$.
\end{lemma}

Lemma~\ref{lemma:gaussquadasym} is a direct consequence of the theory of suprema of second order chaos processes \cite{krahmer2014suprema,lee2015rip}. By Lemma~\ref{lemma:gaussquadasym}, conditioned on $\mUpsilon_{\mathrm{s},m}$,
\[
\norm{\text{(c)}} \leq C_1 \norm{\mUpsilon_{\mathrm{s},m}} \sqrt{KD} (\log M + \beta \log K)
\]
holds with probability $1 - M^{-1} K^{-\beta}$.
Then by Lemmas~\ref{lemma:Upsilon1m} and \ref{lemma:expectation1} with the triangle inequality, it follows that
\[
\norm{\text{(c)}} \leq C(\beta) (\rho_x + \norm{\vx}_2^2) K^{3/2} \sqrt{D} \norm{\vu}_2^2 \log^2 (MK).
\]
holds with probability $1 - M^{-1} K^{-\beta}$.

Next we consider \text{(d)}. Note that
\[
\mathrm{tr}(\mS \mC_{\widetilde{\mPhi}_{m'} \vu_{m'}}^* \mC_{\vx}^* \mC_{\vx} \mC_{\widetilde{\mPhi}_{m'} \vu_{m'}} \mS^*)
= K \norm{\mC_{\vx} \mS^* \mPhi_{m'} \vu_{m'}}_2^2.
\]
Therefore, the spectral norm of \text{(d)} is rewritten as
\begin{align*}
\norm{\text{(d)}}
&=
\Big|
\sum_{\begin{subarray}{c} m'=1 \\ m' \neq m \end{subarray}}^M
K \norm{\mC_{\vx} \mS^* \mPhi_{m'} \vu_{m'}}_2^2
- \mathbb{E}[K \norm{\mC_{\vx} \mS^* \mPhi_{m'} \vu_{m'}}_2^2]
\Big| \\
&=
K \Big|
\sum_{\begin{subarray}{c} m'=1 \\ m' \neq m \end{subarray}}^M
\norm{(\vu_{m'}^\transpose \otimes \mC_{\vx} \mS^*) \vphi_{m'}}_2^2
- \mathbb{E}[\norm{(\vu_{m'}^\transpose \otimes \mC_{\vx} \mS^*) \vphi_{m'}}_2^2]
\Big| \\
&=
K \Big|
\Big\|
\underbrace{
\sum_{\begin{subarray}{c} m'=1 \\ m' \neq m \end{subarray}}^M
(\ve_{m'} \ve_{m'}^* \otimes \vu_{m'}^\transpose \otimes \mC_{\vx} \mS^*)
}_{\mA_m}
\vphi
\Big\|_2^2
- \mathbb{E}
\Big[\Big\|
\sum_{\begin{subarray}{c} m'=1 \\ m' \neq m \end{subarray}}^M
(\ve_{m'} \ve_{m'}^* \otimes \vu_{m'}^\transpose \otimes \mC_{\vx} \mS^*) \vphi
\Big\|_2^2\Big]
\Big|,
\end{align*}
where $\vphi_m = \mathrm{vec}(\mPhi_m)$ for $m=1,\dots,M$ and $\vphi = [\vphi_1^\transpose,\dots,\vphi_M^\transpose]^\transpose$.
In fact, we are computing a tail bound of the Gaussian quadratic form $\norm{\mA_m \vphi}_2^2$, which can be done by the Hanson-Wright inequality,  Lemma~\ref{lemma:hansonwright_ori}. Observing that the block diagonal matrix $\mU$ with blocks $\vu_m$ satisfies $\|\mU^*\mU\|_F\leq \norm{\vu}_{2,4}^2$, we obtain
\begin{align*}
\norm{\text{(d)}}
&= K | \norm{\mA_m \vphi}_2^2 - \mathbb{E}[\norm{\mA_m \vphi}_2^2] | \\
&\leq C_3 K (\norm{\mA_m^* \mA_m}_{\mathrm{F}} \vee \norm{\mA_m}^2) (2 \log M + 2 \beta \log K) \\
&\leq C_3 K ( \sqrt{K} \norm{\vu}_{2,4}^2 \vee \norm{\vu}_{2,\infty}^2) \norm{\mS \mC_{\vx}^* \mC_{\vx} \mS^*} (2 \log M + 2 \beta \log K) \\
&\leq C_4 K^{3/2} \norm{\vu}_{2,4}^2 \rho_x (2 \log M + 2 \beta \log K)
\end{align*}
holds with probability $1-M^{-1}K^{-\beta}$. Note that the tail bound of $\text{(c)}$ dominates that of \text{(d)}.

Collecting the above estimates, it follows that
\begin{align*}
\norm{\text{(a)}}
\leq C(\beta) \rho_x K^{3/2} \sqrt{D} \norm{\vu}_2^2 \log^\alpha(MKL)
\end{align*}
holds with probability $1-K^{-\beta}$.

If we normalize with the spectral gap given in \eqref{eq:gap}, then the relative perturbation due to $\text{(a)}$ is upper bounded by
\begin{equation}
\label{eq:bnd_snEs_part1}
\frac{\norm{\text{(a)}}}{K^2 \norm{\vx}_2^2 \norm{\vu}_2^2}
\leq
C(\beta) \log^\alpha(MKL) \sqrt{\frac{D}{K}}
\end{equation}
with probability $1-K^{-\beta}$.

\noindent\textbf{Off-diagonal block portion:}
Unlike the diagonal block portion $\text{(a)}$, the off-diagonal block portion $\text{(b)}$ does not have a block diagonal structure and computing its tail bound is more involved.

To restrict the convolution of two short vectors of length $K$ to its support, we introduce $\breve{\mS} \in \mathbb{R}^{(2K-1) \times L}$ defined by
\[
\breve{\mS} =
\begin{bmatrix}
\bm{0}_{K-1,L-K+1} & \mId_{K-1} \\
\mId_{K} & \bm{0}_{K,L-K}
\end{bmatrix}.
\]
Then we have
\begin{equation}
\label{eq:SSbreve}
\mS \breve{\mS}^* \breve{\mS} = \mS
\end{equation}
and
\begin{equation}
\label{eq:supp2conv}
\mC_{\widetilde{\mPhi}_m \vu_m}^* \widetilde{\mPhi}_m
= \breve{\mS}^* \breve{\mS} \mC_{\widetilde{\mPhi}_m \vu_m}^* \widetilde{\mPhi}_m, \quad \forall m=1,\dots,M.
\end{equation}

Due to the commutativity of product of two circulant matrices, \eqref{eq:SSbreve}, and \eqref{eq:supp2conv}, we can rewrite $\text{(b)}$ as
\[
\text{(b)} = -
\sum_{m=1}^M \sum_{\begin{subarray}{c} m'=1 \\ m' \neq m \end{subarray}}^M
\ve_m \ve_{m'}^* \otimes
\Big(
(\mC_{\vx} \breve{\mS}^* \mZ_m)^* \mC_{\vx} \breve{\mS}^* \mZ_{m'}
- \mathbb{E}[(\mC_{\vx} \breve{\mS}^* \mZ_m)^* \mC_{\vx} \breve{\mS}^* \mZ_{m'}]
\Big),
\]
where
\[
\mZ_m = \breve{\mS} \mC_{\widetilde{\mPhi}_m \vu_m}^* \widetilde{\mPhi}_m, \quad m=1,\dots,M.
\]

Note that the summation in $\text{(b)}$ runs over all distinct pairs $(m,m')$ with $m \neq m'$. Our main trick here is to add and subtract the terms corresponding to pairs $(m,m)$ for $m=1,\dots,M$. This ends up with a diagonal sum and a full summation over all pairs $(m,m')$. The resulting full summation term now provides a nice factorization, which leads to an analysis using the techniques for the second-order chaos processes.

Indeed, since the $\mPhi_m$'s are independent, we have
\[
\mathbb{E}[\mC_{\vx} \breve{\mS}^* \mZ_m]^*
\mathbb{E}[\mC_{\vx} \breve{\mS}^* \mZ_{m'}]
- \mathbb{E}[(\mC_{\vx} \breve{\mS}^* \mZ_m)^* \mC_{\vx} \breve{\mS}^* \mZ_{m'}]
= \bm{0}_{D,D}, \quad \forall m \neq m'.
\]
Therefore, $\text{(b)}$ is decomposed as $\text{(b)} = \text{(e)} - \text{(f)}$, where
\begin{align*}
\text{(e)}
&= \sum_{m=1}^M \ve_m \ve_m^* \otimes ((\mC_{\vx} \breve{\mS}^* \mZ_m)^* \mC_{\vx} \breve{\mS}^* \mZ_m - \mathbb{E}[\mC_{\vx} \breve{\mS}^* \mZ_m]^* \mathbb{E}[\mC_{\vx} \breve{\mS}^* \mZ_m]), \\
\text{(f)}
&= \sum_{m,m'=1}^M \ve_m \ve_{m'}^* \otimes ((\mC_{\vx} \breve{\mS}^* \mZ_m)^* \mC_{\vx} \breve{\mS}^* \mZ_{m'} - \mathbb{E}[\mC_{\vx} \breve{\mS}^* \mZ_m]^* \mathbb{E}[\mC_{\vx} \breve{\mS}^* \mZ_{m'}]).
\end{align*}

Note that $(\mC_{\vx} \breve{\mS}^* \mZ_m)^* \mC_{\vx} \breve{\mS}^* \mZ_{m'} - \mathbb{E}[\mC_{\vx} \breve{\mS}^* \mZ_m]^* \mathbb{E}[\mC_{\vx} \breve{\mS}^* \mZ_{m'}]$ is decomposed as
\begin{equation}
\label{eq:decomp_Z}
\begin{aligned}
& (\mC_{\vx} \breve{\mS}^* \mZ_m - \mathbb{E}[\mC_{\vx} \breve{\mS}^* \mZ_m])^* (\mC_{\vx} \breve{\mS}^* \mZ_{m'} - \mathbb{E}[\mC_{\vx} \breve{\mS}^* \mZ_{m'}]) \\
&+ \mathbb{E}[\mC_{\vx} \breve{\mS}^* \mZ_m^*] (\mC_{\vx} \breve{\mS}^* \mZ_{m'} - \mathbb{E}[\mC_{\vx} \breve{\mS}^* \mZ_{m'}]) \\
&+ (\mC_{\vx} \breve{\mS}^* \mZ_m - \mathbb{E}[\mC_{\vx} \breve{\mS}^* \mZ_m])^* \mathbb{E}[\mC_{\vx} \breve{\mS}^* \mZ_{m'}].
\end{aligned}
\end{equation}
Therefore, the spectral norm of $\text{(e)}$, which corresponds to the extra diagonal term, is upper-bounded by
\begin{align*}
\norm{\text{(e)}}
& \leq
\Big( \max_{1 \leq m\leq M} \norm{\mC_{\vx} \breve{\mS}^* \mZ_m - \mathbb{E}[\mC_{\vx} \breve{\mS}^* \mZ_m]} \Big)^2 \\
& + 2 \Big( \max_{1 \leq m\leq M} \norm{\mC_{\vx} \breve{\mS}^* \mZ_m - \mathbb{E}[\mC_{\vx} \breve{\mS}^* \mZ_m]} \Big)
\Big( \max_{1 \leq m\leq M} \norm{\mathbb{E}[\mC_{\vx} \breve{\mS}^* \mZ_m]} \Big).
\end{align*}

By Lemma~\ref{lemma:expectation2}, we have
\[
\norm{\mathbb{E}[\mC_{\vx} \breve{\mS}^* \mZ_m]}
= K \norm{\mC_{\vx} \ve_1 \vu_m^*}
\leq K \norm{\vx}_2 \norm{\vu_m}_2
\leq K \norm{\vx}_2 \norm{\vu}_{2,\infty}.
\]

We again use bounds for suprema of second order chaos processes \cite{krahmer2014suprema,lee2015rip} to get a tail bound for $\norm{\mZ_m - \mathbb{E}[\mZ_m]}$, as given in the following lemma, which is proved in Appendix~\ref{sec:proof:lemma:comp}.

\begin{lemma}
\label{lemma:comp}
Suppose that (A1) holds.
For any $\beta \in \mathbb{N}$, there exist a numerical constant $\alpha \in \mathbb{N}$ and a constant $C(\beta)$ that depends only on $\beta$ such that
\[
\sup_{1 \leq m \leq M} \norm{\mZ_m - \mathbb{E}[\mZ_m]}
\leq C(\beta) \norm{\vu}_{2,\infty} K \log^{\alpha}(MKL)
\]
holds with probability $1-K^{-\beta}$.
\end{lemma}

By Lemma~\ref{lemma:comp}, it follows that the relative perturbation due to $\text{(e)}$ is upper bounded by
\begin{equation}
\label{eq:bnd_snEs_part2}
\frac{\norm{\text{(e)}}}{K^2 \norm{\vx}_2^2 \norm{\vu}_2^2}
\leq \frac{C(\beta) \log^\alpha(MKL) \mu^2}{M}
\end{equation}
with probability $1-K^{-\beta}$.

Similarly, $\text{(f)}$, which corresponds to the full 2D summation, is rewritten as
\begin{align*}
\text{(f)} & =
\Big( \sum_{m=1}^M \ve_m^* \otimes \mC_{\vx} \breve{\mS}^* \mZ_m \Big)^*
\Big( \sum_{m'=1}^M \ve_{m'}^* \otimes \mC_{\vx} \breve{\mS}^* \mZ_{m'} \Big) \\
& -
\Big( \sum_{m=1}^M \ve_m^* \otimes \mathbb{E}[\mC_{\vx} \breve{\mS}^* \mZ_m] \Big)^*
\Big( \sum_{m'=1}^M \ve_{m'}^* \otimes \mathbb{E}[\mC_{\vx} \breve{\mS}^* \mZ_{m'}] \Big) \\
& =
\Big( \sum_{m=1}^M \ve_m^* \otimes (\mC_{\vx} \breve{\mS}^* \mZ_m - \mathbb{E}[\mC_{\vx} \breve{\mS}^* \mZ_m]) \Big)^*
\Big( \sum_{m'=1}^M \ve_{m'}^* \otimes (\mC_{\vx} \breve{\mS}^* \mZ_{m'} - \mathbb{E}[\mC_{\vx} \breve{\mS}^* \mZ_{m'}]) \Big) \\
& +
\Big( \sum_{m=1}^M \ve_m^* \otimes \mathbb{E}[\mC_{\vx} \breve{\mS}^* \mZ_m] \Big)^*
\Big( \sum_{m'=1}^M \ve_{m'}^* \otimes (\mC_{\vx} \breve{\mS}^* \mZ_{m'} - \mathbb{E}[\mC_{\vx} \breve{\mS}^* \mZ_{m'}]) \Big) \\
& +
\Big( \sum_{m=1}^M \ve_m^* \otimes (\mC_{\vx} \breve{\mS}^* \mZ_m - \mathbb{E}[\mC_{\vx} \breve{\mS}^* \mZ_m]) \Big)^*
\Big( \sum_{m'=1}^M \ve_{m'}^* \otimes \mathbb{E}[\mC_{\vx} \breve{\mS}^* \mZ_{m'}] \Big).
\end{align*}

Therefore, by the triangle inequality, we have
\begin{align*}
\norm{\text{(f)}}
& \leq
\Big\| \sum_{m=1}^M \ve_m^* \otimes (\mC_{\vx} \breve{\mS}^* \mZ_m - \mathbb{E}[\mC_{\vx} \breve{\mS}^* \mZ_m]) \Big\|^2 \\
& + 2 \Big\| \sum_{m=1}^M \ve_m^* \otimes (\mC_{\vx} \breve{\mS}^* \mZ_m - \mathbb{E}[\mC_{\vx} \breve{\mS}^* \mZ_m]) \Big\|
\Big\| \sum_{m=1}^M \ve_m^* \otimes \mathbb{E}[\mC_{\vx} \breve{\mS}^* \mZ_m] \Big\|.
\end{align*}

Let $\vv_1,\dots,\vv_M \in \mathbb{C}^D$ and $\vv = [\vv_1^\transpose,\dots,\vv_M^\transpose]^\transpose$.
Then, by Lemma~\ref{lemma:expectation2}, we have
\begin{align*}
\Big\|
\sum_{m=1}^M \ve_m^* \otimes \mathbb{E}[\mC_{\vx} \breve{\mS}^* \mZ_m]
\Big\|
& = \sup_{\vv \in B_2^{MD}} \Big\| \sum_{m=1}^M \mathbb{E}[\mC_{\vx} \breve{\mS}^* \mZ_m \vv_m] \Big\|_2 \\
& = \sup_{\vv \in B_2^{MD}} K \Big\|\sum_{m=1}^M \mC_{\vx} \ve_1 \vu_m^*\vv_m\Big\|_2
= \sup_{\vv \in B_2^{MD}} K \norm{\vx}_2 \Big|\sum_{m=1}^M \vu_m^*\vv_m\Big| \\
& \leq \sup_{\vv \in B_2^{MD}} K \norm{\vx}_2 \sum_{m=1}^M \norm{\vu_m}_2 \norm{\vv_m}_2
\leq K \norm{\vx}_2 \norm{\vu}_2.
\end{align*}

On the other hand,
\begin{align*}
\Big\|
\sum_{m=1}^M \ve_m^* \otimes (\mC_{\vx} \breve{\mS}^* \mZ_m - \mathbb{E}[\mC_{\vx} \breve{\mS}^* \mZ_m])
\Big\|
& =
\Big\|
\mC_{\vx} \breve{\mS}^* \Big( \sum_{m=1}^M \ve_m^* \otimes (\mZ_m - \mathbb{E}[\mZ_m]) \Big)
\Big\| \\
& \leq \norm{\breve{\mS} \mC_{\vx}^* \mC_{\vx} \breve{\mS}^*}^{1/2}
\underbrace{
\Big\|
\sum_{m=1}^M \ve_m^* \otimes (\mZ_m - \mathbb{E}[\mZ_m])
\Big\|
}_{\text{($\ddag$)}}.
\end{align*}
As implied by \eqref{eq:cond_rho}, the first factor is bounded by $\sqrt{C_3} \|\vx\|_2$. Hence it remains to show an upper bound on the last term ($\ddag$). This is established in Lemma~\ref{lemma:comp_sum} using the results on suprema of second-order chaos processes \cite{krahmer2014suprema,lee2015rip} together with an entropy bound by polytope approximation of a unit ball \cite{junge2017ripI}, polar duality, and entropy duality \cite{artstein2004duality}; see Appendix~\ref{sec:proof:lemma:comp_sum} for the proof.

\begin{lemma}
\label{lemma:comp_sum}
Suppose that (A1) holds.
For any $\beta \in \mathbb{N}$, there exist a numerical constant $\alpha \in \mathbb{N}$ and a constant $C(\beta)$ that depends only on $\beta$ such that
\[
\Big\|
\sum_{m=1}^M \ve_m^* \otimes (\mZ_m - \mathbb{E}[\mZ_m])
\Big\|
\leq C(\beta) \norm{\vu}_{2,\infty} (K + \sqrt{MKD}) \log^\alpha (MKL)
\]
holds with probability $1-K^{-\beta}$.
\end{lemma}

Collecting the above estimates and noting that, up to log factors, both factors are bounded by $\mu K \norm{\vx}_2 \norm{\vu}_2$, we obtain that the relative perturbation due to $\text{(f)}$ is upper bounded by
\begin{equation}
\label{eq:bnd_snEs_part3}
\frac{\norm{\text{(f)}}}{K^2 \norm{\vx}_2^2 \norm{\vu}_2^2}
\leq
C(\beta) \log^\alpha(MKL) \Big(\sqrt{\frac{1}{M}} + \sqrt{\frac{D}{K}} \, \Big) \mu^2,
\end{equation}
potentially for an increased value of $\alpha$, with probability $1-CK^{-\beta}$.

Finally, \eqref{eq:bnd_snEs} follows by combining \eqref{eq:bnd_snEs_part1}, \eqref{eq:bnd_snEs_part2}, and \eqref{eq:bnd_snEs_part3}.

\subsection{Proof of Lemma~\ref{lemma:Ec}}

The proof of Lemma~\ref{lemma:Ec} is similar to (and easier than) that of Lemma~\ref{lemma:Es}.  We will reuse some of tail estimates obtained in Section~\ref{sec:proof:lemma:Es}.  On the other hand, we also need to derive tail estimates of the suprema of certain Gaussian processes, which did not arise in Section~\ref{sec:proof:lemma:Es}.  We use a moment-version of Dudley's inequality \cite{foucart2013mathematical} to compute these tail estimates.

Similarly to the previous section, we use the following decomposition into the diagonal block portion and the off-diagonal block portion:
\[
\mPhi^* \mY_{\mathrm{s}}^* \mY_{\mathrm{n}} \mPhi = \text{(g)} + \text{(h)},
\]
where
\begin{equation}
\label{eqn_dbl_xw}
\begin{aligned}
\text{(g)} &=
\sum_{m=1}^M \ve_m \ve_m^* \otimes \Big( \sum_{\begin{subarray}{c} m'=1 \\ m' \neq m \end{subarray}}^M
\widetilde{\mPhi}_m^* \mC_{\widetilde{\mPhi}_{m'} \vu_{m'}}^* \mC_{\vx}^* \mC_{\vw_{m'}} \widetilde{\mPhi}_m \Big), \\
\text{(h)} &= -
\sum_{m=1}^M \sum_{\begin{subarray}{c} m'=1 \\ m' \neq m \end{subarray}}^M
\ve_m \ve_{m'}^* \otimes
\widetilde{\mPhi}_m^* \mC_{\widetilde{\mPhi}_{m'} \vu_{m'}}^* \mC_{\vx}^* \mC_{\vw_m} \widetilde{\mPhi}_{m'}.
\end{aligned}
\end{equation}

We derive upper bounds on the spectral norms of $\text{(g)}$ and $\text{(h)}$ respectively in the following.

\noindent\textbf{Diagonal block portion:} Note that $\text{(g)}$ is a block diagonal matrix and its expectation is $\vzero$.
Define
\begin{equation}
\label{eq:defUpsiloncm}
\mUpsilon_{\mathrm{c},m} =
\sum_{\begin{subarray}{c} m'=1 \\ m' \neq m \end{subarray}}^M
\mS \mC_{\widetilde{\mPhi}_{m'} \vu_{m'}}^* \mC_{\vx}^* \mC_{\vw_{m'}} \mS^*.
\end{equation}
Then it follows from the block diagonal structure that
\begin{align*}
\norm{\text{(g)}}
\leq \max_{1\leq m \leq M} \norm{\mPhi_m^* \mUpsilon_{\mathrm{c},m} \mPhi_m}.
\end{align*}

Since $\mPhi_m$ and $\mUpsilon_{\mathrm{c},m}$ are independent, $\mPhi_m^* \mUpsilon_{\mathrm{c},m} \mPhi_m$ is rewritten as
\begin{equation}
\label{eq:blockXic1}
\begin{aligned}
\mPhi_m^* \mUpsilon_{\mathrm{c},m} \mPhi_m
=
\underbrace{
\mPhi_m^* \mUpsilon_{\mathrm{c},m} \mPhi_m - \mathrm{tr}(\mUpsilon_{\mathrm{c},m}) \mId_D
}_{\text{(i)}}
+
\underbrace{
\mathrm{tr}(\mUpsilon_{\mathrm{c},m}) \mId_D
}_{\text{(j)}}.
\end{aligned}
\end{equation}

First, we compute a tail estimate of $\norm{\text{(i)}}$.
Similarly to the previous section, we use Lemma~\ref{lemma:gaussquadasym} conditioned on $\mUpsilon_{\mathrm{c},m}$. Then we apply the tail estimate of $\norm{\mUpsilon_{\mathrm{c},m}}$ given in following lemma, whose proof is in Appendix~\ref{sec:proof:lemma:Upsiloncm}.

\begin{lemma}
\label{lemma:Upsiloncm}
Suppose that (A1) holds.
For any $\beta \in \mathbb{N}$, there exist a numerical constant $\alpha \in \mathbb{N}$ and a constant $C(\beta)$ that depends only on $\beta$ such that, conditional on the noise vector $\vw$,
\begin{align*}
\max_{1 \leq m \leq M} \norm{\mUpsilon_{\mathrm{c},m}}
& \leq C(\beta) \rho_{x,w} \sqrt{K} \norm{\vu}_2 \log^\alpha(MKL)
\end{align*}
holds with probability $1-K^{-\beta}$.
\end{lemma}

By Lemmas~\ref{lemma:gaussquadasym} and \ref{lemma:Upsiloncm}, we obtain that
\[
\norm{\text{(i)}} \leq C(\beta) \rho_{x,w} K \sqrt{D} \norm{\vu}_2 \log^\alpha(MKL)
\]
holds with probability $1-M^{-1}K^{-\beta}$.

On the other hand, a direct calculation shows that $\text{(j)}$ is expressed as
\[
\overline{\mathrm{tr}(\mUpsilon_{\mathrm{c},m})}
= \sum_{\begin{subarray}{c} m'=1 \\ m' \neq m \end{subarray}}^M
(\vu_{m'}^\transpose \otimes K \vw_{m'}^* \mC_{\vx} \mS^*)
\mathrm{vec}(\mPhi_{m'}),
\]
where $\overline{\mathrm{tr}(\mUpsilon_{\mathrm{c},m})}$ denotes the complex conjugate of $\mathrm{tr}(\mUpsilon_{\mathrm{c},m})$.  The above expression implies that $\overline{\mathrm{tr}(\mUpsilon_{\mathrm{c},m})}$ is a linear function of the independent Gaussian matrix entries, and hence is a zero-mean Gaussian random variable. The variance of $\overline{\mathrm{tr}(\mUpsilon_{\mathrm{c},m})}$ is given by
\[
\sum_{\begin{subarray}{c} m'=1 \\ m' \neq m \end{subarray}}^M
K^2 \norm{\vu_{m'}}_2^2 \norm{\vw_{m'}^* \mC_{\vx} \mS^*}_2^2
\leq K^2 \norm{\vu}_2^2 \rho_{x,w}^2.
\]
Therefore, by a tail estimate of a Gaussian variable,
\[
\norm{\text{(j)}} \leq C \rho_{x,w} K \norm{\vu}_2 \sqrt{1+ \log M + \beta \log K}
\]
with probability $1-M^{-1}K^{-\beta}$. Hence $\norm{\text{(i)}}$ dominates $\norm{\text{(j)}}$.

By collecting the above estimates, we obtain that the relative perturbation due to $\text{(g)}$ is upper bounded by
\begin{equation}
\label{eq:bnd_snEc_part1}
\frac{\norm{\text{(g)}}}{K^2 \norm{\vx}_2^2 \norm{\vu}_2^2}
\leq
\frac{C(\beta) \rho_{x,w} \log^{\alpha}(MKL)}{\norm{\vx}_2^2 \norm{\vu}_2}
\frac{\sqrt{D}}{K}
=
\frac{C(\beta) \log^{\alpha}(MKL)}{\sqrt{\eta L}}
\frac{\rho_{x,w}}{\sqrt{K} \sigma_w \norm{\vx}_2}
\sqrt{\frac{D}{M}}
\end{equation}
with probability $1-CK^{-\beta}$, where in the last step we used \eqref{eq:etasimp}.

\noindent\textbf{Off-diagonal portion: } Similarly to the analogous part of the proof of Lemma~\ref{lemma:Es}, we add and subtract the diagonal sum and obtain
\[
\text{(h)} = \text{(k)} + \text{(l)},
\]
where
\begin{equation}
\label{eq:ubXic2}
\begin{aligned}
\text{(k)} &=
\sum_{m=1}^M \ve_m \ve_m^* \otimes
\mPhi_m^* \mS \mC_{\vx}^* \mC_{\vw_m} \breve{\mS}^* \mZ_m, \\
\text{(l)}
&= -\sum_{m,m'=1}^M \ve_m \ve_{m'}^* \otimes
\mPhi_m^* \mS \mC_{\vx}^* \mC_{\vw_m} \breve{\mS}^* \mZ_{m'}.
\end{aligned}
\end{equation}

By the triangle inequality,
\begin{align*}
\norm{\text{(k)}}
&\leq
\Big\|
\underbrace{
\sum_{m=1}^M \ve_m \ve_m^* \otimes
\mPhi_m^* \mS \mC_{\vx}^* \mC_{\vw_m} \breve{\mS}^* (\mZ_m - \mathbb{E}[\mZ_m])
}_{\text{(m)}}
\Big\| \\
&+
\Big\|
\underbrace{
\sum_{m=1}^M \ve_m \ve_m^* \otimes
\mPhi_m^* \mS \mC_{\vx}^* \mC_{\vw_m} \breve{\mS}^* \mathbb{E}[\mZ_m]
}_{\text{(n)}}
\Big\|.
\end{align*}

We use the result by Davidson and Szarek \cite[Theorem~II.13]{davidson2001local} to get a tail estimate of $\norm{\mPhi_m}$. Specifically, it follows from (A1) that
\begin{equation}
\label{eq:ubspnPhim}
\max_{1 \leq m \leq M} \norm{\mPhi_m}
\leq \norm{[\sqrt{2} \mathrm{Re}(\mPhi_m), \sqrt{2} \mathrm{Im}(\mPhi_m)]}
\leq \sqrt{K}+\sqrt{2D}+\sqrt{2 \log M + 2\beta \log K}
\end{equation}
holds with probability $1 - K^{-\beta}$.

By Lemma~\ref{lemma:comp} and \eqref{eq:ubspnPhim},
\begin{align*}
\norm{\text{(m)}}
&= \max_{1\leq m\leq M} \norm{\mPhi_m^* \mS \mC_{\vx}^* \mC_{\vw_m} \breve{\mS}^* (\mZ_m - \mathbb{E}[\mZ_m])} \\
&\leq
\Big(\max_{1\leq m\leq M} \norm{\mS \mC_{\vx}^* \mC_{\vw_m} \breve{\mS}}\Big)
\Big( \max_{1\leq m\leq M} \norm{\mPhi_m} \Big)
\Big( \max_{1\leq m\leq M} \norm{\mZ_m - \mathbb{E}[\mZ_m]} \Big) \\
& \leq \rho_{x,w} C(\beta) K^{3/2} \norm{\vu}_{2,\infty} \log^6(MK)
\end{align*}
holds with probability $1-K^{-\beta}$, where $\norm{\cdot}_{2,\infty}$ is defined in \eqref{eq:defblknorm}.
On the other hand, by Lemmas~\ref{lemma:expectation2} and \ref{lemma:hansonwright},
\begin{align*}
\norm{\text{(n)}}
&= K \max_{1\leq m\leq M} \norm{\mPhi_m^* \mS \mC_{\vx}^* \vw_m \vu_m^*} \\
&\leq K \norm{\vu}_{2,\infty} \max_{1\leq m\leq M} \norm{\mPhi_m^* \mS \mC_{\vx}^* \vw_m}_2 \\
&= K \norm{\vu}_{2,\infty} \max_{1\leq m\leq M}
\norm{(\vw_m^* \mC_{\vx} \mS^* \otimes \mId_D) \mathrm{vec}(\mPhi_m^\transpose)}_2 \\
&\leq K \norm{\vu}_{2,\infty} \max_{1\leq m\leq M}
\norm{\vw_m^* \mC_{\vx} \mS^* \otimes \mId_D}_{\mathrm{F}} \sqrt{1+ \log M + \beta \log K} \\
&= K \norm{\vu}_{2,\infty} \max_{1\leq m\leq M}
\sqrt{D} \norm{\vw_m^* \mC_{\vx} \mS^*}_2 \sqrt{1+ \log M + \beta \log K} \\
&\leq 2 K \norm{\vu}_{2,\infty} \rho_{x,w} \sqrt{D} \sqrt{1+ \log M + \beta \log K}
\end{align*}
holds with probability $1-K^{-\beta}$, where the last inequality follows from the fact that $\norm{\vw_m^* \mC_{\vx} \mS^*}_2 \leq \norm{\mS \vw_m^* \mC_{\vx} \mS^*}$.  Note that $\norm{\text{(m)}}$ dominates $\norm{\text{(n)}}$.

The spectral norm of $\text{(l)}$ is upper bounded through a factorization by
\begin{align*}
\norm{\text{(l)}}
& \leq
\Big\|
\underbrace{
\sum_{m=1}^M
\ve_m^* \otimes \breve{\mS} \mC_{\vw_m}^* \mC_{\vx} \mS^* \mPhi_m
}_{\text{(o)}}
\Big\|
\Big\|
\underbrace{
\sum_{m'=1}^M
\ve_{m'}^* \otimes (\mZ_{m'} - \mathbb{E}[\mZ_{m'}])
}_{\text{(p)}}
\Big\| \\
& +
\Big\|
\underbrace{
\sum_{m,m'=1}^M \ve_m \ve_{m'}^* \otimes
\mPhi_m^* \mS \mC_{\vx}^* \mC_{\vw_m} \breve{\mS}^* \mathbb{E}[\mZ_{m'}]
}_{\text{(q)}}
\Big\|.
\end{align*}

The spectral norm of $\text{(o)}$ is written as the supremum of a Gaussian process and is bounded by the following lemma, which is proved in Appendix~\ref{sec:proof:lemma:comp_cross}. 
\begin{lemma}
\label{lemma:comp_cross}
Suppose that (A1) holds.
For any $\beta \in \mathbb{N}$, there exists a constant $C(\beta)$ that depends only on $\beta$ such that, conditional on the noise vector $\vw$,
\[
\Big\|
\sum_{m=1}^M
\ve_m^* \otimes \breve{\mS} \mC_{\vw_m}^* \mC_{\vx} \mS^* \mPhi_m
\Big\|
\leq C \sqrt{1+\beta} \rho_{x,w} (\sqrt{MD} + \sqrt{K}) \log K
\]
holds with probability $1-K^{-\beta}$.
\end{lemma}

By Lemma~\ref{lemma:comp_sum},
\begin{align*}
\norm{\text{(p)}} \leq C(\beta) \norm{\vu}_{2,\infty} (K + \sqrt{MKD}) \log^\alpha (MKL)
\end{align*}
with probability $1-K^{-\beta}$, where $\norm{\cdot}_{2,\infty}$ is defined in \eqref{eq:defblknorm}.  Note that $\norm{\text{(o)}} \norm{\text{(p)}}$ dominates $\norm{\text{(m)}}$.  Therefore, we may ignore $\norm{\text{(m)}}$.

By Lemma~\ref{lemma:expectation2}, the spectral norm of $\text{(q)}$ is upper bounded by
\begin{align*}
\norm{\text{(q)}}^2
&= \Big\|
\sum_{m,m'=1}^M \ve_m \ve_{m'}^* \otimes
K \mPhi_m^* \mS \mC_{\vx}^* \vw_m \vu_{m'}^*
\Big\|^2 \\
& = \Big\|
\sum_{m=1}^M \ve_m \otimes
K \mPhi_m^* \mS \mC_{\vx}^* \vw_m \vu^*
\Big\|^2 \\
& = K^2 \norm{\vu}_2^2 \sum_{m=1}^M \norm{\mPhi_m^* \mS \mC_{\vx}^* \vw_m}_2^2 \\
& = K^2 \norm{\vu}_2^2 \sum_{m=1}^M \norm{(\mId_D \otimes \vw_m^* \mC_{\vx} \mS^*) \mathrm{vec}(\mPhi_m)}^2 \\
& = K^2 \norm{\vu}_2^2 \Big\| \Big(\sum_{m=1}^M \ve_m^\transpose \otimes \mId_D \otimes \vw_m^* \mC_{\vx} \mS^* \Big) [\mathrm{vec}(\mPhi_m)^\transpose, \dots, \mathrm{vec}(\mPhi_m)^\transpose]^\transpose\Big\|^2.
\end{align*}
Therefore, by Lemma~\ref{lemma:hansonwright},
\[
\norm{\text{(q)}}
\leq C(\beta) K \norm{\vu}_2 \sqrt{MD} \rho_{x,w} \sqrt{\log K}
\]
holds with probability $1-K^{-\beta}$.

By collecting the estimates with the fact that $\mu \leq \sqrt{M}$, we obtain that with probability $1-K^{-\beta}$, the relative perturbation due to $\text{(h)}$ is upper bounded by
\begin{equation}
\label{eq:bnd_snEc_part2}
\begin{aligned}
\frac{\norm{\text{(h)}}}{K^2 \norm{\vx}_2^2 \norm{\vu}_2^2}
& \leq
\frac{C(\beta) \rho_{x,w} \log^{\alpha}(MKL)}{\norm{\vx}_2^2 \norm{\vu}_2}
\Big(
\mu \sqrt{MD + K}
\Big( \frac{1}{\sqrt{M} K} + \frac{\sqrt{D}}{K^{3/2}} \Big)
+ \frac{\sqrt{MD}}{K}
\Big) \\
&\leq
\frac{C'(\beta) \log^{\alpha}(MKL)}{\sqrt{\eta L}}
\frac{\rho_{x,w}}{\sqrt{K} \sigma_w \norm{\vx}_2}
\Big(
\mu
\Big( \frac{\sqrt{K}}{M} + \frac{D}{\sqrt{K}} \Big)
+ \sqrt{D}
\Big).
\end{aligned}
\end{equation}

Note that the tail estimate in \eqref{eq:bnd_snEc_part2} dominates that in \eqref{eq:bnd_snEc_part1}.
It therefore follows that \eqref{eq:bnd_snEc} holds with probability $1-CK^{-\beta}$.
This completes the proof.

\subsection{Proof of Lemma~\ref{lemma:En}}

The analysis of the noise term $\mE_{\mathrm{n}}$ only involves second-order chaos processes and can be reduced to bounds for suprema of such processes as they are established in \cite{krahmer2014suprema,lee2015rip}.

The first lemma used in the proof bounds the maximum cross-correlation deviation of the noise terms, which is given by
\[
\rho_w := \max_{1\leq m,m' \leq M} \norm{\mS (\mC_{\vw_m}^* \mC_{\vw_{m'}} - \mathbb{E}[\mC_{\vw_m}^* \mC_{\vw_{m'}}]) \mS^*}.
\]
See Appendix~\ref{sec:proof:lemma:rho_w} for the proof of the lemma.

\begin{lemma}
\label{lemma:rho_w}
Suppose that (A2) holds.
For any $\beta \in \mathbb{N}$, there is a constant $C(\beta)$ that depends only on $\beta$ such that
\begin{equation}
\label{eq:ubrhow}
\rho_w \leq C(\beta) \sigma_w^2 \sqrt{KL} \log^\alpha (MKL)
\end{equation}
holds with probability $1 - K^{-\beta}$.
\end{lemma}

The second proof ingredient is a bound for the average auto-correlation deviation of the noise terms, which is given by
\[
\bar{\rho}_w :=
\Big\|
\frac{1}{M}
\sum_{m=1}^M
\mS (\mC_{\vw_m}^* \mC_{\vw_m} - \mathbb{E}[\mC_{\vw_m}^* \mC_{\vw_m}]) \mS^*
\Big\|.
\]

The bound is provided by the following lemma; see Appendix~\ref{sec:proof:lemma:barrho_w} for its proof.

\begin{lemma}
\label{lemma:barrho_w}
Suppose that (A2) holds.
For any $\beta \in \mathbb{N}$, there is a constant $C(\beta)$ that depends only on $\beta$ such that
\begin{equation}
\label{eq:ubbarrhow}
\bar{\rho}_w \leq C(\beta) \sigma_w^2 M^{-1/2} \sqrt{KL} \log^\alpha (MKL)
\end{equation}
holds with probability $1 - K^{-\beta}$.
\end{lemma}

For the remainder of the proof of Lemma~\ref{lemma:En}, we condition on the event that  \eqref{eq:ubrhow} and \eqref{eq:ubbarrhow} hold.

Define $\mLambda \in \mathbb{C}^{MK \times MK}$ by
\[
\mLambda = \sum_{m=1}^M \sum_{\begin{subarray}{c} m'=1 \\ m' \neq m \end{subarray}}^M \norm{\vw_{m'}}_2^2 \ve_m \ve_m^* \otimes \mId_K.
\]

Under the assumption in (A2), we have
\[
\mathbb{E}[\norm{\vw_{m'}}_2^2] = \sigma_w^2 L, \quad \forall m'=1,\dots,M.
\]
Then it follows that $\mathbb{E}[\mLambda]$ is a scalar multiple of the identity given by
\[
\mathbb{E}[\mLambda] = \sigma_w^2 (M-1) L \mId_{MK}.
\]
Here we assume that $\sigma_w$ is known a priori or can be estimated from the data.

We decompose $\mPhi^* (\mY_{\mathrm{n}}^* \mY_{\mathrm{n}} - \sigma_w^2 (M-1) L \mId_{MK}) \mPhi$ into two parts as follows:
\begin{equation}
\label{eq:decomp_noise}
\begin{aligned}
\mPhi^* (\mY_{\mathrm{n}}^* \mY_{\mathrm{n}} - \sigma_w^2 (M-1) L \mId_{MK}) \mPhi
=
\underbrace{
\mPhi^* (\mY_{\mathrm{n}}^* \mY_{\mathrm{n}} - \mLambda) \mPhi
}_{\text{(r)}}
+
\underbrace{
\mPhi^* (\mLambda - \mathbb{E}[\mLambda]) \mPhi
}_{\text{(s)}}.
\end{aligned}
\end{equation}
Then we estimate the summands in the right-hand side of \eqref{eq:decomp_noise}.  The following lemma, which is proved in Appendix~\ref{sec:proof:lemma:noise_part1}, provides a tail estimate of $\norm{\text{(r)}}$.

\begin{lemma}
\label{lemma:noise_part1}
Suppose that (A1) holds.
For any $\beta \in \mathbb{N}$, there exist a numerical constant $\alpha \in \mathbb{N}$ and a constant $C(\beta)$ that depends only on $\beta$ such that, conditional on the noise vector $\vw$,
\[
\norm{\mPhi^* (\mY_{\mathrm{n}}^* \mY_{\mathrm{n}} - \mLambda) \mPhi}
\leq C(\beta) M \sqrt{KD} (\bar{\rho}_w + 2 \rho_w) \log^\alpha(MD)
\]
holds with probability $1-K^{-\beta}$.
\end{lemma}

Next, due to the block diagonal structure, the spectral norm of the second term $\text{(s)}$ is upper bounded by
\begin{align*}
& \max_{1\leq m\leq M}
\Big\|
\sum_{\begin{subarray}{c} m'=1 \\ m' \neq m \end{subarray}}^M (\norm{\vw_{m'}}_2^2 - \mathbb{E}[\norm{\vw_{m'}}_2^2]) \mPhi_m^* \mPhi_m
\Big\| \\
& \leq
\Big(
\max_{1\leq m\leq M} \Big|
\sum_{\begin{subarray}{c} m'=1 \\ m' \neq m \end{subarray}}^M
\norm{\vw_{m'}}_2^2 - \mathbb{E}[\norm{\vw_{m'}}_2^2]
\Big|
\Big)
\Big(
\max_{1\leq m\leq M} \norm{\mPhi_m^* \mPhi_m}
\Big).
\end{align*}
The first factor divided by $\sigma_w^2$ is a $\chi^2$ random variable with $LM$ degrees of freedom. By Lemma~\ref{lemma:chisq},
\[
\max_{1\leq m\leq M} \Big|
\sum_{\begin{subarray}{c} m'=1 \\ m' \neq m \end{subarray}}^M
\norm{\vw_{m'}}_2^2 - \mathbb{E}[\norm{\vw_{m'}}_2^2]
\Big| \leq C_1 (1+\beta) \sigma_w^2 \sqrt{LM} \log (MK)
\]
holds with probability $1-K^{-\beta}$.
On the other hand, by \eqref{eq:ubspnPhim},
\[
\max_{1\leq m\leq M} \norm{\mPhi_m^* \mPhi_m}
\leq C_2 (1+\beta) K \log(MK)
\]
holds with probability $1-K^{-\beta}$.
Therefore,
\[
\norm{\text{(s)}} \leq C(\beta) \sigma_w^2 \sqrt{M} K \sqrt{L} \log^2 (MK)
\]
with probability $1-2K^{-\beta}$.

Finally, \eqref{eq:bnd_snEn} follows by collecting the above estimates.
This completes the proof.

\subsection{Proof of Lemma~\ref{lemma:Enu}}
\label{sec:proof:lemma:Enu}

Similarly to the proof of Lemma~\ref{lemma:En}, through the triangle inequality, the left-hand side of \eqref{eq:bnd_snEnu} is upper bounded by
\begin{equation}
\label{eq:esterr_noiseterm}
\frac{\norm{\mPhi^* (\mY_{\mathrm{n}}^* \mY_{\mathrm{n}} - \mLambda) \mPhi \vu}_2}{\norm{\vu}_2}
+
\frac{\norm{\mPhi^* (\mLambda - \mathbb{E}[\mLambda]) \mPhi \vu}_2}{\norm{\vu}_2}.
\end{equation}

For the first term in the right-hand side of \eqref{eq:esterr_noiseterm}, we modify Lemma~\ref{lemma:noise_part1} as follows; see Appendix~\ref{sec:proof:lemma:noise_part1_fixu} for the proof.

\begin{lemma}
\label{lemma:noise_part1_fixu}
Suppose that (A1) holds.
For any $\beta \in \mathbb{N}$, there exist a numerical constant $\alpha \in \mathbb{N}$ and a constant $C(\beta)$ that depends only on $\beta$ such that, conditional on the noise vector $\vw$,
\[
\frac{
\norm{
\mPhi^* (\mY_{\mathrm{n}}^* \mY_{\mathrm{n}} - \mLambda) \mPhi \vu
}_2
}{\norm{\vu}_2}
\leq C(\beta) \sqrt{MKD} (\sqrt{M} \bar{\rho}_w + \rho_w) \log^\alpha(MD)
\]
holds with probability $1-K^{-\beta}$.
\end{lemma}

Combining Lemmas~\ref{lemma:rho_w},\ref{lemma:barrho_w}, and \ref{lemma:noise_part1_fixu} implies that the first term in the right-hand side of \eqref{eq:esterr_noiseterm} is smaller than $\norm{\mPhi^* (\mY_{\mathrm{n}}^* \mY_{\mathrm{n}} - \mLambda) \mPhi}$ by a factor of $\sqrt{M}$.

For the second term in the right-hand side of \eqref{eq:esterr_noiseterm}, we use the fact that it is no larger than $\norm{\mPhi^* (\mLambda - \mathbb{E}[\mLambda]) \mPhi}$.
Then we may use the tail estimate derived in the proof of Lemma~\ref{lemma:En}.

By collecting the estimates, we obtain that \eqref{eq:bnd_snEnu} holds with probability $1-CK^{-\beta}$. This completes the proof.

\section{Conclusion}
We studied a passive imaging problem with multiple channels, which is formulated as multichannel blind deconvolution with noise-like source and time-limited impulse responses.  Additionally, motivated by several real world applications, we impose that the FIR coefficients of impulse responses belong to corresponding low dimensional subspaces.  For such a scenario, we proposed a spectral method called subspace-constrained cross-convolution (SCCC) that modifies and improves upon a classical method developed in the 1990s by overcoming the noise sensitivity.  SCCC provides stable estimates of the impulse responses from finitely many samples and its performance is backed by theoretical error bounds under generic subspace models.  In this scenario, SCCC also empirically outperforms competing approaches.  The fundamental estimates in the analysis of this paper extend to the sparsity or low-rank cases with minor changes.  Corresponding recovery results on these extended models will be left for follow-up work.

\appendix

\section{Toolbox}

\subsection{Concentration of $\chi^2$ Random Variables}
\begin{lemma}[{Complexification of \cite[Lemma~1]{laurent2000adaptive}}]
\label{lemma:chisq}
Let $g_1,\dots,g_n$ be independent copies of a standard complex Gaussian variable.
Let $a_1,\dots,a_n$ be nonnegative and $\va = [a_1,\dots,a_n]^\transpose$.
Let $Z = \sum_{k=1}^n a_k (|g_k|^2-1)$.
Then, for any $t > 0$,
\begin{align*}
\mathbb{P}(Z \geq \sqrt{2} \norm{\va}_2 \sqrt{t} + \norm{\va}_\infty t) & \leq \exp(-t), \\
\mathbb{P}(Z \leq -\sqrt{2} \norm{\va}_2 \sqrt{t}) & \leq \exp(-t).
\end{align*}
\end{lemma}

\subsection{Hanson-Wright Inequality}
\begin{lemma}[{Complexification of \cite[Theorem 1.1]{rudelson2013hanson}}]
\label{lemma:hansonwright_ori}
Let $\mA \in \mathbb{C}^{m \times n}$.
Let $\vg \in \mathbb{C}^n$ be a standard complex Gaussian vector.
For any $0 < \zeta < 1$, there exists an absolute constant $C$ such that
\begin{align*}
|\norm{\mA \vg}_2^2 - \mathbb{E}[\norm{\mA \vg}_2^2]|
\leq C (
\norm{\mA^* \mA}_{\mathrm{F}} \sqrt{\log(2\zeta^{-1})}
\vee
\norm{\mA}^2 \log(2\zeta^{-1})
)
\end{align*}
holds with probability $1-\zeta$.
\end{lemma}

\begin{lemma}[{Complexification of \cite[Theorem 2.1]{rudelson2013hanson}}]
\label{lemma:hansonwright}
Let $\mA \in \mathbb{C}^{m \times n}$.
Let $\vg \in \mathbb{C}^n$ be a standard complex Gaussian vector.
For any $0 < \zeta < 1$, there exists an absolute constant $C$ such that
\begin{align*}
|\norm{\mA \vg}_2 - \norm{\mA}_{\mathrm{F}}| \leq C \norm{\mA} \sqrt{\log(2\zeta^{-1})}
\end{align*}
holds with probability $1-\zeta$.
\end{lemma}

\subsection{Complexification of Maurey's Lemma}

The following lemma is a direct consequence of Maurey's empirical method \cite{carl1985inequalities}.  Define a block norm on $\mathbb{R}^{md}$ by
\[
\norm{[\vq_1^\transpose,\dots,\vq_m^\transpose]^\transpose}_{\ell_\infty^m(\ell_2^d)}
:= \max_{1\leq k \leq m} \norm{\vq_k}_2.
\]
Let $\ell_\infty^m(\ell_2^d(\mathbb{R}))$ denote the corresponding Banach space.  Similarly, $\ell_\infty^m(\ell_2^d(\mathbb{C}))$ is defined over the complex scalar field.

\begin{lemma}[{Maurey's empirical method \cite[Lemma~3.1]{junge2017ripI}}]
\label{lemma:maurey}
Let $k,m,n\in \mathbb{N}$ and $T:\ell_1^k(\mathbb{R}) \to \ell_\infty^m(\ell_2^d(\mathbb{R}))$ be a linear operator.
Then
\[
\int_0^\infty \sqrt{\log N(T(B_{\ell_1^k(\mathbb{R})}), \norm{\cdot}_{\ell_\infty^m(\ell_2^d(\mathbb{R}))}, t)} dt
\leq C \sqrt{(1+\log k)(1+\log m)}(1+\log m+\log d) \norm{T}_{\mathrm{op}},
\]
where $\norm{\cdot}_{\mathrm{op}}$ denotes the operator norm.
\end{lemma}

Lemma~\ref{lemma:maurey} extends to the complex field case, which is shown in the following corollary.

\begin{corollary}
\label{cor:maurey}
Let $k,m,n\in \mathbb{N}$ and $T:\ell_1^k(\mathbb{C}) \to \ell_\infty^m(\ell_2^d(\mathbb{C}))$ be a linear operator.
Then
\[
\int_0^\infty \sqrt{\log N(T(B_{\ell_1^k(\mathbb{C})}), \norm{\cdot}_{\ell_\infty^m(\ell_2^d(\mathbb{C}))}, t)} dt
\leq C \sqrt{(1+\log k)(1+\log m)}(1+\log m+\log d) \norm{T}_{\mathrm{op}}.
\]
\end{corollary}

\begin{proof}[Proof of Corollary~\ref{cor:maurey}]
Let $\iota: \mathbb{C} \to \mathbb{R}^2$ be a natural map defined by
\[
\iota(x) = [\mathrm{Re}(x), \mathrm{Im}(x)]^\transpose, \quad \forall x \in \mathbb{C}.
\]
By a slight abuse of notation, we apply $\iota$ elementwise to $\mathbb{C}^k$, i.e.
\[
\iota([x_1,\dots,x_k]^\transpose) =
[\mathrm{Re}(x_1), \mathrm{Im}(x_1), \dots, \mathrm{Re}(x_k), \mathrm{Im}(x_k)]^\transpose.
\]
Then we have
\[
\iota B_{\ell_1^k(\mathbb{C})} \subset \sqrt{2} B_{\ell_1^{2k}(\mathbb{R})}
\]
and
\[
\iota B_{\ell_\infty^m(\ell_2^d(\mathbb{C}))} = B_{\ell_\infty^m(\ell_2^{2d}(\mathbb{R}))}.
\]
Moreover, $\iota$ is bijective.
Therefore,
\[
\int_0^\infty \sqrt{\log N(T(B_{\ell_1^k(\mathbb{C})}), \norm{\cdot}_{\ell_\infty^m(\ell_2^d(\mathbb{C}))}, t)} dt
\leq
\int_0^\infty \sqrt{\log N(T(\sqrt{2} B_{\ell_1^{2k}(\mathbb{R})}), \norm{\cdot}_{\ell_\infty^m(\ell_2^{2d}(\mathbb{R}))}, t)} dt
\]
Then the assertion follows from Lemma~\ref{lemma:maurey} with a change of variable in the integral.
\end{proof}

\subsection{Suprema of Gaussian Processes}

We use the following lemma that provides tail estimates of suprema of first order chaos processes.
\begin{lemma}
\label{lemma:firstorder}
Let $\vxi \in \cz^n$ be a standard Gaussian vector with $\mathbb{E}[\vxi \vxi^*] = \mId_n$, $\Delta \subset \cz^n$, and $0 < \zeta < e^{-1/2}$.
There is an absolute constants $C$ such that
\[
\sup_{\vf \in \Delta} |\vf^* \vxi|
\leq C \sqrt{\log(\zeta^{-1})} \int_0^\infty \sqrt{\log N(\Delta,\norm{\cdot}_2,t)} dt
\]
holds with probability $1-\zeta$.
\end{lemma}

\begin{proof}
Lemma~\ref{lemma:firstorder} is a direct consequence of the moment version of Dudley's inequality \cite[p. 263]{foucart2013mathematical} and a version of Markov's inequality \cite[Proposition~7.11]{foucart2013mathematical}.
\end{proof}

\section{Expectations}

\begin{lemma}
\label{lemma:expectation1}
Under the assumption in (A1),
\[
\mathbb{E}[\mC_{\widetilde{\mPhi}_m \vu_m}^* \mC_{\widetilde{\mPhi}_m \vu_m}]
= K \norm{\vu_m}_2^2 \mId_L.
\]
\end{lemma}
\begin{proof}[Proof of Lemma~\ref{lemma:expectation1}]
\[
\mathbb{E}[\mC_{\widetilde{\mPhi}_m \vu_m}^* \mC_{\widetilde{\mPhi}_m \vu_m}]
= \norm{\vu_m}_2^2
\mathbb{E}[\mC_{\mS^* \vg}^* \mC_{\mS^* \vg}],
\]
where $\vg \in \mathbb{C}^K$ is a standard complex Gaussian vector.
Let $g_k$ denote the $k$th entry of $\vg$.
Since $\mathbb{E}[g_k \overline{g_l}] = 0$ for all $k \neq l$, we have
\[
\mathbb{E}[\mC_{\mS^* \vg}^* \mC_{\mS^* \vg}]
= \sum_{k=1}^K \mC_{\ve_k}^* \mC_{\ve_k} = K \mId_L.
\]
This completes the proof.
\end{proof}

\begin{lemma}
\label{lemma:expectation2}
Under the assumption in (A1),
\[
\mathbb{E}[\mC_{\widetilde{\mPhi}_m \vu_m}^* \widetilde{\mPhi}_m]
= K \ve_1 \vu_m^*.
\]
\end{lemma}
\begin{proof}[Proof of Lemma~\ref{lemma:expectation2}]
By (A1), it follows that $\widetilde{\mPhi}_m \vu_m$ and $\widetilde{\mPhi}_m \vq$ are independent for any $\vq \in \mathbb{C}^K$ such that $\vq^* \vu_m = 0$.
Therefore,
\begin{align*}
\mathbb{E}[\mC_{\widetilde{\mPhi}_m \vu_m}^* \widetilde{\mPhi}_m]
& = \mathbb{E}\Big[\mC_{\widetilde{\mPhi}_m \vu_m}^* \widetilde{\mPhi}_m \frac{\vu_m \vu_m^*}{\norm{\vu_m}_2^2}\Big]
= \frac{\mathbb{E}[\mC_{\widetilde{\mPhi}_m \vu_m}^* \widetilde{\mPhi}_m \vu_m]}{\norm{\vu_m}_2^2} \vu_m^* \\
& = \frac{\mathbb{E}[\mC_{\widetilde{\mPhi}_m \vu_m}^* \mC_{\widetilde{\mPhi}_m \vu_m} \ve_1]}{\norm{\vu_m}_2^2} \vu_m^*
= K \ve_1 \vu_m^*,
\end{align*}
where the last step follows from Lemma~\ref{lemma:expectation1}.
\end{proof}

\begin{lemma}
\label{lemma:expectation3}
Under the assumption in (A1),
\[
\mathbb{E}[\widetilde{\mPhi}_m^* \mC_{\widetilde{\mPhi}_{m'} \vu_{m'}}^* \mC_{\vx}^* \mC_{\vx} \mC_{\widetilde{\mPhi}_{m'} \vu_{m'}} \widetilde{\mPhi}_m]
= \begin{cases}
K^2 \norm{\vx}_2^2 \norm{\vu_{m'}}_2^2 \mId_D & m\neq m', \\
K^2 \norm{\vx}_2^2 (\norm{\vu_{m'}}_2^2 \mId_D + \vu_{m'} \vu_{m'}^*) & m = m'.
\end{cases}
\]
\end{lemma}
\begin{proof}[Proof of Lemma~\ref{lemma:expectation3}]
Suppose that $m \neq m'$. Then, by the independence of $\mPhi_m$ and $\mPhi_{m'}$, it follows that
\begin{align*}
\mathbb{E}[\widetilde{\mPhi}_m^* \mC_{\widetilde{\mPhi}_{m'} \vu_{m'}}^* \mC_{\vx}^* \mC_{\vx} \mC_{\widetilde{\mPhi}_{m'} \vu_{m'}} \widetilde{\mPhi}_m]
& =
\mathbb{E}_{\mPhi_m}\Big[\widetilde{\mPhi}_m^* \mathbb{E}_{\mPhi_{m'}}\Big[ \mC_{\widetilde{\mPhi}_{m'} \vu_{m'}}^* \mC_{\vx}^* \mC_{\vx} \mC_{\widetilde{\mPhi}_{m'} \vu_{m'}} \Big] \widetilde{\mPhi}_m\Big] \\
& = K \norm{\vu_{m'}}_2^2 \mathbb{E}_{\mPhi_m}[\widetilde{\mPhi}_m^* \widetilde{\mPhi}_m] \\
& = K \norm{\vu_{m'}}_2^2 \mathbb{E}_{\mPhi_m}[\mPhi_m^* \mS \mS^* \mPhi_m] \\
& = K \norm{\vu_{m'}}_2^2 \mathrm{tr}(\mS \mS^*) \mId_D
= K^2 \norm{\vu_{m'}}_2^2 \mId_D,
\end{align*}
where the second identity holds by Lemma~\ref{lemma:expectation1}.
This proves the first case.

Next we assume that $m' = m$.
For notational simplicity, let $\mP \in \mathbb{C}^{D \times D}$ denote the orthogonal projection onto the span of $\vu_m$, i.e.
\[
\mP = \frac{\vu_m \vu_m^*}{\norm{\vu_m}_2^2}.
\]
Then by (A1) it follows that $\mPhi_m \vu_m$ and $\mPhi_m (\mId_D - \mP)$ are independent.
Therefore
\begin{align*}
& \mathbb{E}[\widetilde{\mPhi}_m^* \mC_{\widetilde{\mPhi}_m \vu_m}^* \mC_{\vx}^* \mC_{\vx} \mC_{\widetilde{\mPhi}_m \vu_m} \widetilde{\mPhi}_m] \\
& =
\mathbb{E}[(\mP + \mId_D - \mP) \widetilde{\mPhi}_m^* \mC_{\widetilde{\mPhi}_m \vu_m}^* \mC_{\vx}^* \mC_{\vx} \mC_{\widetilde{\mPhi}_m \vu_m} \widetilde{\mPhi}_m (\mP + \mId_D - \mP)] \\
& =
\underbrace{
\mathbb{E}[\mP \widetilde{\mPhi}_m^* \mC_{\widetilde{\mPhi}_m \vu_m}^* \mC_{\vx}^* \mC_{\vx} \mC_{\widetilde{\mPhi}_m \vu_m} \widetilde{\mPhi}_m \mP]
}_{\text{($\star$)}} \\
& +
\underbrace{
\mathbb{E}[(\mId_D - \mP) \widetilde{\mPhi}_m^* \mC_{\widetilde{\mPhi}_m \vu_m}^* \mC_{\vx}^* \mC_{\vx} \mC_{\widetilde{\mPhi}_m \vu_m} \widetilde{\mPhi}_m (\mId_D - \mP)]
}_{\text{($\star\star$)}}.
\end{align*}

The first term ($\star$) is rewritten as
\[
\mathbb{E}[\vg^* \mS \mC_{\mS^* \vg}^* \mC_{\vx}^* \mC_{\vx} \mC_{\mS^* \vg} \mS^* \vg \vu_m \vu_m^*],
\]
where $\vg \in \mathbb{C}^K$ is a standard complex Gaussian vector.
On the other hand, we have
\begin{align*}
\mathbb{E}[\vg^* \mS \mC_{\mS^* \vg}^* \mC_{\vx}^* \mC_{\vx} \mC_{\mS^* \vg} \mS^* \vg]
& = \sum_{j_1,j_2,j_3,j_4=1}^K \mathbb{E}[\overline{g_{j_1}} \overline{g_{j_2}} g_{j_3} g_{j_4}] \ve_{j_1}^* \mC_{\ve_{j_2}}^* \mC_{\vx}^* \mC_{\vx} \mC_{\ve_{j_3}} \ve_{j_4} \\
& = 2K \norm{\vx}_2^2 + 2K(K-1) \norm{\vx}_2^2
= 2K^2 \norm{\vx}_2^2.
\end{align*}
Therefore,
\[
\text{($\star$)} = 2K^2 \norm{\vx}_2^2 \vu_m \vu_m^*.
\]

By the independence of $\mPhi_m \vu_m$ and $\mPhi_m (\mId_D - \mP)$ together with the commutativity of convolution, the second term ($\star\star$) is computed as
\begin{align*}
\text{($\star\star$)}
& = \mathbb{E}[(\mId_D - \mP) \widetilde{\mPhi}_m^*
\mathbb{E}[\mC_{\widetilde{\mPhi}_m \vu_m}^* \mC_{\widetilde{\mPhi}_m \vu_m}] \mC_{\vx}^* \mC_{\vx}  \widetilde{\mPhi}_m (\mId_D - \mP)] \\
& = K \norm{\vu_m}_2^2 \mathbb{E}[(\mId_D - \mP) \widetilde{\mPhi}_m^* \mC_{\vx}^* \mC_{\vx} \widetilde{\mPhi}_m (\mId_D - \mP)] \\
& = K \norm{\vu_m}_2^2 \mathrm{tr}(\mS \mC_{\vx}^* \mC_{\vx} \mS^*) (\mId_D - \mP) \\
& = K^2 \norm{\vx}_2^2 (\norm{\vu_m}_2^2 \mId_D - \vu_m \vu_m^*).
\end{align*}
Collecting these results proves the second case and the proof is done.
\end{proof}

\section{Proof of Lemmas in Section~\ref{sec:mainres}}
\subsection{Proof of Lemma~\ref{lemma:rho_wx}}
\label{sec:proof:lemma:rho_wx}
By the homogeneity of $\rho_{x,w}$, we may assume that $\sigma_w = 1$.
Let
\[
\va_{\vz,\vq} = \mC_{\widetilde{\mS}^* \vz}^* \mC_{\vx} \widetilde{\mS}^* \vq,
\]
where $\widetilde{\mS}$ is defined in \eqref{eq:deftildeS}.  Then
\begin{align*}
\norm{\widetilde{\mS} \mC_{\vx}^* \mC_{\vw_m} \widetilde{\mS}^*}
&= \sup_{\vz, \vq \in B_2^{3K}} |\va_{\vz,\vq}^* \vw_m|.
\end{align*}
We use Lemma~\ref{lemma:firstorder} to get a tail estimate.

Since
\begin{align*}
\norm{\va_{\vz,\vq} - \va_{\vz',\vq'}}_2
&\leq \norm{\va_{\vz,\vq} - \va_{\vz,\vq'}}_2
+ \norm{\va_{\vz,\vq'} - \va_{\vz',\vq'}}_2 \\
&\leq \norm{\mC_{\widetilde{\mS}^* \vz}^*} \norm{\mC_{\vx} \widetilde{\mS}^*} \norm{\vq - \vq'}_2
+ \norm{\mC_{\widetilde{\mS}^* \vq'}} \norm{\mC_{\vx}^* \widetilde{\mS}^*} \norm{\vz - \vz'}_2 \\
&\leq \sqrt{K \rho_x} (\norm{\vq - \vq'}_2 + \norm{\vz - \vz'}_2),
\end{align*}
it follows that the $\gamma_2$ functional of the set $\{\va_{\vz,\vq} ~|~ \vz, \vq \in B_2^{3K} \}$ is upper bounded by
\begin{align*}
C_1 \sqrt{K \rho_x} \int_0^\infty \sqrt{\log N(B_2^{3K}, \norm{\cdot}_2, t)} dt
\leq C_2 K \sqrt{\rho_x},
\end{align*}
where we used a standard volume argument.
The assertion follows from Lemma~\ref{lemma:firstorder} and a union bound argument.

\subsection{Proof of Lemma~\ref{lemma:rho_wx_randx}}
\label{sec:proof:lemma:rho_wx_randx}
By the homogeneity, we may assume that $\sigma_x = \sigma_w = 1$.
Then
\begin{align*}
\norm{\widetilde{\mS} \mC_{\vx}^* \mC_{\vx}  \widetilde{\mS}^*}
= \sup_{\vz \in B_2^{3K}}
| \vx^* \mC_{\widetilde{\mS}^* \vz}^* \mC_{\widetilde{\mS}^* \vz} \vx |
\end{align*}
and
\begin{align*}
\norm{\widetilde{\mS} \mC_{\vx}^* \mC_{\vw_m}  \widetilde{\mS}^*}
= \sup_{\vz, \vq \in B_2^{3K}}
| \vx^* \mC_{\widetilde{\mS}^* \vq}^* \mC_{\widetilde{\mS}^* \vz} \vw_m |
\end{align*}
are suprema of second order chaos processes.
We estimate their tail decay using Theorems~\ref{thm:kmr} and \ref{thm:ip}.
(For more details, see the proof of Lemma~\ref{lemma:rho_w}.)
By Theorem~\ref{thm:kmr} and a union bound argument, it follows that
\[
\rho_x \leq C_\beta \sigma_x^2 \sqrt{KL} \log^5 (KL)
\]
holds with probability $1-K^{-\beta}$.
Similarly, by Theorem~\ref{thm:ip},
\[
\rho_{x,w} \leq C_\beta \sigma_w \sigma_x \sqrt{KL} \log^5 (MKL)
\]
holds with probability $1-K^{-\beta}$.

Moreover, by Lemma~\ref{lemma:chisq}, we also have that
\[
\norm{\vx}_2^2 \geq \sigma_x^2 (L - \sqrt{2L \beta \log K})
\]
holds with probability $1-K^{-\beta}$.
The assertions follow by assembling the above estimates.

\section{Proof of Lemmas in Section~\ref{sec:proof_tech_lemma}}
\subsection{Proof of Lemma~\ref{lemma:gaussquadasym}}
\label{sec:proof:lemma:gaussquadasym}

First note that $\norm{\mPsi^* \mA \mPsi - \mathbb{E}[\mPsi^* \mA \mPsi]}$ is written as a variational form given by
\begin{equation}
\label{eq:varform1}
\max_{\vq,\vq' \in B_2^D} |\vq^* (\mPsi^* \mA \mPsi - \mathbb{E}[\mPsi^* \mA \mPsi]) \vq'|.
\end{equation}

For all $\vq \in \mathbb{C}^D$, we have
\[
\mPsi \vq = (\vq^\transpose \otimes \mId_K) \mathrm{vec}(\mPsi).
\]
Let $\vpsi = \mathrm{vec}(\mPsi)$. Then
\[
\vq^* \mPsi^* \mA \mPsi \vq'
= \vpsi^* (\vq^\transpose \otimes \mId_K)^* \mA ((\vq')^\transpose \otimes \mId_K) \vpsi
= \vpsi^* (\overline{\vq} (\vq')^\transpose \otimes \mA) \vpsi.
\]

Then \eqref{eq:varform1} becomes the supremum of the second order chaos process
\[
\vpsi^* (\overline{\vq} (\vq')^\transpose \otimes \mA) \vpsi - \mathbb{E}[\vpsi^* (\overline{\vq} (\vq')^\transpose \otimes \mA) \vpsi]
\]
over $\vq \in B_2^D$. We obtain its tail estimate by applying Theorem~\ref{thm:ip} with
\[
\Delta_1 = \{ \vq^\transpose \otimes \mId_K ~|~ \vq \in B_2^D \}
\]
and
\[
\Delta_2 = \{ (\vq')^\transpose \otimes \mA ~|~ \vq' \in B_2^D \}.
\]

By direct calculation, we obtain
\[
d_{\mathrm{S}}(\Delta_1) \leq 1, \quad
d_{\mathrm{F}}(\Delta_1) \leq \sqrt{K},
\]
and
\[
d_{\mathrm{S}}(\Delta_2) \leq \norm{\mA}, \quad
d_{\mathrm{F}}(\Delta_2) \leq \sqrt{K} \norm{\mA}.
\]
Moreover
\begin{align*}
\gamma_2(\Delta_1,\norm{\cdot})
\leq C_1 \int_0^\infty \sqrt{N(B_2^D,\norm{\cdot}_2,t)} dt
\leq C_1 \int_0^1 \sqrt{2D \log \Big(1+\frac{2}{t}\Big)} dt
\leq C_2 \sqrt{D}.
\end{align*}
Similarly, we also have
\[
\gamma_2(\Delta_2,\norm{\cdot}) \leq C_2 \sqrt{D} \norm{\mA}.
\]
The assertion follows from Theorem~\ref{thm:ip} with
\begin{equation}
\label{eq:a4gm}
a = \sqrt{\frac{\gamma_2(\Delta_2,\norm{\cdot}) d_{\mathrm{F}}(\Delta_2)}{\gamma_2(\Delta_1,\norm{\cdot}) d_{\mathrm{F}}(\Delta_1)}}.
\end{equation}

\subsection{Proof of Lemma~\ref{lemma:Upsilon1m}}
\label{sec:proof:lemma:Upsilon1m}

Note that $\norm{\mUpsilon_{\mathrm{s},m} - \mathbb{E}[\mUpsilon_{\mathrm{s},m}]}$ is written as a variational form given by
\[
\sup_{\vz \in B_2^K} | \vz^* (\mUpsilon_{\mathrm{s},m} - \mathbb{E}[\mUpsilon_{\mathrm{s},m}]) \vz |.
\]

By the commutativity of convolution, we have
\begin{align*}
\vz^* \mUpsilon_{\mathrm{s},m} \vz
&= \sum_{\begin{subarray}{c} m'=1 \\ m' \neq m \end{subarray}}^M
\vz^* \mS \mC_{\mS^* \mPhi_{m'} \vu_{m'}}^* \mC_{\vx}^* \mC_{\vx} \mC_{\mS^* \mPhi_{m'} \vu_{m'}} \mS^* \vz \\
&= \sum_{\begin{subarray}{c} m'=1 \\ m' \neq m \end{subarray}}^M
\vu_{m'}^* \mPhi_{m'}^* \mS \mC_{\mS^* \vz}^* \mC_{\vx}^* \mC_{\vx} \mC_{\mS^* \vz} \mS^* \mPhi_{m'} \vu_{m'} \\
&= \sum_{\begin{subarray}{c} m'=1 \\ m' \neq m \end{subarray}}^M
\mathrm{vec}(\mPhi_{m'})^* (\overline{\vu_{m'}} \otimes \mId_K) \mS \mC_{\mS^* \vz}^* \mC_{\vx}^* \mC_{\vx} \mC_{\mS^* \vz} \mS^* (\vu_{m'}^\transpose \otimes \mId_K) \mathrm{vec}(\mPhi_{m'}) \\
&= \sum_{\begin{subarray}{c} m'=1 \\ m' \neq m \end{subarray}}^M
\mathrm{vec}(\mPhi_{m'})^*
(\overline{\vu_{m'}} \vu_{m'}^\transpose \otimes \mS \mC_{\mS^* \vz} \mC_{\vx}^* \mC_{\vx} \mC_{\mS^* \vz}^* \mS^*) \mathrm{vec}(\mPhi_{m'}),
\end{align*}
where the third identity follows from
\begin{equation}
\label{eq:id_Phiu}
\mPhi_{m'} \vu_{m'} = (\vu_{m'}^\transpose \otimes \mId_K) \mathrm{vec}(\mPhi_{m'}).
\end{equation}

Let
\[
\mQ(\vz) = \sum_{\begin{subarray}{c} m'=1 \\ m' \neq m \end{subarray}}^M
\ve_{m'} \ve_{m'}^* \otimes \vu_{m'}^\transpose \otimes \mC_{\vx} \mC_{\mS^* \vz}^* \mS^*,
\]
and
\[
\vphi = [\mathrm{vec}(\mPhi_1)^\transpose,\dots,\mathrm{vec}(\mPhi_M)^\transpose]^\transpose.
\]
Then $\vphi$ follows the distribution $\mathcal{CN}(\vzero_{MKD,1},\mId_{MKD})$ and
\[
\vz^* \mUpsilon_{\mathrm{s},m} \vz = \vphi^* \mQ(\vz)^* \mQ(\vz) \vphi.
\]

Therefore,
\[
\sup_{\vz \in B_2^K} |\vz^* (\mUpsilon_{\mathrm{s},m} - \mathbb{E}[\mUpsilon_{\mathrm{s},m}]) \vz|
= \sup_{\vz \in B_2^K} |\vphi^* \mQ(\vz) \mQ(\vz) \vphi|.
\]

We get a tail bound of the supremum of the second order chaos process $\vphi^* \mQ(\vz)^* \mQ(\vz) \vphi - \mathbb{E}[\vphi^* \mQ(\vz)^* \mQ(\vz) \vphi]$ by applying Theorem~\ref{thm:kmr} with $\mM = \mQ(\vz)$ and $\Delta = \{ \mQ(\vz) | \vz \in B_2^K  \}$.

Recall that in Section~\ref{sec:proof:lemma:Es} we defined $\breve{\mS}$ by
\[
\breve{\mS} =
\begin{bmatrix}
\bm{0}_{K-1,L-K+1} & \mId_{K-1} \\
\mId_{K} & \bm{0}_{K,L-K}
\end{bmatrix}.
\]

Then the radius of $\Delta$ with respect to the Frobenius norm is upper bounded by
\begin{align*}
d_{\mathrm{F}}(\Delta)
&\leq \sup_{\vz \in B_2^K} \norm{\vu}_2 \norm{\mC_{\vx} \breve{\mS}^*} \norm{\mS \mC_{\mS^* \vz}}_{\mathrm{F}} \\
&\leq \sup_{\vz \in B_2^K} \norm{\vu}_2 \norm{\breve{\mS} \mC_{\vx}^* \mC_{\vx} \breve{\mS}^*}^{1/2} \sqrt{K} \norm{\vz}_2
\leq \norm{\vu}_2 \sqrt{\rho_x K},
\end{align*}
where the first inequality follows from the identity $\mC_{\mS^* \vz}^* \mS^* = \breve{\mS}^* \breve{\mS} \mC_{\mS^* \vz}^* \mS^*$.

Let $\norm{\cdot}_{2,\infty}$ be defined in \eqref{eq:defblknorm}.  Then the radius of $\Delta$ with respect to the spectral norm is upper bounded by
\begin{align}
d_{\mathrm{S}}(\Delta)
&\leq \sup_{\vz \in B_2^K} \norm{\vu}_{2,\infty} \norm{\mC_{\vx} \breve{\mS}^*} \norm{\mS \mC_{\mS^* \vz}} \nonumber \\
&\leq \sup_{\vz \in B_2^K} \norm{\vu}_{2,\infty} \sqrt{\rho_x} \sqrt{L} \norm{\mF \mS^* \vz}_\infty \label{eq:ubsnorm1} \\
&\leq \sup_{\vz \in B_2^K} \norm{\vu}_{2,\infty} \sqrt{\rho_x} \norm{\vz}_1 \leq \norm{\vu}_{2,\infty} \sqrt{\rho_x K}, \nonumber
\end{align}
where the second inequality follows from the identity $\mC_{\mS^* \vz} = \sqrt{L} \mF^* \mathrm{diag}(\mF \mS^* \vz) \mF$ and the third inequality follow from $\norm{\mF: \ell_1^L \to \ell_\infty^L} = \frac{1}{\sqrt{L}}$.

By \eqref{eq:ubsnorm1} and Dudley's inequality, the $\gamma_2$ functional of $\Delta$ is upper bounded by
\begin{align*}
\gamma_2(\Delta,\norm{\cdot})
&\leq C_1 \norm{\vu}_{2,\infty} \sqrt{\rho_x} \sqrt{L}
\int_0^\infty \sqrt{\log N(\mF \mS^* B_2^K, \norm{\cdot}_\infty, t)} dt \\
&\leq C_2 \norm{\vu}_{2,\infty} \sqrt{\rho_x} \sqrt{K} \sqrt{\log K} \log^{3/2}L,
\end{align*}
where the last step follows from Corollary~\ref{cor:maurey}.
The assertion follows by applying these estimates to Theorem~\ref{thm:kmr} together with a union bound argument over $m=1,\dots,M$.

\subsection{Proof of Lemma~\ref{lemma:comp}}
\label{sec:proof:lemma:comp}

Note that $\norm{\mZ_m - \mathbb{E}[\mZ_m]}$ is written as a variational form given by
\begin{equation}
\label{eq:varform2}
\sup_{\vz \in B_2^{2L-1}, \vq \in B_2^D}
|\vz^* \mZ_m \vq - \mathbb{E}[\vz^* \mZ_m \vq]|.
\end{equation}

Let $\vphi_m = \mathrm{vec}(\mPhi_m)$. Then
\[
\vz^* \mZ_m \vq
= \vz^* \breve{\mS} \mC_{\mS^* \mPhi_m \vu_m}^* \mS^* \mPhi_m \vq
= \vphi_m^* (\overline{\vu_m} \vq^\transpose \otimes \mS \mC_{\breve{\mS}^* \vz}^* \mS^*) \vphi_m.
\]

Then \eqref{eq:varform2} becomes the supremum of the second order chaos process
\[
\vphi_m^* (\vu_m^\transpose \otimes \mS \mC_{\breve{\mS}^* \vz}^* \mS^*)^* (\vq^\transpose \otimes \mId_K) \vphi_m
- \mathbb{E}[\vphi_m^* (\vu_m^\transpose \otimes \mS \mC_{\breve{\mS}^* \vz}^* \mS^*)^* (\vq^\transpose \otimes \mId_K) \vphi_m]
\]
over $\vz \in B_2^{2K-1}$ and $\vq \in B_2^D$. We obtain its tail estimate by applying Theorem~\ref{thm:ip} with
\[
\Delta_1 = \{ \vu_m^\transpose \otimes \mS \mC_{\breve{\mS}^* \vz}^* \mS^* ~|~ \vz \in B_2^{2K-1} \}
\]
and
\[
\Delta_2 = \{ \vq^\transpose \otimes \mId_K ~|~ \vq \in B_2^D \}.
\]

The radii of $\Delta_1$ and $\Delta_2$ are upper bounded by
\[
d_{\mathrm{S}}(\Delta_1) \leq d_{\mathrm{F}}(\Delta_1) \leq \norm{\vu}_{2,\infty} \sqrt{K}
\]
and
\[
d_{\mathrm{S}}(\Delta_2) \leq 1,
\quad
d_{\mathrm{F}}(\Delta_2) \leq \sqrt{K}.
\]

Moreover, since
\[
\norm{\vu_m^\transpose \otimes \mS \mC_{\breve{\mS}^* \vz}^* \mS^*}
\leq \norm{\vu}_{2,\infty} \sqrt{L} \norm{\mF \breve{\mS}^* \vz}_\infty,
\]
we have
\begin{align*}
\gamma_2(\Delta_1,\norm{\cdot})
&\leq C_1 \norm{\vu}_{2,\infty} \sqrt{L} \int_0^\infty \sqrt{N(\mF \breve{\mS}^* B_2^{2K-1},\norm{\cdot}_\infty,t)} dt \\
&\leq C_2 \norm{\vu}_{2,\infty} \sqrt{LK} \int_0^\infty \sqrt{N(\mF \breve{\mS}^* B_1^{2K-1},\norm{\cdot}_\infty,t)} dt \\
&\leq C_3 \norm{\vu}_{2,\infty} \sqrt{K} \sqrt{\log (2K-1)} \log^{3/2}L,
\end{align*}
where the last step follows from Corollary~\ref{cor:maurey}.

On the other hand, we also have
\begin{align*}
\gamma_2(\Delta_2,\norm{\cdot})
\leq C_1 \int_0^\infty \sqrt{N(B_2^D,\norm{\cdot}_2,t)} dt
\leq C_1 \int_0^1 \sqrt{2D \log \Big(1+\frac{2}{t}\Big)} dt
\leq C_4 \sqrt{D}.
\end{align*}

By applying the estimates to Theorem~\ref{thm:ip} with $a$ given in \eqref{eq:a4gm}, we obtain that the supremum is upper bounded by
\[
C(
\norm{\vu}_{2,\infty} K^{3/2} \sqrt{D} \sqrt{\log K} \log^{3/2}L
+ \norm{\vu}_{2,\infty} K \sqrt{\log K} \log^{3/2}L
+ \norm{\vu}_{2,\infty} \sqrt{KD}
) \log(\log M + \beta \log K)
\]
with probability $1-M^{-1}K^{-\beta}$.

The assertion follows by applying a union bound argument over $m=1,\dots,M$.

\subsection{Proof of Lemma~\ref{lemma:comp_sum}}
\label{sec:proof:lemma:comp_sum}

Note that the spectral norm of $\sum_{m=1}^M \ve_m^* \otimes (\mZ_m - \mathbb{E}[\mZ_m])$ admits a variational form given by
\begin{equation}
\label{eq:varform3}
\sup_{\vz \in B_2^{2L-1}, \vv \in B_2^{MD}}
\Big|
\sum_{m=1}^M \vz^* \mZ_m \vv_m - \mathbb{E}[\vz^* \mZ_m \vv_m]
\Big|,
\end{equation}
where $\vv_m \in \mathbb{C}^D$ for $m=1,\dots,M$ and $\vv = [\vv_1^\transpose,\dots,\vv_M^\transpose]^\transpose$.

Let $\vphi_m = \mathrm{vec}(\mPhi_m)$ for $m=1,\dots,M$
and $\vphi = [\vphi_1^\transpose,\dots,\vphi_M^\transpose]^\transpose$. Then as before
\begin{align*}
\sum_{m=1}^M \vz^* \mZ_m \vv_m
&= \sum_{m=1}^M \vz^* \breve{\mS} \mC_{\mS^* \mPhi_m \vu_m}^* \mS^* \mPhi_m \vv_m \\
&= \sum_{m=1}^M \vphi_m^* (\overline{\vu_m} \vv_m^\transpose \otimes \mS \mC_{\breve{\mS}^* \vz}^* \mS^*) \vphi_m \\
&= \vphi^*
\Big(\sum_{m=1}^M \ve_m \ve_m^* \otimes \vu_m^\transpose \otimes \mS \mC_{\breve{\mS}^* \vz}^* \mS^*\Big)^*
\Big(\sum_{m=1}^M \ve_m \ve_m^* \otimes \vv_m^\transpose \otimes \mId_K\Big)
\vphi.
\end{align*}

Then \eqref{eq:varform3} becomes the supremum of a second order chaos process. We obtain its tail estimate by applying Theorem~\ref{thm:ip} with
\[
\Delta_1 = \Big\{ \sum_{m=1}^M \ve_m \ve_m^* \otimes \vu_m^\transpose \otimes \mS \mC_{\breve{\mS}^* \vz}^* \mS^* ~|~ \vz \in B_2^{2K-1} \Big\}
\]
and
\[
\Delta_2 = \Big\{ \sum_{m=1}^M \ve_m \ve_m^* \otimes \vv_m^\transpose \otimes \mId_K ~|~ \vv \in B_2^{MD} \Big\}.
\]

Let $\norm{\cdot}_{2,\infty}$ be defined in \eqref{eq:defblknorm}.  Then the radii of $\Delta_1$ and $\Delta_2$ are upper bounded by
\[
d_{\mathrm{S}}(\Delta_1) \leq \norm{\vu}_{2,\infty} \sqrt{K}, \quad
d_{\mathrm{F}}(\Delta_1) \leq \norm{\vu}_2 \sqrt{K} \leq \norm{\vu}_{2,\infty} \sqrt{MK},
\]
and
\[
d_{\mathrm{S}}(\Delta_2) \leq 1,
\quad
d_{\mathrm{F}}(\Delta_2) \leq \sqrt{K}.
\]

Moreover, since
\[
\Big\|
\sum_{m=1}^M \ve_m \ve_m^* \otimes \vu_m^\transpose \otimes \mS \mC_{\breve{\mS}^* \vz}^* \mS^*
\Big\|
\leq \norm{\vu}_{2,\infty} \sqrt{L} \norm{\mF \breve{\mS}^* \vz}_\infty,
\]
we have
\begin{align*}
\gamma_2(\Delta_1,\norm{\cdot})
&\leq C_1 \norm{\vu}_{2,\infty} \sqrt{L} \int_0^\infty \sqrt{N(\mF \breve{\mS}^* B_2^{2K-1},\norm{\cdot}_\infty,t)} dt \\
&\leq C_2 \norm{\vu}_{2,\infty} \sqrt{LK} \int_0^\infty \sqrt{N(\mF \breve{\mS}^* B_1^{2K-1},\norm{\cdot}_\infty,t)} dt \\
&\leq C_3 \norm{\vu}_{2,\infty} \sqrt{K} \sqrt{\log (2K-1)} \log^{3/2}L,
\end{align*}
where the last step follows from Corollary~\ref{cor:maurey}.

On the other hand, since
\[
\Big\|
\sum_{m=1}^M \ve_m \ve_m^* \otimes \vv_m^\transpose \otimes \mId_K
\Big\|
\leq \norm{\vv}_{2,\infty},
\]
we also have
\begin{align*}
\gamma_2(\Delta_2,\norm{\cdot})
\leq C_1 \int_0^\infty \sqrt{N(B_2^{MD},\norm{\cdot}_{2,\infty},t)} dt
\leq C_4 \sqrt{D} \sqrt{\log D} \log(MD).
\end{align*}

By applying these estimates to Theorem~\ref{thm:ip} with $a$ given in \eqref{eq:a4gm}, we obtain that the supremum is upper bounded by
\[
C'(\beta) \log^\alpha(MKL) \norm{\vu}_{2,\infty}
(M^{1/4} K^{3/4} D^{1/4} + K + \sqrt{MKD})
\]
with probability $1-K^{-\beta}$.
Finally, by the inequality of arithmetic and geometric means, we have
\[
M^{1/4} K^{3/4} D^{1/4} \leq \frac{K + \sqrt{MKD}}{2}.
\]
This completes the proof.

\subsection{Proof of Lemma~\ref{lemma:Upsiloncm}}
\label{sec:proof:lemma:Upsiloncm}

Similar to the proof of Lemma~\ref{lemma:Upsilon1m}, we rewrite $\norm{\mUpsilon_{\mathrm{c},m}}$ as a variational form given by
\[
\sup_{\vz_1,\vz_2 \in B_2^K} |\vz_1^* \mUpsilon_{\mathrm{c},m} \vz_2|.
\]

By the commutativity of convolution and \eqref{eq:id_Phiu},
$\vz_1^* \mUpsilon_{\mathrm{c},m} \vz_2$ is written as
 follows
\begin{align*}
\sum_{\begin{subarray}{c} m'=1 \\ m' \neq m \end{subarray}}^M
\vz_1^* \mS \mC_{\mS^* \mPhi_{m'} \vu_{m'}}^* \mC_{\vx}^* \mC_{\vw_{m'}} \mS^* \vz_2
= \sum_{\begin{subarray}{c} m'=1 \\ m' \neq m \end{subarray}}^M
\mathrm{vec}(\mPhi_{m'})^*
(\overline{\vu_{m'}} \otimes \mS \mC_{\mS^* \vz_1} \mC_{\vx}^* \mC_{\vw_{m'}} \mS^* \vz_2),
\end{align*}
which, conditional on $\vw$, is a centered Gaussian process indexed by $(\vz_1,\vz_2) \in B_2^K \times B_2^K$.
We compute a tail estimate of the supremum by applying Lemma~\ref{lemma:firstorder} with
\[
\Delta = \{ \vf(\vz_1,\vz_2) ~|~ \vz_1,\vz_2 \in B_2^K \},
\]
where
\[
\vf(\vz_1,\vz_2) = \sum_{\begin{subarray}{c} m'=1 \\ m' \neq m \end{subarray}}^M
\ve_{m'} \otimes \overline{\vu_{m'}} \otimes \mS \mC_{\vx}^* \mC_{\vw_{m'}} \mC_{\mS^* \vz_1} \mS^* \vz_2.
\]

Since $\|\mC_{\mS^* (\vz_j-\vz'_j)}\|=\sqrt{L}\norm{\mF \mS^* (\vz_j-\vz'_j)}_\infty $ for $j=1,2$ and hence
\begin{align*}
\norm{\vf(\vz_1,\vz_2) - \vf(\vz'_1,\vz'_2)}_2
&\leq \norm{\vf(\vz_1,\vz_2) - \vf(\vz'_1,\vz_2)}_2
+ \norm{\vf(\vz'_1,\vz_2) - \vf(\vz'_1,\vz'_2)}_2 \\
&\leq \norm{\vu}_2 \rho_{x,w} \sqrt{L} (\norm{\mF \mS^* (\vz_1-\vz'_1)}_\infty + \norm{\mF \mS^* (\vz_2-\vz'_2)}_\infty),
\end{align*}
it follows that
\begin{align*}
\int_0^\infty \sqrt{\log N(\Delta,\norm{\cdot}_2,t)} dt
&\leq C_1 \norm{\vu}_2 \rho_{x,w} \sqrt{L} \int_0^\infty \sqrt{\log N(\mF \mS^* B_2^K,\norm{\cdot}_\infty,t)} dt \\
&\leq C_1 \norm{\vu}_2 \rho_{x,w} \sqrt{KL} \int_0^\infty \sqrt{\log N(\mF \mS^* B_1^K,\norm{\cdot}_\infty,t)} dt \\
&\leq C_2 \norm{\vu}_2 \rho_{x,w} \sqrt{K} \sqrt{\log K} \log^{3/2} L,
\end{align*}
where the last step follows from Corollary~\ref{cor:maurey} together with the observation that $\|\mF\|_{\ell_1^L \rightarrow \ell_\infty^L} \leq L^{-1/2}$.  Then the assertion follows from Lemma~\ref{lemma:firstorder} with a union bound argument over $m=1,\dots,M$.

\subsection{Proof of Lemma~\ref{lemma:comp_cross}}
\label{sec:proof:lemma:comp_cross}

Let $\vphi_m = \mathrm{vec}(\mPhi_m)$ for $m=1,\dots,M$ and $\vphi = [\vphi_1^\transpose, \dots, \vphi_M^\transpose]^\transpose$.
Let $\vv_1,\dots,\vv_M \in \mathbb{C}^D$ and $\vv = [\vv_1^\transpose,\dots,\vv_M^\transpose] \in B_2^{MD}$.
Then the spectral norm of $\sum_{m=1}^M
\ve_m^* \otimes \breve{\mS} \mC_{\vw_m}^* \mC_{\vx} \mS^* \mPhi_m$ is rewritten as
\begin{align*}
\Big\|
\sum_{m=1}^M
\ve_m^* \otimes \breve{\mS} \mC_{\vw_m}^* \mC_{\vx} \mS^* \mPhi_m
\Big\|
& = \sup_{\vz \in B_2^{2K-1}} \sup_{\vv \in B_2^{MD}}
\Big|
\sum_{m=1}^M \vz^* \breve{\mS} \mC_{\vw_m}^* \mC_{\vx} \mS^* \mPhi_m \vv_m
\Big| \\
& = \sup_{\vz \in B_2^{2K-1}} \sup_{\vv \in B_2^{MD}}
\Big|
\sum_{m=1}^M (\vv_m^\transpose \otimes \vz^* \breve{\mS} \mC_{\vw_m}^* \mC_{\vx} \mS^*) \vphi_m
\Big| \\
& = \sup_{\vz \in B_2^{2K-1}} \sup_{\vv \in B_2^{MD}}
\Big|
\sum_{m=1}^M
\ve_m^* \otimes (\vv_m^\transpose \otimes \vz^* \breve{\mS} \mC_{\vw_m}^* \mC_{\vx} \mS^*) \vphi
\Big|.
\end{align*}

Let
\[
\vf(\vz,\vv)
= \sum_{m=1}^M
\ve_m \otimes (\overline{\vv_m} \otimes \mS \mC_{\vx}^* \mC_{\vw_m} \breve{\mS}^* \vz).
\]
Then we obtain
\[
\Big\|
\sum_{m=1}^M
\ve_m^* \otimes \breve{\mS} \mC_{\vw_m}^* \mC_{\vx} \mS^* \mPhi_m
\Big\|
= \sup_{\vv \in B_2^{MD}} \sup_{\vz \in B_2^{2K-1}} |\vf(\vz,\vv)^* \vphi|.
\]
Note that $\vf(\vz,\vv)^* \vphi$, conditioned on $\vw$, is a centered Gaussian process.  We compute a tail estimate of this supremum by applying Lemma~\ref{lemma:firstorder} with
\[
\Delta = \{ \vf(\vz,\vv) ~|~ \vz \in B_2^{2K-1},~ \vv \in B_2^{MD} \}.
\]
Then we need to compute the entropy integral for $\Delta$.
Recall
\[
\rho_{x,w} = \max_{1\leq m\leq M} \norm{\widetilde{\mS} \mC_{\vx}^* \mC_{\vw_1} \widetilde{\mS}^*}
\geq \norm{\breve{\mS} \mC_{\vw_m}^* \mC_{\vx} \mS^*}, \quad \forall m=1,\dots,M.
\]

By the triangle inequality, we obtain
\begin{align*}
\norm{\vf(\vz,\vv) - \vf(\vz',\vv')}_2
& \leq \norm{\vf(\vz,\vv) - \vf(\vz,\vv')}_2
+ \norm{\vf(\vz,\vv') - \vf(\vz',\vv')}_2 \\
& \leq \rho_{x,w}
(\norm{\vz}_2 \norm{\vv-\vv'}_2
+ \norm{\vz-\vz'}_2 \norm{\vv'}_2) \\
& \leq \rho_{x,w}
(\norm{\vv-\vv'}_2 + \norm{\vz-\vz'}_2).
\end{align*}

The integral of the log-entropy number is computed as
\begin{align*}
& \sup_{\vv \in B_2^{MD}} \sup_{\vz \in B_2^{2K-1}} |\vf(\vz,\vv)^* \vphi| \\
& \leq C_1 \int_0^\infty \sqrt{\log N(\Delta,\norm{\cdot}_2,t)} dt \\
& \leq C_1 \rho_{x,w}
\Big( \int_0^\infty \sqrt{\log N(B_2^{MD},\norm{\cdot}_2,t)} dt
+ \int_0^\infty \sqrt{\log N(B_2^{2K-1},\norm{\cdot}_2,t)} dt \Big) \\
& \leq C_2 \rho_{x,w} (\sqrt{MD} + \sqrt{K}),
\end{align*}
where the last step follows from a standard volume argument.
Then the assertion follows from Lemma~\ref{lemma:firstorder}.

\subsection{Proof of Lemma~\ref{lemma:rho_w}}
\label{sec:proof:lemma:rho_w}
By the homogeneity of $\rho_w$, we may assume that $\sigma_w = 1$.
We first consider the case that $m'=m$.
Let $\vz \in \mathbb{C}^K$. Then by the commutativity of convolution, we have
\[
\mC_{\vw_m} \mS^* \vz
= \vw_m \conv \mS^* \vz
= \mC_{\mS^* \vz} \vw_m.
\]
Therefore, the spectral norm of $\mS (\mC_{\vw_m}^* \mC_{\vw_m} - \mathbb{E}[\mC_{\vw_m}^* \mC_{\vw_m}]) \mS^*$ is rewritten as
\begin{align*}
\norm{\mS (\mC_{\vw_m}^* \mC_{\vw_m} - \mathbb{E}[\mC_{\vw_m}^* \mC_{\vw_m}]) \mS^*}
&= \sup_{\vz \in B_2^K} | \vz^* \mS (\mC_{\vw_m}^* \mC_{\vw_m} - \mathbb{E}[\mC_{\vw_m}^* \mC_{\vw_m}]) \mS^* \vz | \\
&= \sup_{\vz \in B_2^K}
| \vw_m^* \mC_{\mS^* \vz}^* \mC_{\mS^* \vz} \vw_m - \mathbb{E}[\vw_m^* \mC_{\mS^* \vz}^* \mC_{\mS^* \vz} \vw_m] |,
\end{align*}
where the last term is the supremum of a second order chaos.
We use Theorem~\ref{thm:kmr} to get its tail estimate.
Define
\[
\Delta_{\vz} = \{ \mC_{\mS^* \vz} ~|~ \vz \in B_2^K \}.
\]
Then the radii of $\Delta_{\vz}$ with respect to the spectral and Frobenius norms are given by
\begin{align*}
d_{\mathrm{S}}(\Delta_{\vz})
& = \sup_{\vz \in B_2^K} \norm{\mC_{\mS^* \vz}}
= \sup_{\vz \in B_2^K} \sqrt{L} \norm{\mF \mS^* \vz}_\infty
= \sup_{\vz \in B_2^K} \norm{\vz}_1
= \sqrt{K} \\
d_{\mathrm{F}}(\Delta_{\vz}) &= \sup_{\vz \in B_2^K} \norm{\mC_{\mS^* \vz}}_{\mathrm{F}}
= \sup_{\vz \in B_2^K} \sqrt{L} \norm{\vz}_2 = \sqrt{L}.
\end{align*}
Moreover, the $\gamma_2$ functional of $\Delta_{\vz}$ is bounded by
\begin{align*}
\gamma_2(\Delta_{\vz},\norm{\cdot})
& \leq C_1 \sqrt{L} \int_0^\infty \sqrt{\log N(\mF \mS^* B_2^K, \norm{\cdot}_\infty, t)} dt \\
& \leq C_1 \sqrt{KL} \int_0^\infty \sqrt{\log N(\mF \mS^* B_1^K, \norm{\cdot}_\infty, t)} dt \\
& \leq C_1 \sqrt{K} \sqrt{\log K} \log^{3/2} L,
\end{align*}
where the last step follows from Lemma~\ref{lemma:maurey}.
By Theorem~\ref{thm:kmr},
\[
\norm{\mS (\mC_{\vw_m}^* \mC_{\vw_m} - \mathbb{E}[\mC_{\vw_m}^* \mC_{\vw_m}]) \mS^*}
\leq C \sigma_w^2 K (1 + \log K) \log^3 L (1 + 2 \log M + \beta \log K)
\]
holds with probability $1-M^{-2} K^{-\beta}$.

Next we consider the case where $m' \neq m$.
In this case, we have $\mathbb{E}[\mC_{\vw_m}^* \mC_{\vw_{m'}}] = \bm{0}_{L,L}$.
Similarly to the previous case, the spectral norm of $\mS \mC_{\vw_m}^* \mC_{\vw_{m'}} \mS^*$ is rewritten as
\begin{align*}
\norm{\mS \mC_{\vw_m}^* \mC_{\vw_{m'}}  \mS^*}
= \sup_{\vz, \vq \in B_2^K}
| \vw_m^* \mC_{\mS^* \vq}^* \mC_{\mS^* \vz} \vw_{m'} |
= \sup_{\vz, \vq \in B_2^K}
|\widetilde{\vw}_{m,m'}^* \mL_{\vq}^* \mR_{\vz} \widetilde{\vw}_{m,m'}|,
\end{align*}
where
\[
\widetilde{\vw}_{m,m'} =
\begin{bmatrix}
\vw_m \\
\vw_{m'}
\end{bmatrix}, \quad
\mL_{\vq} =
\begin{bmatrix}
\bm{0}_{L,L} & \bm{0}_{L,L} \\
\mC_{\mS^* \vq} & \bm{0}_{L,L}
\end{bmatrix}, \quad
\mR_{\vz} =
\begin{bmatrix}
\bm{0}_{L,L} & \mC_{\mS^* \vz} \\
\bm{0}_{L,L} & \bm{0}_{L,L}
\end{bmatrix}.
\]
Define
\[
\Delta_{\mL_{\vq}} = \{ \mL_{\vq} ~|~ \vq \in B_2^K \}
\]
and
\[
\Delta_{\mR_{\vz}} = \{ \mR_{\vz} ~|~ \vz \in B_2^K \}.
\]
Then, the radii and $\gamma_2$ functional of $\Delta_{\mL_{\vq}}$ and $\Delta_{\mR_{\vz}}$ are identical to those of $\Delta_{\vz}$.
Therefore, by Theorem~\ref{thm:ip},
\[
\norm{\mS (\mC_{\vw_m}^* \mC_{\vw_{m'}} - \mathbb{E}[\mC_{\vw_m}^* \mC_{\vw_{m'}}]) \mS^*}
\leq C \sigma_w^2 K (1 + \log K) \log^3 L (1 + 2 \log M + \beta \log K)
\]
holds with probability $1-M^{-2} K^{-\beta}$.
The assertion follows by applying a union bound argument.

\subsection{Proof of Lemma~\ref{lemma:barrho_w}}
\label{sec:proof:lemma:barrho_w}

Without loss of generality, we assume that $\sigma_w = 1$.
Similarly to proof of Lemma~\ref{lemma:rho_w} in Appendix~\ref{sec:proof:lemma:rho_w},
we can rewrite $M \bar{\rho}_w$ as the supremum of a second order chaos process as follow:
\begin{align*}
M \bar{\rho}_w
&= \sup_{\vz \in B_2^K}
\Big|
\sum_{m=1}^M \vw_m^* \mC_{\mS^* \vz}^* \mC_{\mS^* \vz} \vw_m - \mathbb{E}[\vw_m^* \mC_{\mS^* \vz}^* \mC_{\mS^* \vz} \vw_m]
\Big| \\
&= \sup_{\vz \in B_2^K}
\Big|
\vw^* (\mId_M \otimes \mC_{\mS^* \vz}^* \mC_{\mS^* \vz}) \vw
- \mathbb{E}[\vw^* (\mId_M \otimes \mC_{\mS^* \vz}^* \mC_{\mS^* \vz}) \vw]
\Big|,
\end{align*}
where $\vw = [\vw_1^\transpose,\dots,\vw_M^\transpose]^\transpose$ is a standard complex Gaussian vector of length $ML$.

Define
\[
\widetilde{\Delta}_{\vz} = \{ \mId_M \otimes \mC_{\mS^* \vz} ~|~ \vz \in B_2^K \}.
\]
Then the radii of $\widetilde{\Delta}_{\vz}$ with respect to the spectral and Frobenius norms are upper bounded respectively by
\begin{align*}
d_{\mathrm{S}}(\widetilde{\Delta}_{\vz})
&= \sup_{\vz \in B_2^K} \sqrt{L} \norm{\mF \mS^* \vz}_\infty \leq \sqrt{K} \\
d_{\mathrm{F}}(\widetilde{\Delta}_{\vz})
&= \sup_{\vz \in B_2^K} \sqrt{M} \norm{\mC_{\mS^* \vz}}_{\mathrm{F}}
\leq \sqrt{ML}.
\end{align*}
Moreover, since $\norm{\mId_M \otimes \mC_{\mS^* \vz}} = \norm{\mC_{\mS^* \vz}}$, the $\gamma_2$ functional of $\widetilde{\Delta}_{\vz}$ is upper bounded by
\begin{align*}
\gamma_2(\widetilde{\Delta}_{\vz},\norm{\cdot})
= \gamma_2(\Delta_{\vz},\norm{\cdot})
\leq C_1 \sqrt{K} \sqrt{\log K} \log^{3/2} L,
\end{align*}
where the last step has been shown in Appendix~\ref{sec:proof:lemma:rho_w}.
By Theorem~\ref{thm:kmr},
\[
\bar{\rho}_w
\leq C \beta \sigma_w^2 M^{-1/2} \sqrt{KL} (1 + \log K) \log^3 L \log K
\]
holds with probability $1-K^{-\beta}$.

\subsection{Proof of Lemma~\ref{lemma:noise_part1}}
\label{sec:proof:lemma:noise_part1}

First we rewrite $\mPhi \vv$ as
\[
\mPhi \vv
= \mW_{\vv} \vphi,
\]
where
\[
\vphi = [\mathrm{vec}(\mPhi_1)^\transpose, \dots, \mathrm{vec}(\mPhi_1)^\transpose]^\transpose
\]
and
\[
\mW_{\vv} = \sum_{m=1}^M \ve_m \ve_m^* \otimes \vv_m^\transpose \otimes \mId_K.
\]

Then it follows that
\[
\vv^* \mPhi^*
(\mY_{\mathrm{n}}^* \mY_{\mathrm{n}} - \mLambda) \mPhi \vv
= \vphi^* \mW_{\vv}^*
(\mY_{\mathrm{n}}^* \mY_{\mathrm{n}} - \mLambda)
\mW_{\vv} \vphi,
\]
where the latter, conditional on $\vw$, is a quadratic Gaussian form.
Furthermore, by direct calculation, we have
\[
\mathbb{E}_{\vphi}[\vphi^* \mW_{\vv}^*(\mY_{\mathrm{n}}^* \mY_{\mathrm{n}} - \mLambda)\mW_{\vv} \vphi]
= \mathbb{E}_{\vphi}[\mPhi^* (\mY_{\mathrm{n}}^* \mY_{\mathrm{n}} - \mLambda) \mPhi]
= \mathrm{tr}(\mY_{\mathrm{n}}^* \mY_{\mathrm{n}} - \mLambda) \mId_{MD}
= \bm{0}_{MD,MD}.
\]
Then $\norm{\mPhi^* (\mY_{\mathrm{n}}^* \mY_{\mathrm{n}} - \mLambda) \mPhi}$ is written as
\begin{align*}
\sup_{\vv \in B_2^{MD}}
|
\vphi^* \mW_{\vv}^*
(\mY_{\mathrm{n}}^* \mY_{\mathrm{n}} - \mLambda)
\mW_{\vv} \vphi
- \mathbb{E}_{\vphi}[\vphi^* \mW_{\vv}^*(\mY_{\mathrm{n}}^* \mY_{\mathrm{n}} - \mLambda)\mW_{\vv} \vphi]
|,
\end{align*}
which is the supremum of a second order Gaussian chaos process.
We compute its tail estimate by applying Theorem~\ref{thm:ip} with
\[
\Delta_1 = \{ \mW_{\vv} ~|~ \vv \in B_2^{MD} \}
\]
and
\[
\Delta_2 = \{ (\mY_{\mathrm{n}}^* \mY_{\mathrm{n}} - \mLambda) \mW_{\vv} ~|~ \vv \in B_2^{MD} \}.
\]

Let $\norm{\cdot}_{p,q}$ be defined in \eqref{eq:defblknorm}.  Then the radii of $\Delta_1$ with respect to the Frobenius and spectral norms are upper bounded respectively by
\begin{align*}
d_{\mathrm{S}}(\Delta_1)
&\leq \sup_{\vv \in B_2^{MD}} \norm{\vv}_{2,\infty} \leq 1 \\
d_{\mathrm{F}}(\Delta_1)
&\leq \sup_{\vv \in B_2^{MD}} \sqrt{K} \norm{\vv}_2 \leq \sqrt{K}.
\end{align*}
By Lemma~\ref{lemma:entropyv2inf}, the $\gamma_2$ functional is bounded by
\begin{align*}
\gamma_2(\Delta_1,\norm{\cdot})
\leq C_1 \int_0^\infty \sqrt{\log N(B_2^{MD}, \norm{\cdot}_{2,\infty}, t)} dt
\leq C_2 \sqrt{D} \sqrt{\log D} \log(MD).
\end{align*}

We repeat the calculation for $\Delta_2$.
Note that $(\mY_{\mathrm{n}}^* \mY_{\mathrm{n}} - \mLambda) \mW_{\vv}$ is expressed as
\begin{align*}
(\mY_{\mathrm{n}}^* \mY_{\mathrm{n}} - \mLambda) \mW_{\vv}
&=
\sum_{m=1}^M \ve_m \ve_m^* \otimes \vv_m^\transpose \otimes
\sum_{\begin{subarray}{c} m'=1 \\ m' \neq m \end{subarray}}^M
\mS (\mC_{\vw_{m'}} \mC_{\vw_{m'}}^* - \norm{\vw_{m'}}_2^2 \mId_L) \mS^* \\
&- \sum_{m=1}^M \sum_{\begin{subarray}{c} m'=1 \\ m' \neq m \end{subarray}}^M
\ve_m \ve_{m'}^* \otimes \vv_{m'}^\transpose \otimes
\mS \mC_{\vw_m} \mC_{\vw_{m'}}^* \mS^*.
\end{align*}

Noting that all the summands in this decomposition are $K\times K$ matrices and orthogonal with respect to the Frobenius inner product, we obtain that the Frobenius norm of $(\mY_{\mathrm{n}}^* \mY_{\mathrm{n}} - \mLambda) \mW_{\vv}$ is upper bounded by
\begin{align*}
\norm{(\mY_{\mathrm{n}}^* \mY_{\mathrm{n}} - \mLambda) \mW_{\vv}}_{\mathrm{F}}^2
&\leq \sum_{m=1}^M \norm{\vv_m}_2^2 K
\Big\|
\sum_{\begin{subarray}{c} m'=1 \\ m' \neq m \end{subarray}}^M
\mS (\mC_{\vw_{m'}} \mC_{\vw_{m'}}^* - \norm{\vw_{m'}}_2^2 \mId_L) \mS^*
\Big\|^2 \\
&+ \sum_{m=1}^M \sum_{\begin{subarray}{c} m'=1 \\ m' \neq m \end{subarray}}^M
\norm{\vv_{m'}}_2^2 K \norm{\mS \mC_{\vw_m} \mC_{\vw_{m'}}^* \mS^*}^2 \\
&\leq 2 \norm{\vv}_2^2 K (M^2 \bar{\rho}_w^2 + \rho_w^2) +
\norm{\vv}_2^2 K M \rho_w^2,
\end{align*}
which implies that
\[
\norm{(\mY_{\mathrm{n}}^* \mY_{\mathrm{n}} - \mLambda) \mW_{\vv}}_{\mathrm{F}}
\leq C \norm{\vv}_2 \sqrt{K} (M \bar{\rho}_w + \sqrt{M} \rho_w).
\]

On the other hand, it follows from the block Gershgorin disk theorem \cite{feingold1962block} that
\begin{align*}
&\norm{(\mY_{\mathrm{n}}^* \mY_{\mathrm{n}} - \mLambda) \mW_{\vv}} \\
&\leq
\max_{1\leq m \leq M} \Big(
\norm{\vv_m}_2 \Big\|
\sum_{\begin{subarray}{c} m'=1 \\ m' \neq m \end{subarray}}^M
\mS (\mC_{\vw_{m'}} \mC_{\vw_{m'}}^* - \norm{\vw_{m'}}_2^2 \mId_L) \mS^* \Big\|
+ \sum_{\begin{subarray}{c} m'=1 \\ m' \neq m \end{subarray}}^M
\norm{\vv_{m'}}_2 \norm{\mS \mC_{\vw_m} \mC_{\vw_{m'}}^* \mS^*} \Big) \\
&\leq \norm{\vv}_{2,\infty} (M \bar{\rho}_w + \rho_w) + \norm{\vv}_{2,1} \rho_w
\leq \norm{\vv}_{2,\infty} M (\bar{\rho}_w + 2 \rho_w).
\end{align*}

Therefore, the radii of $\Delta_2$ with respect to the Frobenius and spectral norms are upper bounded respectively by
\begin{align*}
d_{\mathrm{S}}(\Delta_2)
&\leq M (\bar{\rho}_w + 2 \rho_w), \\
d_{\mathrm{F}}(\Delta_2)
&\leq C \sqrt{K} (M \bar{\rho}_w + \sqrt{M} \rho_w).
\end{align*}
Moreover, by Lemma~\ref{lemma:entropyv2inf}, the $\gamma_2$ functional is bounded by
\begin{align*}
\gamma_2(\Delta_2,\norm{\cdot})
&\leq C_1 M (\bar{\rho}_w + 2 \rho_w) \int_0^\infty \sqrt{\log N(B_2^{MD}, \norm{\cdot}_{2,\infty}, t)} dt \\
&\leq C_2 M \sqrt{D} (\bar{\rho}_w + 2 \rho_w) \sqrt{\log D} \log(MD).
\end{align*}

Consequently, Theorem~\ref{thm:ip} yields that
\[
\norm{\mPhi^* (\mY_{\mathrm{n}}^* \mY_{\mathrm{n}} - \mLambda) \mPhi}
\leq C(\beta) M \sqrt{KD} (\bar{\rho}_w + 2 \rho_w) \log D \log^2(MD)
\]
holds with probability at least  $1-K^{-\beta}$.

\subsection{Proof of Lemma~\ref{lemma:noise_part1_fixu}}
\label{sec:proof:lemma:noise_part1_fixu}

We modify the proof of Lemma~\ref{lemma:noise_part1} in Appendix~\ref{sec:proof:lemma:noise_part1} as follows.

First note that
\begin{align*}
\frac{
\norm{
\mPhi^* (\mY_{\mathrm{n}}^* \mY_{\mathrm{n}} - \mLambda) \mPhi \vu}_2}{\norm{\vu}_2}
= \sup_{\vv \in B_2^{MD}}
\frac{|
\vphi^* \mW_{\vv}^*
\mQ_w^* \mQ_w
\mW_{\vu} \vphi|}{\norm{\vu}_2}.
\end{align*}

We only need to replace $\breve{\Delta}_{\mathrm{R}}$ by the following singleton set
\[
\breve{\Delta}_{\mathrm{R},\vu} = \{ (\mY_{\mathrm{n}}^* \mY_{\mathrm{n}} - \mLambda) \mW_{\vu} \}.
\]
Indeed, the radii of $\breve{\Delta}_{\mathrm{R},\vu}$ and $\breve{\Delta}_{\mathrm{R}}$ are the same for both the Frobenius and spectral norms.
However, the $\gamma_2$ functional of $\breve{\Delta}_{\mathrm{R},\vu}$ is 0.
The assertion follows by applying the modified estimates to Theorem~\ref{thm:ip}.

\section{Entropy Estimates}
\label{sec:entropy_blocknorm}

\begin{lemma}
\label{lemma:entropyv2inf}
\begin{equation}
\label{eq:entintegral}
\int_0^\infty \sqrt{\log N(B_2^{MD}, \norm{\cdot}_{2,\infty}, t)} dt
\leq C \sqrt{D} \sqrt{\log D} \log(MD).
\end{equation}
\end{lemma}

\begin{proof}[Proof of Lemma~\ref{lemma:entropyv2inf}]
Let us recall that the $(2,\infty)$-block norm of $\vv \in \mathbb{C}^{MD}$ is defined by
\[
\norm{\vv}_{2,\infty} = \max_{m \in [M]} \norm{\vv_m}_2,
\]
where $\vv_k \in \mathbb{C}^D$ for $k=1,\dots,M$ denotes the blocks of $\vv$ such that $\vv = [\vv_1^\transpose,\dots,\vv_M^\transpose]$.

Since
\[
\norm{\vv}_{2,\infty} \leq \norm{\vv}_2,
\]
the interval in the integral in \eqref{eq:entintegral} can be restricted to the unit interval $[0,1)$.

Indeed, the $(2,\infty)$-block norm of $\vv$ is rewritten as
\[
\norm{\vv}_{2,\infty} = \max_{m \in [M]} \max_{\vzeta \in B_2^D} \langle \vv_m, \vzeta \rangle.
\]

To compute an estimate of the entropy integral in \eqref{eq:entintegral}, we adopt the strategy \cite{junge2017ripI} that estimates a unit ball using a polytope.  The original strategy \cite{junge2017ripI} was developed for the RIP analysis for low-rank tensors and applied to the tensor nuclear norm.  The same strategy applies to the block norm in this section too.

\begin{lemma}
\label{lemma:polytopeapprox}
There exist $\vzeta_1,\dots,\vzeta_N \in \mathbb{S}^{D-1}$ such that
\[
B_2^D \subset 2 ~ \mathrm{conv} \{\vzeta_1,\dots,\vzeta_N\}
\]
and $\log N \leq (2D+1) \log(4D+3)$.
\end{lemma}

\begin{proof}[Proof of Lemma~\ref{lemma:polytopeapprox}]
Let $\{\vzeta_1,\dots,\vzeta_N\} \subset \mathbb{S}^{D-1}$ be an $\epsilon$-net of $\mathbb{S}^{D-1}$. Then by the standard volume argument, we have $N \leq (1+2/\epsilon)^{2D}$.
Furthermore, it follows that
\[
\mathbb{S}^{D-1} \subset \frac{1}{1-\epsilon} \cdot \mathrm{abs} \mathrm{conv} \{\vzeta_1,\dots,\vzeta_N\}.
\]
Indeed, for any $\vw \in \mathbb{S}^{D-1}$, we construct a sequence $(\widehat{\vw}_k,\alpha_k)_{k \in \mathbb{N}} \subset \{\vzeta_1,\dots,\vzeta_N\} \times \mathbb{C}$ as follows. Let $\alpha_1 = 1$ and $\widehat{\vw}_1$ be the closest vector to $\vw$ in $\{\vzeta_1,\dots,\vzeta_N\}$.
If $\vw = \alpha_1 \widehat{\vw}_1$, then $\alpha_k = 0$ and $\widehat{\vw}_k = \bm{0}$ for all $k \geq 2$. Otherwise, let $\alpha_2 = \norm{\vw-\alpha_1\widehat{\vw}_1}_2$ and $\widehat{\vw}_2$ be the closest vector to $\alpha_2^{-1}(\vw-\alpha_1\widehat{\vw}_1)$ in $\{\vzeta_1,\dots,\vzeta_N\}$.
If $\vw = \alpha_1 \widehat{\vw}_1 + \alpha_2 \widehat{\vw}_2$, then $\alpha_k = 0$ and $\widehat{\vw}_k = \bm{0}$ for all $k \geq 3$. Otherwise, let $\alpha_3 = \norm{\vw-\alpha_1\widehat{\vw}_1-\alpha_2\widehat{\vw}_2}_2$ and $\widehat{\vw}_3$ be the closest vector to $\alpha_3^{-1}(\vw-\alpha_1\widehat{\vw}_2-\alpha_1\widehat{\vw}_2)$ in $\{\vzeta_1,\dots,\vzeta_N\}$.
By continuing in this way, we have
\[
\vw = \sum_{k \in \mathbb{N}} \alpha_k \widehat{\vw}_k,
\]
where $|\alpha_k|\leq \epsilon^{k-1}$ and $\norm{\widehat{\vw}_k}_2 = 1$ for all $k \in \mathbb{N}$. Therefore,
\[
\sum_{k \in \mathbb{N}} |\alpha_k| \leq \frac{1}{1-\epsilon}
\]
and the assertion follows.
By including $\pm \vzeta_k$ instead of $\vzeta_k$, we can replace the absolute convex hull by convex hull and the cardinality increases only by factor $2$.
Choosing $\epsilon = 1/(2D+1)$ completes the proof.
\end{proof}

By Lemma~\ref{lemma:polytopeapprox}, we approximate the $(2,\infty)$-block norm of $\vv$ as a polytope norm as follows:
\begin{equation}
\label{eq:ubnorm}
\begin{aligned}
\norm{\vv}_{2,\infty}
& = \max_{\vzeta \in \mathbb{S}^{D-1}} \max_{m \in [M]} \langle \vv_m, \vzeta \rangle \\
& \leq 2 \max_{\vzeta \in \mathrm{conv} \{\vzeta_1,\dots,\vzeta_N\}} \max_{m \in [M]} \langle \vv_m, \vzeta \rangle \\
& = 2 \max_{n \in [N]} \max_{m \in [M]} |\langle \vv_m, \vzeta_n \rangle| \\
& =: 2 \tnorm{\vv},
\end{aligned}
\end{equation}
where $\log N \leq (D+1) \log(2D+3)$.

Define
\[
S_{\zeta_n} = \Big\{ \vv = [\vv_1^\transpose,\dots,\vv_M^\transpose]^\transpose ~\Big|~ \max_{m \in [M]} |\langle \vv_m, \vzeta_n \rangle| \leq 1 \Big\}.
\]
Then its polar set is given by
\[
S_{\zeta_n}^\circ
= \{ \vz \otimes \vzeta_n | \vz \in B_1^M \}
= \mathrm{conv} \{ \ve_d \otimes \vzeta_n | d \in [D] \}.
\]

Note that the unit ball with respect to $\tnorm{\cdot}$ is given as $\bigcap_{n \in [N]} S_{\zeta_n}$. To compute the unit ball with respect to the dual norm, we will use a well known polar duality result in the following lemma. Note that $S_{\zeta_n}$ is not bounded. As the lemma is typically stated for bounded sets, we provide the proof for completeness, verifying that boundedness is not a crucial assumption.
\begin{lemma}
\label{lemma:polardual}
Let $A$ and $B$ be convex sets.
The polar set of the intersection of $A$ and $B$ is given by
\[
(A \cap B)^\circ = \mathrm{conv} (A^\circ \cup B^\circ).
\]
\end{lemma}

\begin{proof}[Proof of Lemma~\ref{lemma:polardual}]
We first show
\begin{equation}
\label{eq:polarinc1}
(\mathrm{conv}(A^\circ \cup B^\circ))^\circ \subset A \cap B.
\end{equation}

Suppose that $x \not\in A \cap B$.
Since $A$ and $B$ are convex, we have $A = A^{\circ\circ}$ and $B = B^{\circ\circ}$.
Without loss of generality, we may assume that $x \not\in A^{\circ\circ}$.
Then there exists $w \in A^\circ$ such that
\begin{equation}
\label{eq:notinpolar}
|\langle x, w\rangle| > 1.
\end{equation}
Since $A^\circ \subset A^\circ \cup B^\circ \subset \mathrm{conv}(A^\circ \cup B^\circ)$, $w$ also satisfies
\begin{equation}
\label{eq:wmembership}
w \in \mathrm{conv}(A^\circ \cup B^\circ).
\end{equation}
The existence of $w$ satisfying both \eqref{eq:notinpolar} and \eqref{eq:wmembership} implies that
\[
x \not\in (\mathrm{conv}(A^\circ \cup B^\circ))^\circ.
\]
Then \eqref{eq:polarinc1} follows by contraposition.

Next, we show the other inclusion, which is
\begin{equation}
\label{eq:polarinc2}
A \cap B \subset (\mathrm{conv}(A^\circ \cup B^\circ))^\circ.
\end{equation}

Suppose that $x \in A \cap B = A^{\circ\circ} \cap B^{\circ\circ}$.
Then for all $w_A \in A^\circ$, $w_B \in B^\circ$, and $t \in [0,1]$, it follows that
\[
|\langle x, t w_A+(1-t)w_B\rangle|
\leq t |\langle w, w_A\rangle| + (1-t) |\langle w, w_B\rangle| \leq 1.
\]
Therefore
\[
x \in \mathrm{conv}(A^\circ \cup B^\circ)^\circ.
\]

We have shown that
\[
A \cap B = \mathrm{conv}(A^\circ \cup B^\circ)^\circ.
\]
The assertion follows from the definition of polar sets.
\end{proof}

By the polar duality in Lemma~\ref{lemma:polardual}, the unit ball with respect to the dual of $\tnorm{\cdot}$ is given as
\begin{align*}
\mathrm{conv} \Big( \bigcup_{n \in [N]} S_{\zeta_n}^\circ \Big)
& = \mathrm{conv} \Big( \bigcup_{n \in [N]} \mathrm{conv} \{ \ve_d \otimes \vzeta_n | d \in [D] \} \Big) \\
& = \mathrm{conv} \{ \ve_d \otimes \vzeta_n | n \in [N], ~ d \in [D] \}.
\end{align*}

Collecting the above estimates, we bound the log entropy number in \eqref{eq:entintegral} as follows:
\begin{align*}
\log N(B_2^{MD}, \norm{\cdot}_{2\infty}, t)
& \lesssim \log N(B_2^{MD}, \tnorm{\cdot}, t/2) \\
& \lesssim \log N\Big(\mathrm{conv} \bigcup_{n \in [N]} S_{\zeta_n}^\circ, \norm{\cdot}_2, t/2\Big) \\
& \lesssim \log N\Big(\mathrm{conv} \{ \ve_d \otimes \vzeta_n | n \in [N], ~ d \in [D] \}, \norm{\cdot}_2, t/2\Big),
\end{align*}
where the first inequality follows from \eqref{eq:ubnorm} and the second inequality holds by the entropy duality by Artstein et al. \cite{artstein2004duality}.

Next we define a linear map $Q: \ell_1^{ND} \to \ell_2^{MD}$ so that the standard basis vectors in $\ell_1^{ND}$ are mapped to distinct elements in $\{ \ve_d \otimes \vzeta_n ~|~ n \in [N], ~ d \in [D] \}$. In this construction, we only care the one-to-one correspondence and we do not care how we enumerate the elements of $\{ \ve_d \otimes \vzeta_n ~|~ n \in [N], ~ d \in [D] \}$. Although $Q$ is not uniquely determined and there is ambiguity up to a permutation in $\ell_1^{ND}$, every map $Q$ constructed as above satisfies that $\norm{Q:\ell_1^{ND} \to \ell_2^{MD}} = 1$. Fix $Q$ and we get
\begin{align*}
\int_0^1 \sqrt{\log N(B_2^{MD}, \norm{\cdot}_{2,\infty}, t)} dt
& \lesssim \int_0^1 \sqrt{\log N\Big(Q(B_1^{ND}), \norm{\cdot}_2, t/2\Big)} dt \\
& \lesssim \sqrt{\log(ND)} \log(MD)
\lesssim \sqrt{D} \sqrt{\log D} \log(MD),
\end{align*}
where the second inequality follows from Corollary~\ref{cor:maurey}.
\end{proof}



\end{document}